\documentclass[11pt,oneside,letterpaper]{article}
\usepackage{eurosym}
\usepackage{graphicx,amsmath,amsfonts}
\usepackage{setspace}
\usepackage{amssymb}
\usepackage{dsfont}
\usepackage{amsmath}
\usepackage{amsfonts}
\usepackage{ifthen}
\usepackage{mathtools}
\usepackage{graphicx}
\usepackage{enumerate}
\usepackage{color}
\usepackage{framed}
\usepackage[margin=1in]{geometry}
\usepackage{natbib}
\usepackage[dvipsnames]{xcolor}
\usepackage[colorlinks=true, urlcolor=Mahogany, linkcolor=Mahogany, citecolor=Mahogany]{hyperref}
\usepackage{fp}
\usepackage[capitalize]{cleveref}
\usepackage{amsthm}
\usepackage{caption}
\usepackage{subcaption}
\usepackage{palatino}
\usepackage{cuted,balance}
\usepackage{amssymb}
\usepackage{soul}
\usepackage{graphicx,amsmath,amsfonts}
\usepackage{tikz}
\usetikzlibrary{positioning}

\DeclareMathOperator*{\argmax}{arg\,max}
\newtheorem{assumption}{Assumption}

\newtheorem{corollary}{Corollary}
\newtheorem{definition}{Definition}
\newtheorem{proposition}{Proposition}
\newtheorem{theorem}{Theorem}
\newtheorem{lemma}{Lemma}
\newtheorem{claim}{Claim}

\newtheorem{example}{Example}
\newtheorem{fact}{Fact}

\allowdisplaybreaks

\usepackage{bm}
\usepackage{bbm}
\let\originaleqref\eqref
\renewcommand{\eqref}{\originaleqref}
\begin{document}

\title{The Limits of Price Discrimination Under Privacy Constraints}
\author{Alireza Fallah \thanks{University of California, Berkeley. afallah@berkeley.edu} \and 
Michael I. Jordan \thanks{University of California, Berkeley. jordan@cs.berkeley.edu} \and Ali Makhdoumi \thanks{Fuqua School of Business, Duke University. ali.makhdoumi@duke.edu} \and Azarakhsh Malekian \thanks{Rotman School of Management, University of Toronto. azarakhsh.malekian@rotman.utoronto.ca}}
\date{June 2024}
\maketitle

\sloppy

\begin{abstract}
We consider a producer's problem of selling a product to a continuum of privacy-conscious consumers, where the producer can implement third-degree price discrimination, offering different prices to different market segments. {Our privacy mechanism provides a degree of protection by probabilistically masking each market segment.  We establish that the resultant set of all consumer-producer utilities forms a convex polygon, characterized explicitly as a linear mapping of a certain high-dimensional convex polytope into $\mathbb{R}^2$.} This characterization enables us to investigate the impact of the privacy mechanism on both producer and consumer utilities. In particular, we establish that the privacy constraint always hurts the producer by reducing both the maximum and minimum utility achievable. From the consumer's perspective, although the privacy mechanism ensures an increase in the minimum utility compared to the non-private scenario, interestingly, it may reduce the maximum utility. Finally, we demonstrate that increasing the privacy level does not necessarily intensify these effects. For instance, the maximum utility for the producer or the minimum utility for the consumer may exhibit nonmonotonic behavior in response to an increase of the privacy level.
\end{abstract}

\section{Introduction}
Price discrimination has long been a topic of considerable interest in economics and computer science. At its core, price discrimination involves a seller's ability to charge different prices to different consumers for the same product or service. This strategy, underpinned by the seller's understanding of consumer preferences, can lead to maximized profits and a more efficient allocation of resources. While economically rational from a business perspective, price discrimination raises complex questions regarding fairness, market power, and consumer welfare.

Simultaneously, consumer privacy has become a growing concern in the digital age. With the increasing ability of companies to collect, analyze, and utilize vast amounts of personal data, the protection of consumer privacy has become a major issue for consumers, companies, and regulators. Privacy concerns are not just limited to the safeguarding of personal information but also encompass how this information is used in market practices, including pricing strategies. 

Thus, the intersection of price discrimination and consumer privacy is a timely and intriguing topic. One might initially assume that these two concepts are at odds. Price discrimination, by its very nature, relies on detailed information about consumers to tailor prices effectively. In contrast, an emphasis on consumer privacy could limit the ability of sellers to implement price discrimination strategies. However, consumers can benefit from price discrimination, and thus, this intersection warrants a more nuanced examination. Understanding how privacy constraints affect the ability of firms to engage in price discrimination helps in understanding existing market dynamics and also aids in shaping policies that protect consumer interests.

Our study of the intersection of price discrimination and privacy builds on prior work of \cite{bergemann2015limits}. Let us first describe their model and then explain how things change if we consider consumers' privacy concerns. Consider a monopolist who is selling a product to a continuum of consumers. If the producer knows the values of consumers, then they could engage in first-degree price discrimination to maximize their utility, while the utility of consumers would be reduced to zero. On the other hand, if the producer has no information besides a prior distribution about the value of consumers, then they can only enact uniform pricing. This leads to two pairs of producer utilities and consumer utilities. \cite{bergemann2015limits} delineate the set of all possible pairs by using different \emph{segmentations}. Let us clarify this approach with an example. Imagine a producer aiming to set a product's price for its customers across the United States. In this context, the \textit{aggregated market} is defined as the entire US market, meaning the distribution of consumers’ valuations nationwide. Each zip code within the US represents its own market, which is basically the distribution of the consumers’ values in that zip code.  Segmentation, in this case, refers to the division of the aggregated US market into these specific markets. Generally speaking, a segmentation involves breaking down the aggregated market into smaller markets, under the condition that these smaller markets must add up to the aggregated market. \cite{bergemann2015limits}  establish that the set of all possible pairs of consumer and producer utilities is a triangle as depicted in Figure \ref{fig:S_bergemann}. The point $Q$ corresponds to the first-degree price discrimination, and the line $TR$ represents all possible consumer utilities when the producer utility is at its minimum, equating to the utility derived from the uniform pricing.

\begin{figure}
  \centering
  \includegraphics[width=.5\linewidth]{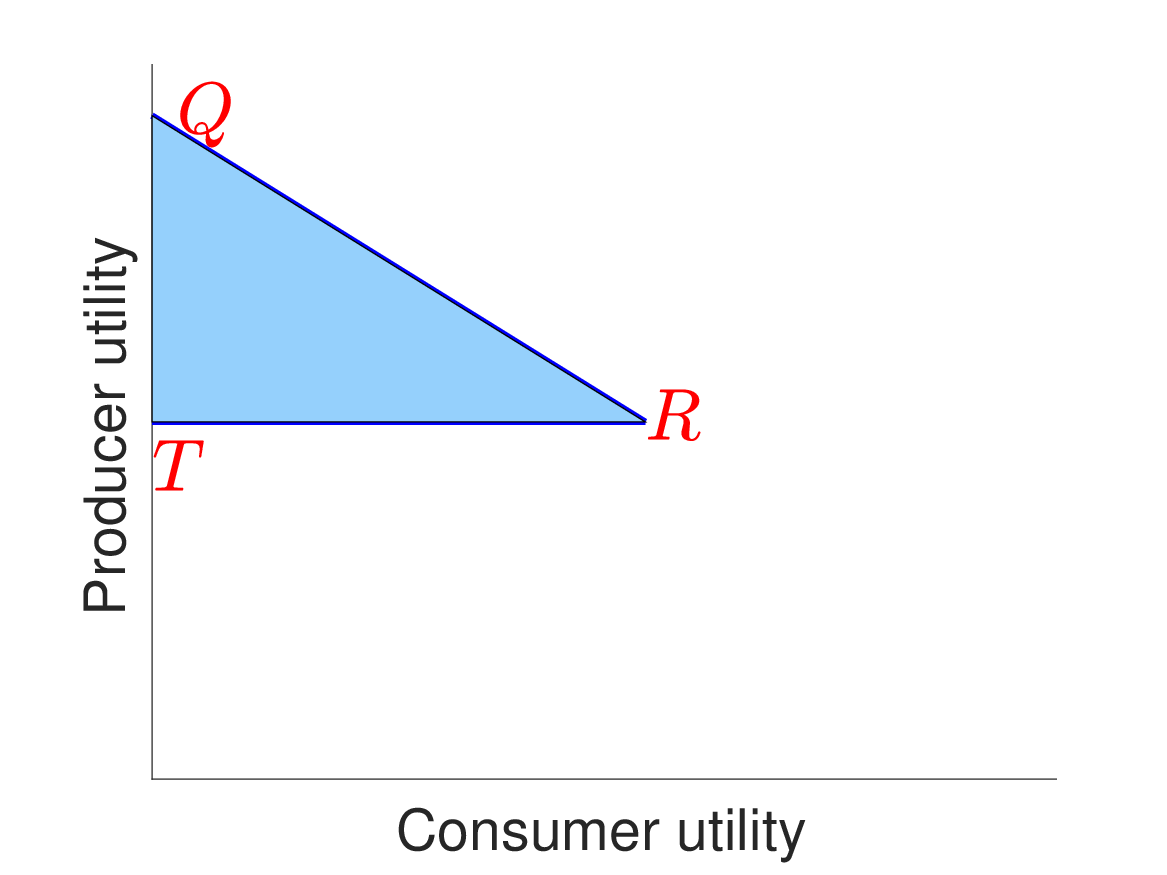}
  \caption{The surplus triangle of \cite{bergemann2015limits}}
  \label{fig:S_bergemann}
\end{figure}

For a given market, knowing the underlying distribution of consumers' values reveals information about the consumers in that market. Take, for instance, the consumers residing in a certain zip code. Knowing the distribution of the consumers in that zip code induces privacy costs for a multitude of reasons, including potential profiling, harmful targeting, potential manipulation, and unfair business practices. { In this paper, we aim to explore how limiting information disclosure about the market affects the potential combinations of consumer and seller utilities. We conceptualize privacy leakage as the extent to which revealing a message about the value distribution within a market segment reveals the actual distribution of values of consumers in that segment. This measure quantifies the change in our understanding of a segment’s distribution following the receipt of a specific message.} More precisely, we require that the distribution of consumer values in each market must go through a \emph{privacy mechanism} that masks the true market and outputs a uniformly sampled market with some probability $\beta$ and outputs the true market with the complementary probability $1-\beta$. Our privacy mechanism is motivated by a market-level \emph{randomized response mechanism}---where randomized response is a well-known technique in the privacy literature, commonly utilized for the privatization of discrete-valued data \citep[see, e.g., ][]{warner1965randomized, greenberg1969unrelated}. { Additionally, as we will highlight in the next section when presenting our model, our privacy mechanism is methodology used in practice in technology companies and government agencies when aggregate statistics such as histograms are published in a 
manner that provides privacy guarantees~\cite{desfontaines2021list}. These methods ensure that a function reported over users' data does not significantly depend on any individual user's data, thereby safeguarding personal privacy \cite{dwork2014algorithmic}. }

{ In this paper, we operate within a Bayesian setting where we assume that the producer holds a prior for every market and updates it after observing the output of the privacy mechanism. To avoid any unaccounted privacy loss, it is important to choose the least informative prior. Intuitively, under the principle of indifference, a reasonable choice of prior belief for the producer would be the uniform distribution over the entire space of market distributions. Indeed, the question of selecting an \textit{uninformative prior} has been the subject of extensive literature in  Bayesian analysis, where \textit{reference priors} are a leading framework \cite{bernardo1979reference}. As we will demonstrate, the uniform distribution is the reference prior in our setting.
Given this Bayesian framework}, we ask the following question:

\begin{quote}
\textit{Given the privacy constraint, what set of consumer and producer utility pairs can be attained through all possible market segmentations?}
\end{quote}
{ Notice that by increasing $\beta$, we strengthen the privacy guarantees. The case of $\beta=1$ corresponds to a fully private setting in which the producer has no information about each segment or the aggregated market. It is worth emphasizing that in our setting, unlike in \cite{bergemann2015limits}, the producer does not know the aggregated market. This assumption is crucial when studying privacy loss, as having additional knowledge about the aggregated market would constitute unaccounted information leakage. Consider, for example, a setup with two possible values where each consumer’s value is either low or high. If the aggregated market indicates that everyone's value is high with a probability of one, then the producer would know every segment's market, even if $\beta=1$ (which is supposed to deliver a full privacy guarantee). Lastly, the case of $\beta=0$ reduces our setting to the non-private case, where the producer can see all segments without any perturbation, and hence can deduce the aggregated market as well.}

Both our analysis and our results are different in significant ways from those of the non-private case. {In particular, since the aggregated market is no longer known to the producer, there is no guarantee that the producer's surplus is at least as high as profits
under the uniform monopoly price (i.e., at or above the line $TR$ in Figure \ref{fig:S_bergemann}).}
Moreover, our results allow us to resolve the question of whether the privacy constraint always benefits the consumer. In more detail, our main results and insights are as follows.

First, in the setting of two consumer values, we are able to fully characterize the space of all consumer and producer utilities. Moreover, we are able to identify nuances in this setting that we later explore in a more general setting. In particular, we establish that the set of all consumer and producer utilities is a triangle, similar to the non-private case, but the triangle takes different shapes depending on the consumer values, the aggregate market distribution, and the privacy parameter. Figure \ref{fig:S_private_K_2} illustrates our characterization for two cases that display qualitative differences compared to the setting without privacy constraints. In this figure, the set of possible producer and consumer utilities under the privacy mechanism is depicted by the blue triangle $ABC$. The triangle $TQR$, marked by red dashed lines, represents the non-private case for the corresponding set of values and the aggregated market. We see that the privacy mechanism impacts the limits of price discrimination through a \emph{direct effect} and an \emph{indirect effect}. Let us describe these effects and explain how they shape the set of all possible consumer and producer utilities.

The direct effect arises because the producer observes a random market instead of the true market, with probability $\beta$, as a consequence of the privacy mechanism. This means that with probability $\beta$, the consumer and producer utilities are independent of the segment's markets. The direct effect implies a shift in the space of consumer and producer utilities. Considering Figure \ref{fig:S_private_K_2}, the direct effect implies that the minimum consumer utility, as opposed to the non-private case, is no longer zero; instead, there is a non-zero lower bound on the consumer utility (line AB versus line QT). It also implies that the sum of consumer and producer utilities is smaller than that of the non-private case (line BC versus line QR).

To understand the indirect effect, notice that with probability $1-\beta$, the producer observes the true market. However, the producer does not know whether or not this observed market is the truth and, therefore, adopts a different pricing strategy compared to the non-private case. The indirect effect refers to this change in the producer pricing strategy. Notably, the producer's pricing strategy is optimal from their perspective, taking into account the privacy mechanism, but it does not correspond to the optimal pricing strategy in the non-private case. To better illustrate the impact of this indirect effect, let us revisit the non-private case (as depicted by the triangle $TQR$ in Figure \ref{fig:S_private_K_2}). Imagine a segmentation for which the producer is indifferent between several prices for certain segments. In such a case, although different price choices do not alter the producer's utility, they can lead to different consumer utility. For example, every point along the line $TR$ represents the same level of utility for the producer. However, selecting various prices can minimize or maximize the consumer utilities (corresponding to points $T$ and $R$, respectively). In the private setting, however, various pricing options that make the producer indifferent ex-ante (i.e., yielding the same expected utility when considering the privacy mechanism) might actually result in different ex-post producer utilities. Consequently, opting for different prices alters both the consumer and producer utilities. This is the reason why, with privacy constraints, the line $AC$ is no longer parallel to the x-axis. This difference is more pronounced as we consider $K>2$ values (see Figure \ref{fig:S_private_K}) where the lower limit of the set of all possible consumer and producer utilities, in sharp contrast to the non-private case, is not a straight line.

The above effects and the corresponding set of all possible consumer and producer utilities are visualized in Figure \ref{fig:S_private_K_2}: the triangle $ABD$ represents the non-private triangle $TQR$ after undergoing the scaling and shifting attributable to the direct impact of privacy. The discrepancy between triangles $ABC$ and $ABD$ thus precisely indicates the indirect effect that arises from the change in the producer's pricing policy due to the privacy mechanism.

\begin{figure}
\centering
\begin{subfigure}{.5\textwidth}
  \centering
  \includegraphics[width=1.1\linewidth]{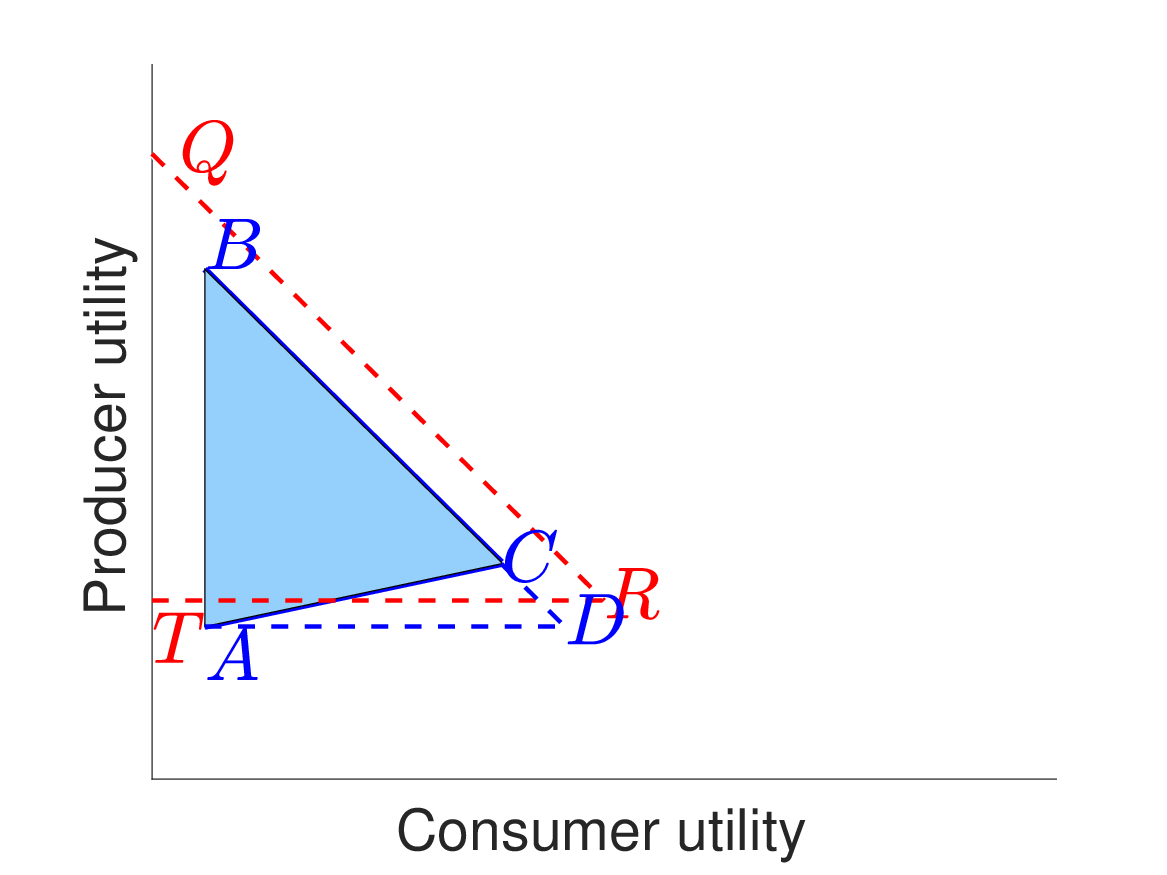}
\end{subfigure}
\begin{subfigure}{.5\textwidth}
  \centering
  \includegraphics[width=1.1\linewidth]{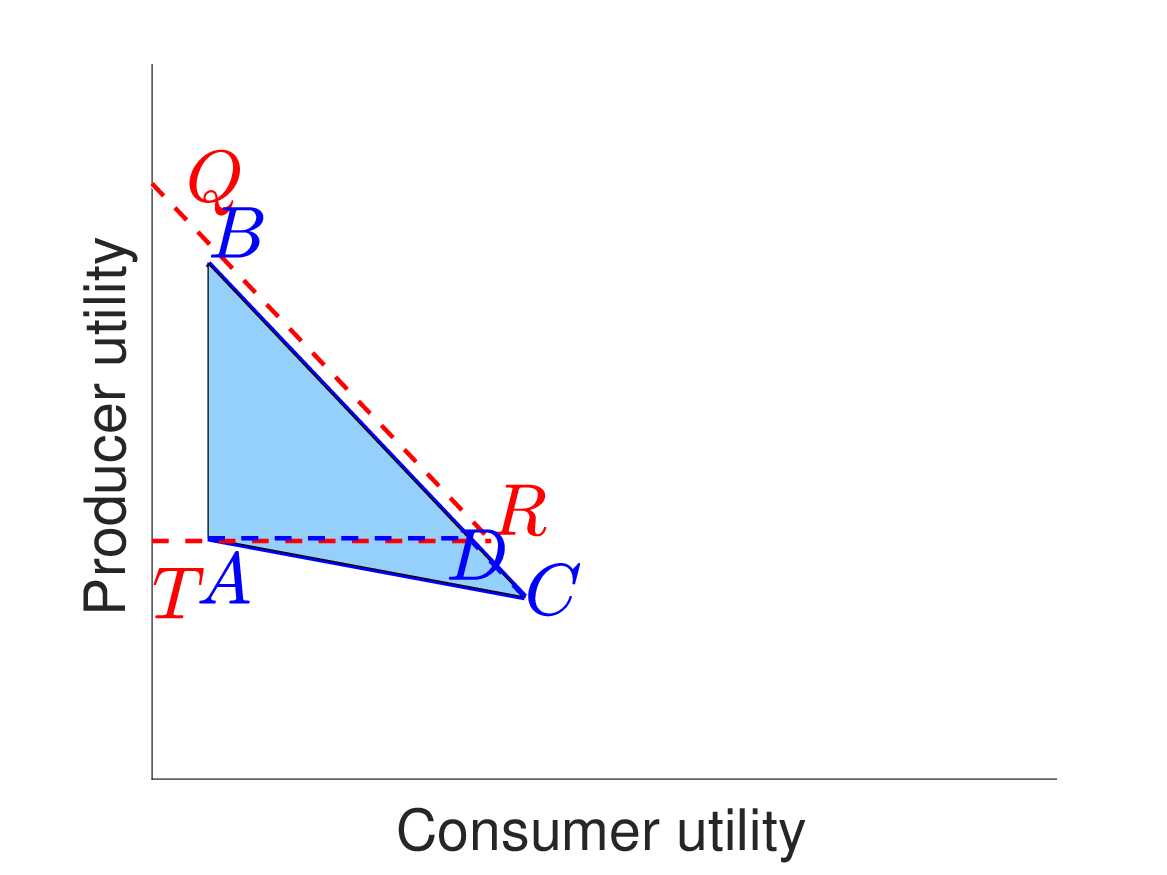}
\end{subfigure}
\caption{Illustration of our characterization of the set of all consumer and producer utilities under a privacy mechanism (the triangle $ABC$) with a comparison to the setting without privacy (the triangle $TQR$). Both panels correspond to a setting with two consumer values: high and low. The left panel corresponds to a setting where high and low consumer values are far apart. The right panel corresponds to a setting where high and low consumer values are close, and the privacy parameter is small.}
\label{fig:S_private_K_2}
\end{figure}

Second, our characterization yields the following insights. 

\emph{Insight 1:} Imposing the privacy constraint hurts the producer by decreasing both its minimum and maximum utility across all segmentations. This is intuitive because after adopting the privacy mechanism, the producer knows less about the consumer's values and, therefore, cannot extract as much surplus as in the non-private case.

\emph{Insight 2:} Imposing the privacy constraint helps the consumers by increasing their minimum utility. However, interestingly, it can hurt consumers by decreasing their maximum utility (compare, for example, the consumer utility in points $C$ and $R$ in the two panels of Figure \ref{fig:S_private_K_2}). This first phenomenon aligns with the direct impact of the privacy mechanism, as previously elaborated. Because the producer does not always observe the true market, they may sometimes choose a suboptimal price, resulting in a positive expected utility for consumers. The second phenomenon is more nuanced. We observe that the introduction of a privacy mechanism can influence the producer's market pricing strategy, leading to more conservative or riskier decisions depending on the consumers' values. This effect arises because the privacy mechanism reduces the informativeness of the market observed by the producer. On one hand, in cases where there is a significant gap between high and low consumer values, the producer finds the risk of setting higher prices justified. This shift towards riskier pricing strategies is the indirect impact of privacy we discussed earlier. Under such circumstances, if the actual market includes many low-value consumers, selecting higher prices could detrimentally affect consumer utility. On the other hand, when consumer values are relatively similar, the producer might opt for a lower, safer pricing strategy, avoiding ``unnecessary" risks. Now, if the true market is composed largely of high-value consumers, this conservative approach can lead to increased utility for consumers compared to the non-private case, thereby increasing maximum consumer utility.

\emph{Insight 3:} As we discussed above, imposing a privacy constraint decreases both the minimum and maximum producer utility while it increases the minimum consumer utility. One may conjecture that these changes amplify as the privacy parameter increases. However, interestingly, we prove that an increase in the privacy parameter $\beta$ does not necessarily decrease the producer's maximum utility and, similarly, does not necessarily increase the minimum consumer utility. Put differently, while imposing the privacy constraint always reduces the maximum producer utility, more privacy does not necessarily lead to a more reduction in this utility. Similarly, although the privacy mechanism ensures a non-zero minimum utility for the consumer, this minimum does not always increase with a higher privacy parameter. The intuition behind these nuanced non-monotonicities is similar to Insight 2, discussed above.

Third, we establish that our main insights carry over to the general setting with $K>2$ consumer values. Notably, the set of all consumer and producer utilities pairs is no longer a triangle, in sharp contrast to that of the non-private case. Instead, it is a convex polygon (see Figure \ref{fig:S_private_K} illustration for one example and how it differs from the non-private case). Similar to the case of $K=2$, we identify that the privacy mechanism introduces both direct and indirect effects on the set of possible consumer and producer utilities. The direct impact implies a combination of scaling and shifting, transforming the triangle $TQR$ into triangle $ABD$ as depicted in Figure \ref{fig:S_private_K}. The indirect effect of privacy, however, changes the shape of the utility set from a triangle to a convex polygon. More precisely, we illustrate that this set can be represented as a linear mapping of a convex polytope in $\mathbb{R}^{K^2}$ into a 2-dimensional space. This representation enables us to extend all the insights above from the case with $K=2$ to the general case with $K\geq 2$. Our analysis also provides the construction of a segmentation to achieve any feasible point for consumer and producer utilities.

\begin{figure}
  \centering
  \includegraphics[width=.5\linewidth]{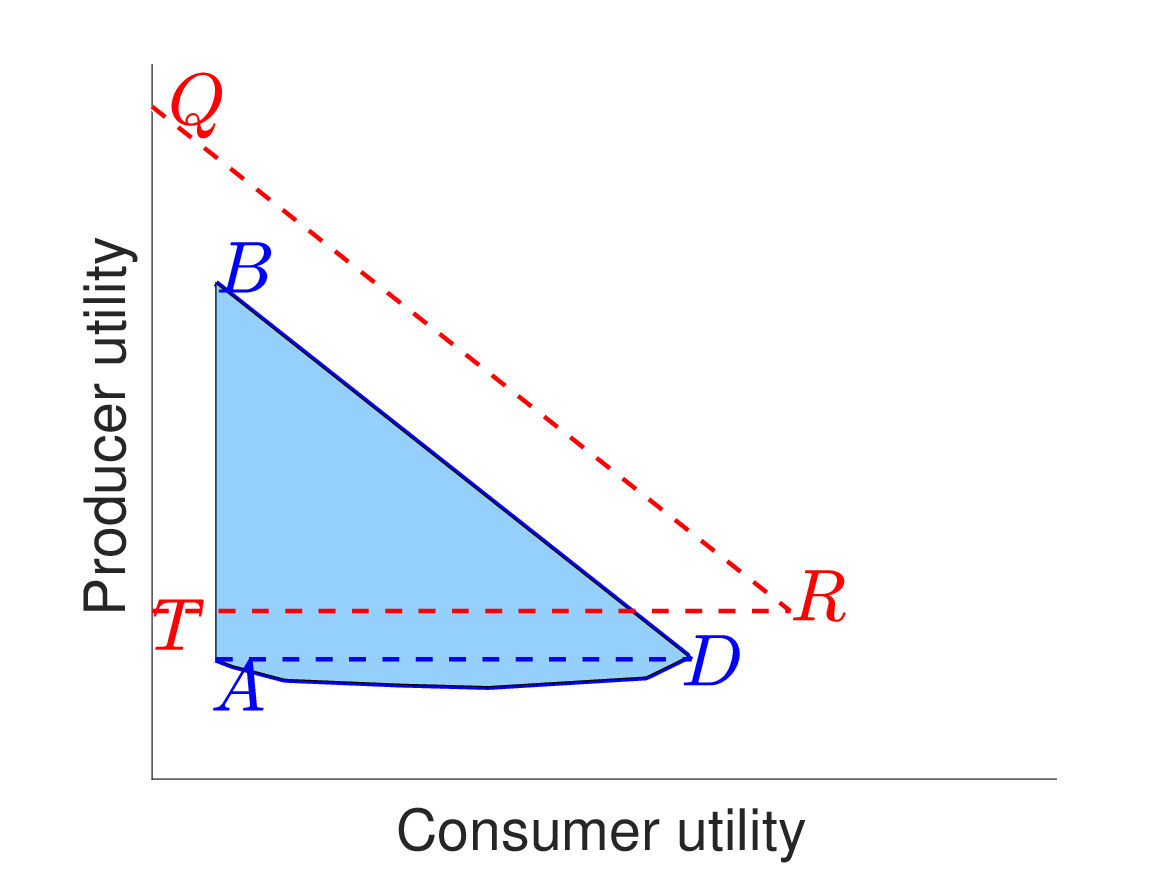}
  \caption{Illustration of the set of possible consumer and producer utilities (the solid blue area) and its comparison with the non-private case (the triangle $TQR$) for an example with $K=5$ values.}
  \label{fig:S_private_K}
\end{figure}

Furthermore, our characterization allows us to identify several key properties of the polygon that delineate the limits of price discrimination under privacy constraints. One that we highlight here is regarding the impact of privacy on first-degree price discrimination. In the non-private scenario, point $Q$ in Figure \ref{fig:S_private_K} represents perfect discrimination utilities, where each segment contains consumers of identical value. Now, let us consider the same segmentation and see how the privacy mechanism impacts the producer's utility (i.e., point B in Figure \ref{fig:S_private_K}). First, notice that because of the direct impact of the privacy mechanism (that masks segments with probability $\beta$), the producer's utility decreases to another point. The question arises: does the indirect effect of privacy further decrease the producer's utility? In other words, the question is whether the privacy mechanism can influence the producer's pricing strategy to such an extent that, even when they observe a market comprising consumers with similar values, they set a different price for that market. Interestingly, we show that this scenario is indeed possible. We specifically identify a threshold for the privacy parameter $\beta$. Once this threshold is surpassed, it prompts the producer to avoid using certain values for market pricing, as other values yield higher expected utility. Below this threshold, however, the indirect effect of privacy does not further decrease the producer's utility.

From a technical point of view, to characterize the set of all possible consumer and producer utilities under privacy constraints, we need novel techniques and analysis compared to the existing ones in the literature. Notably, the analysis in \cite{bergemann2015limits} relies heavily on the premise that producer utility is bounded below by optimal uniform pricing. This assumption, however, does not hold under privacy constraints since the aggregated market remains undisclosed to the producer. Furthermore, their argument is built on characterizing extremal markets—markets wherein the producer is indifferent between choosing any value within the support of the distribution of the consumer values. In contrast, our argument employs a \emph{merging technique}, effectively simplifying the problem to an examination of segmentations where, for each value, no more than one segment is priced at that value. This approach not only makes our analysis applicable to the private case but also builds a framework that easily extends to the general setting with multiple consumer values.

\subsection{Related work}
Our paper relates to the literature on price discrimination. Earlier works on this topic include  \cite{robinson1969economics}, \cite{schmalensee1981output}, \cite{varian1985price}, \cite{aguirre2010monopoly},  \cite{cowan2016welfare}. More recent papers on this topic include 
\cite{roesler2017buyer} who characterize the buyer-optimal signaling, 
\cite{bergemann2015limits} who prove that it is possible to achieve any total surplus and its division between consumers and the producer through segmentation, \cite{haghpapanah2019consumer} who characterize the consumer optimal segmentation, and
\cite{elliott2021market} who study how a platform's information design can enable the division of surplus between the consumer and the seller.

The closest paper to ours is \cite{bergemann2015limits}, which characterizes the set of all possible consumer and producer utilities achievable by segmentations. We build on the setting of this paper by asking what happens if we impose privacy constraints. As discussed above, this fundamentally changes the results and also requires different analyses. This work has motivated several follow-up works that consider either extensions or variations of it. For instance, a similar problem in multiproduct settings has been studied in \cite{haghpanah2021pure}, \cite{haghpanah2022limits}, and \cite{haghpanah2023pareto}.  Consumer profiling via information design has been studied in \cite{fainmesser2023consumer}. Robust price discrimination has been studied in \cite{arieli2024robust}. A unified approach to second and third degree price discrimination has been studied in \cite{bergemann2024unified}. Extension to randomized auctions (possibly over a menu of large size) has been studied in \cite{ko2022optimal}, and price discrimination with fairness considerations has been studied in \cite{banerjee2024fair}. To the best of our knowledge, the limits of price discrimination under privacy constraints have not been studied in the literature.\footnote{Segmentation can be viewed as a Bayesian persuasion problem as introduced by \cite{kamenica2011bayesian}. We refer to \cite{bergemann2015limits} for a detailed description of this connection.}

Our paper also relates to the literature on privacy in markets and platforms. \cite{acquisti2016economics} survey the empirical and theoretical works that study the economic value and consequences of protecting and disclosing personal information on consumers and producers. More recently, \cite{posner2018radical} explore the consequences of consumers sharing data with online platforms, and \cite{jones2020nonrivalry} study the consequences of data non-rivalry (i.e., consumer data can be used by many producers simultaneously) for the consumers and the producers. \cite{acemoglu2019too},  \cite{bergemann2020economics}, and \cite{ichihashi2020online} study the privacy consequences of data externality, whereby a
user's data reveals information about others, \cite{acemoglu2023good} study the platform's optimal architecture that achieves the Pareto frontier of user privacy level and platform estimation error, and \cite{acemoglu2023model} develop an experimentation game to study the harms of the platform's information advantage in product offering. Finally, privacy and personal data in financial markets have been studied in \cite{farboodi2023data}.

More closely related to ours are the works at the intersection of privacy and price discrimination. In this regard,  \cite{ali2020voluntary} study the implications of giving consumers control over their data and the ability to choose how producers access their data. They investigate whether such regulations improve consumer utility and find that consumer control can guarantee gains for every consumer type relative to both perfect price discrimination and no personalized pricing. \cite{hidir2021privacy} characterize the consumer optimal market segmentation compatible with the consumer’s incentives to reveal their preferences voluntarily. They further investigate the implication of their results for consumer privacy and price discrimination in online retail. We depart from this literature by characterizing the set of all possible consumer and producer utilities achievable by market segmentations under privacy constraints.\footnote{Our paper is also related to the literature on differential privacy introduced in \cite{dwork2006our} and \cite{dwork2006calibrating} (see \cite{dwork2014algorithmic} for a survey). In particular, our privacy mechanism can be viewed as a variation of the \emph{randomized response} that is used to guarantee differential privacy for discrete value data points (see, e.g., \cite{erlingsson2014rappor}). }

The rest of the paper proceeds as follows. In Section \ref{sec:Model}, we present our model and introduce the problem formulation. In Section \ref{sec:SpecialK=2}, we focus on a setting with $K=2$ consumer values and characterize the set of all possible consumer and producer utilities that can be achieved by segmentation. We then discuss the insights we obtain from our analysis and, in particular, whether imposing privacy constraints helps or hurts the producer and the consumer.  In Section \ref{Sec:GeneralK}, we show how our main results and insights continue to hold for the general case of $K>2$ consumer values. Section \ref{Sec:Conclusion} concludes while the Appendix presents the omitted proofs from the text. 

\section{Model}\label{sec:Model}
We build our model upon the one proposed by \cite{bergemann2015limits}. In particular, we consider a producer's problem of pricing a product for a continuum of consumers represented by the interval $\Omega := [0,1]$. We normalize the cost of production to zero and assume that each consumer's value for the product belongs to the set $\mathcal{V} := \{v_1, v_2, \cdots, v_K\}$ for some $K \in \mathbb{N}$, where $v_K > \cdots > v_2 > v_1 > 0$. The value of a consumer is denoted by a random variable $V : \Omega \to \mathcal{V}$. We assume that a consumer purchases the product if the offered price is lower than or equal to their value.

\paragraph{Markets:} A \textit{market} $x$ is a distribution over $\mathcal{V}$, and hence, the set of all markets is the simplex over $\mathcal{V}$, given by
\begin{equation} \label{eqn:all_markets}
\mathcal{X}:= \Delta(\mathcal{V}) = \left \{x:\mathcal{V} \to \mathbb{R}_{\geq 0} ~\big \vert ~ \sum_{k=1}^{K} x(v_k) = 1 \right \}.   
\end{equation}
We denote the aggregated market, which corresponds to the distribution of $V$, by $x^*$.

\paragraph{Segmentation:} A \textit{segmentation} is a partitioning of the aggregated market $x^*$ into different markets. We illustrate a segmentation using a distribution $\sigma$ over the space of all markets, $\mathcal{X} = \Delta(\mathcal{V})$, as defined in equation \eqref{eqn:all_markets}. For any $x \in \mathcal{X}$ within the support of $\sigma$, we have a segment in the partitioning of the aggregated market which its corresponding market is $x$ and its population size is $\sigma(x)$. The set of all segmentations is given by\footnote{As our proofs show, we can limit our attention to partitions comprising finitely many segments, i.e., distributions $\sigma$ with finite support, without loss of generality.}
\begin{equation} \label{eqn:all_segmentations}
\Sigma := \left \{ \sigma \in \Delta(\mathcal{X}) ~\big \vert ~
\sum_{x \in \text{supp}(\sigma)} \sigma(x).x = x^*, |\text{supp}(\sigma)| < \infty
\right \}.
\end{equation}

\paragraph{Privacy mechanism:} A \textit{privacy mechanism} $\mathcal{M}$ is a mapping from the set of all markets to itself, i.e., $\mathcal{M}:\mathcal{X} \to \mathcal{X}$. More specifically, let $x \in \text{supp}(\sigma)$ for some segmentation $\sigma \in \Sigma$. To take the privacy considerations into account, we assume the producer \textit{observes} $\mathcal{M}(x)$ (instead of the true $x$), which is a private version of the market $x$. In particular, we focus on the following class of mechanisms, parameterized by $\beta \in [0,1]$:
\begin{equation} \label{eqn:privacy_mechanism}
\mathcal{M}_\beta (x) := 
\begin{cases}
x & \text{ with probability } 1-\beta \\
\text{Unif}(\mathcal{X}) & \text{ with probability } \beta,	
\end{cases}
\end{equation}
where $\text{Unif}(\mathcal{X})$ represents drawing a market from $\mathcal{X}$ uniformly at random. One can interpret $\mathcal{M}_\beta(\cdot)$ as a mechanism that returns the true market with probability $1-\beta$ but masks it with the complimentary probability $\beta$ (and outputs a uniform noise instead).  Notice that for $\beta=0$, our problem becomes identical to that of \cite{bergemann2015limits}, which does not consider privacy. Moreover, as $\beta$ increases, one can learn less about the distribution of consumer values in a market. 

Our privacy mechanism can be linked to the randomized response mechanism, albeit applied at the market level. The randomized response is a well-known technique in the privacy literature, commonly utilized for making a variable with discrete values private \citep{warner1965randomized, greenberg1969unrelated}. This mechanism maintains the true value of the variable with a certain probability while substituting it with a uniformly chosen random variable from the variable's finite range with the complementary probability. Our approach is similar, with the distinction being that the market space is infinite. Moreover, applying definitions from the privacy literature, we can quantify the \emph{privacy-leakage} associated with the mechanism $\mathcal{M}_\beta(\cdot)$. Let us formalize this connection next.

\begin{definition}\label{def:privacy:leakage}[\cite{rassouli2019optimal}]
   The privacy-leakage about random variable $X$ by revealing $Y$ is
   \begin{align*}
       \mathbb{E}_{Y} \left[\mathrm{TV}\left( p_{X|Y}(\cdot | Y), p_{X}(\cdot)\right) \right],
   \end{align*}
   where $p_X(\cdot)$ and $p_{X|Y}(\cdot | Y)$ are the probability density function and the conditional probability density function of $X$, respectively, and $\mathrm{TV}(\cdot, \cdot)$ denotes the total variation distance between two random variables.
\end{definition}
\begin{proposition}\label{Pro:privacy:leakage}
   The privacy-leakage (defined in Definition \ref{def:privacy:leakage}) of our masking mechanism $\mathcal{M}_\beta(\cdot)$ is equal to $1-\beta$.
\end{proposition}
This privacy-leakage of our masking mechanism is maximized when $\beta=0$ and is reduced to zero when $\beta=1$, where the mechanism returns pure noise (see \cite{rassouli2019optimal} for further discussion on the relationship between this measure and other definitions of privacy loss).

{ Our privacy mechanism privatizes the distribution of values, i.e., the market, and is therefore closely related to the literature on differentially private histograms, with applications ranging from limiting inferences about the prevalence of a disease in a group of people to protecting password frequency lists (see for instance \cite{xu2013differentially, suresh2019differentially}), and with fielded applications in companies and agencies such as Apple, Google, and the United States Census Bureau \cite{AppleDP, abowd2018us}. Briefly, a randomized algorithm, as a function of users' data, is considered differentially private if the distribution of the algorithm's output does not change significantly by altering the data of one user. The motivation behind this definition is that if the output is not overly sensitive to a user's data, it implies that we cannot learn much about a user by observing the output \citep{dwork2006our, dwork2006calibrating}. In \cref{sec:dp_histogram}, we further discuss differential privacy and its connection to the privacy mechanism in \eqref{eqn:privacy_mechanism}.

}

\paragraph{Utility functions:} If we set the product's price for a market $x$ equal to some $p \in [0, \infty)$, the producer's utility and the consumers' utility are given by
\begin{align} \label{eqn:utilities}
\begin{split}
\mathcal{U}_p(p,x) &:= p \sum_{k=1}^{K} x(v_k) \mathbbm{1}(p \leq v_k),\\
\mathcal{U}_c(p,x) &:= \sum_{k=1}^{K} x(v_k) (v_k - p) \mathbbm{1}(p \leq v_k).
\end{split}
\end{align}
{
\paragraph{Producer's belief update:} As we discussed in the introduction, the producer does not know the aggregated market or any of the segment markets. Instead, they hold a prior regarding each segment's market and then update it using Bayes' theorem after observing the market through the privacy mechanism. We denote the prior by $\pi(\cdot)$ and the posterior after observing $\hat{x} = \mathcal{M}_\beta(x)$ by $\pi_\beta(\cdot | \hat{x})$. Note that both distributions are defined over the space of all markets $\mathcal{X}$.

Here, we need to make a modeling assumption about the choice of prior. In particular, we want the prior to carry no information about the underlying market, as otherwise, it would mean a privacy leakage that is not accounted for. Consider, for instance, a setting with $K=2$ values where all consumers have the higher value ($v_2$), i.e., $x^*(v_2)=1$. If the prior assigns a high probability to the markets where the probability of the high value is close to one, i.e., $x$ such that $x(v_2) \approx 1$, then even with $\beta = 1$, the posterior would still place a high probability on markets with $x(v_2)$ close to one. This implies that the seller almost knows that most of the market comprises high-value consumers. In other words, even $\beta=1$ fails to provide strong privacy protection.
In the remainder of this paper we make the following assumption concerning the prior distribution $\pi(\cdot)$.
\begin{assumption}
The prior $\pi(\cdot)$ is equal to $\text{Unif}(\mathcal{X})$ where $\text{Unif}(\mathcal{X})$ denotes the uniform distribution over the space of all markets $\mathcal{X}$. In other words, $\pi(x) = \tfrac{1}{\text{Vol}(\mathcal{X})}$ for any market $x \in \mathcal{X}$.
\end{assumption}
This assumption ensures that we select a least informative prior, which can be formalized variationally as the prior that maximizes the expected Kullback-Leibler (KL) divergence of posterior from the prior \citep{bernardo1979reference}. 
\begin{definition}
The reference prior is the prior that maximizes the KL divergence in expectation: 
\begin{equation} \label{eqn:argmax_reference_prior}
\argmax_{\pi(\cdot)} \int_{\mathcal{X}} \mu_\beta(\hat{x}) \left( \int_{\mathcal{X}} \pi_\beta(x|\hat{x}) \log(\frac{\pi_\beta(x|\hat{x})}{\pi(x)}) ~dx \right) d\hat{x},
\end{equation}
where $\mu_\beta(\cdot)$ denotes the distribution of the observed market $\hat{x}$
\end{definition}
The reference prior in our setting is the uniform distribution over the space of markets $\mathcal{X}$. 
\begin{proposition}\label{propsotion:reference_prior}
$\text{Unif}(\mathcal{X})$ is the solution to \eqref{eqn:argmax_reference_prior}.   
\end{proposition}
In fact, if we want to choose a prior that conveys no information about the underlying market, the prior distribution is an intuitive choice, as all different markets are equally likely under such a prior. Therefore, the producer starts with $\text{Unif}(\mathcal{X})$ as their prior belief, and after observing $\hat{x} = \mathcal{M}_\beta(x)$, they form the posterior. By applying Bayes' theorem, it can be demonstrated that this posterior distribution has a mass of size $1-\beta$ over $\hat{x}$, and the remaining $\beta$ mass is uniformly distributed over $\mathcal{X}$.
}
\paragraph{Optimal pricing rules:} 

The producer's expected revenue from a market conditional on observing $\hat{x}$ and under some price $p \in \mathbb{R}_{\geq 0}$ is given by
\begin{equation} \label{eqn:expected_utility}
\mathcal{U}_p (p | \hat{x}) := 
\mathbb{E}_{~x \sim \pi_\beta(x | \hat{x} )} \left [ \mathcal{U}_p(p,x) \right ].	
\end{equation}
A price $p$ is called \emph{optimal} for a market given observation $\hat{x}$ if its corresponding expected revenue is at least as great as that of any other price:
\begin{equation} \label{eqn:optimal_pricing}
\mathcal{U}_p (p | \hat{x}) \geq \mathcal{U}_p (p' | \hat{x})	
\text{ for any } p' \geq 0.
\end{equation}

It is evident that the optimal price always belongs to the set $\mathcal{V}$; therefore, we will limit our search for the optimal prices to this set. 

A pricing rule $\phi:\mathcal{X} \to \Delta(\mathcal{V})$ is a function where, for any observation $\hat{x}$, $\phi(\hat{x})$ specifies a distribution over prices. Specifically, a pricing rule $\phi$ is considered optimal if, for each $\hat{x}$, the support of $\phi(\hat{x})$ consists solely of prices that are optimal given the observation $\hat{x}$.

\paragraph{Formulating our goal:} We now formalize our main question. 
Given an optimal pricing rule $\phi(\cdot)$, we denote the expected utility of the producer and consumer from a market $x \in \mathcal{X}$ by $U_p^\phi(x)$ and $U_c^\phi(x)$, respectively. This expectation accounts for the randomness arising from the producer observing the private version $\mathcal{M}_\beta(x)$ instead of $x$ and applying the (potentially random) optimal pricing rule $\phi(\cdot)$ based on this observation. More formally, we have
\begin{align}
U_p^\phi(x) := \mathbb{E}_{\hat{x} \sim \mathcal{M}_\beta(x)} \left [ 
\mathbb{E}_{p \sim \phi(\hat{x})} [~\mathcal{U}_p(p,x)]\right ] \text{ and } \quad 
U_c^\phi(x) := \mathbb{E}_{\hat{x} \sim \mathcal{M}_\beta(x)} \left [ 
\mathbb{E}_{p \sim \phi(\hat{x})} [~\mathcal{U}_c(p,x)]\right ].
\end{align}
{It is worth highlighting that the above utilities are computed with respect to the true market $x$ (and not as an expectation with respect to the prior $\text{Unif}(\mathcal{X})$). In other words, there is a fixed market $x$, and the producer has a uniform prior $\text{Unif}(\mathcal{X})$ on it. This prior is updated to the posterior $\pi_\beta(x | \hat{x})$ after observing $\hat{x} \sim \mathcal{M}_\beta(x)$. Based on this posterior, the producer selects the optimal price $p \sim \phi(\hat{x})$. Now, the functions $U_p^\phi(x)$ and $U_c^\phi(x)$ calculate the utilities for the producer and consumer of the original market $x$ when priced at $p$. }

Formally, our goal is to characterize the set of all possible pairs of expected utilities for both the producer and consumer across all potential segmentations and optimal pricing rules. In other words, we would like to characterize the following set:
\begin{equation} \label{eqn:all_pairs}
\mathcal{S} := \left \{ 
\Big[ \sum_{x \in \text{supp}(\sigma)} \sigma(x) U_c^\phi(x), \sum_{x \in \text{supp}(\sigma)} \sigma(x) U_p^\phi(x)
\Big]^\top ~ \Big \vert ~ \sigma \in \Sigma \text{ and } \phi \text{ is an optimal pricing rule}
\right \}.    
\end{equation}
\section{Price discrimination limits under privacy: the case of $K=2$}\label{sec:SpecialK=2}

We begin by examining the simpler scenario where the consumer's valuation is limited to two values, i.e., $K=2$,  to gain insights into the limits of price discrimination under privacy constraints. In this scenario, the set of all markets, $\mathcal{X}$, is, in fact, the set of Bernoulli distributions, which can be parameterized by the probability of the high value, $v_2$, denoted by $\alpha \in [0,1]$. Throughout this section, we use $x$ and $\alpha$ interchangeably to represent the market. In particular, $\alpha^*$ denotes the probability of the high value for the aggregated market $x^*$. We also define $\eta$ as the ratio of the lower value to the higher value, i.e., 
\begin{equation}
\eta := \frac{v_1}{v_2}\le 1.    
\end{equation}
We first characterize the optimal price conditioning on the observed market $\hat{x}$. 
\begin{lemma} \label{lemma:pricing}
Suppose the producer observes $\hat{x} = (1- \hat{\alpha}, \hat{\alpha})$, where $\hat{\alpha}$ denotes the observed probability of the high value $v_2$. Define the threshold $t^*$ as:
\begin{equation}
t^* := \frac{\eta - \beta/2}{1-\beta}.
\end{equation}
Then, the producer should set the price to $v_1$ (i.e., it is the optimal price) if and only if $\hat{\alpha} \leq t^*$, and $v_2$ is the optimal price otherwise.
\end{lemma}
Note that when $2\eta \le \beta$, this lemma implies that the optimal price becomes $v_2$, regardless of the observation $\hat{x}$. Similarly, when $\beta \geq 2 (1-\eta)$, the optimal price is $v_1$. Note that when the producer sets the price to the high value $v_2$, only consumers with the high value will purchase the product, resulting in their expected utility being zero due to the absence of consumer surplus. Conversely, when the producer opts for the low value $v_1$ as the price, all consumers—regardless of their valuation—will purchase the product. This pricing strategy allows consumers with the high value to achieve a positive expected surplus, as they gain the difference between their valuation and the lower price. Therefore, if $\beta\ge 2 \eta$, we have 
\begin{align*}
    \mathcal{S}= \left\{\left[0,v_2  \alpha^*\right]^T\right\}
\end{align*}
and if $\beta\ge 2 (1-\eta)$, then we have 
\begin{align*}
    \mathcal{S}= \left\{\left[\left(v_2-v_1\right) \alpha^*,v_1\right]^T\right\}.
\end{align*}
We next focus on the more interesting scenario in which the following assumption holds.
\begin{assumption}\label{assump:beta:intermediate}
We suppose $\beta \leq \min\{2\eta, 2(1-\eta)\}$. 
\end{assumption}
As established in Lemma \ref{lemma:pricing}, in this case, the optimal pricing rule sets the price to $v_1$ for observations $\hat{\alpha} < t^*$ and $v_2$ for $\hat{\alpha} > t^*$. The following proposition  analyzes the effects of privacy mechanism on the set $\mathcal{S}$, distinguishing between the direct and the indirect effects of privacy. 
\begin{proposition}\label{proposition:S_to_S'}
Suppose Assumption \ref{assump:beta:intermediate} holds. Then, $\mathcal{S}$ is given by
\begin{equation}
\mathcal{S} = \left \{ 
\beta \Big[ t^* \alpha^* (v_2-v_1), ~ t^* v_1 + (1-t^*) \alpha^* v_2 \Big]^\top + (1-\beta) s'
~ \Big \vert ~ s' \in \mathcal{S}'
\right \},    
\end{equation}    
where $\mathcal{S}'$ is defined as
\begin{align}
\mathcal{S'} & := 
\biggl\{ \Big[ \gamma_1 \alpha_1 (v_2-v_1), \gamma_1 v_1 + \gamma_2 \alpha_2 v_2 \Big]^\top ~ \Big \vert \label{eqn:prop_merge} \\
&
\qquad \alpha_1 \in [0,\min\{t^*, \alpha^*\}], \alpha_2 \in [\max\{t^*, \alpha^*\},1], \gamma_1 \text{ and } \gamma_2 \in [0,1], \gamma_1 + \gamma_2 = 1, \gamma_1 \alpha_1 + \gamma_2 \alpha_2 = \alpha^*
\biggr\}. \nonumber
\end{align}
\end{proposition} 
This result shows that the set $\mathcal{S}$ undergoes three distinct changes as $\beta$ varies. First, the set $\mathcal{S}'$ is scaled by a factor $1-\beta$. Second, its location in the plane shifts due to the vector $\beta \Big[ t^* \alpha^* (v_2-v_1), ~ t^* v_1 + (1-t^*) \alpha^* v_2 \Big]^\top$ added to it. These two changes are because the producer observes a completely random market with probability $\beta$ and sets the price based on this random market, resulting in a constant expected utility that is independent of the true market. Lastly, the shape of $\mathcal{S}$ evolves due to the corresponding changes in the shape of $\mathcal{S}'$. 
Notice that the shape of $\mathcal{S}'$ is influenced by $t^*$, which represents the producer's threshold for the proportion of high-value customers in order to choose the higher price. In fact, substituting $t^*$ in the definition of $\mathcal{S}'$ with $\eta$ (the optimal threshold in the non-private case) aligns $\mathcal{S}'$ precisely with the set $\mathcal{S}$ in the absence of privacy constraints. This last change can be interpreted as the indirect impact of the privacy mechanism, which is rooted in the producer adjusting its pricing strategy due to the noisy observation. 

Also, notice that Proposition \ref{proposition:S_to_S'} suggests that it suffices to focus only on segmentations with two markets $(1-\alpha_1, \alpha_1)$ and $(1-\alpha_2, \alpha_2)$ with probabilities $\gamma_1$ and $\gamma_2$, respectively. Moreover, the first segment's fraction of high values, i.e., $\alpha_1$, is less than or equal to $t^*$, and the other's (i.e., $\alpha_2)$ is greater than or equal to $t^*$. 

In the appendix, we characterize the set $\mathcal{S}'$ and show that it takes the form of a triangle, with its shape changing as $\beta$ varies. That said, 
we next formally characterize the set $\mathcal{S}$, taking into account all these three factors. 

\begin{theorem} \label{theorem:S_K=2}
Suppose Assumption \ref{assump:beta:intermediate} holds. Then, the set $\mathcal{S}$ is the triangle $ABC$, given by     
\begin{align} \label{eqn:ABC}
\begin{split}
A &= \beta \bm{c} + (1-\beta) A' \text{ with } A':= \biggl[0, ~\alpha^*v_2 + v_1 \frac{(t^*-\alpha^*)_+}{t^*} \biggr]^\top, \\
B &= \beta \bm{c} + (1-\beta) B' \text{ with } B':= \biggl [0,~\alpha^*v_2 + (1-\alpha^*)v_1 \biggr ]^\top, \\
C &= \beta \bm{c} + (1-\beta) C' \text{ with } C':=\biggl [\alpha^*(v_2-v_1) - (v_2-v_1)\frac{(\alpha^*-t^*)_+}{1-t^*}, ~v_1 + (v_2-v_1)\frac{(\alpha^*-t^*)_+}{1-t^*} \biggr]^\top,
\end{split}
\end{align}
where 
\begin{equation}
\bm{c} =  \Big[ t^* \alpha^* (v_2-v_1), ~ t^* v_1 + (1-t^*) \alpha^* v_2 \Big]^\top.   
\end{equation}
\end{theorem}
To better understand this characterization, recall that, as we discussed earlier, the effect of the privacy mechanism can be decomposed into two components: a direct one and an indirect one. The direct effect is due to the market being masked with probability $\beta$ and the producer choosing the price based on pure noise observation. This results in a constant term in utilities that is independent of the true market, captured by the term $\beta \bm{c}$ in Equation $\eqref{eqn:ABC}$.

The indirect effect of privacy arises from changes in the producer’s pricing strategy and is represented by the difference between triangle $A’B’C’$ and the non-private case $TQR$ in Figure $\ref{fig:S_bergemann}$. Specifically, the lines $A’B’$ and $B’C’$ are analogous to the corresponding lines $QT$ and $QR$ in the non-private case, where the first represents the condition that consumer utility is non-negative, and the second shows that the sum of consumer and producer utilities is less than the maximum possible surplus. However, line $A’C’$ is where the indirect effect of privacy is observed. In the non-private case, the line $TR$ is parallel to the x-axis, since the producer utility is lower bounded by the utility obtained from uniform pricing. In the presence of a privacy constraint, this line is no longer parallel to the x-axis. Moreover, this indirect effect of privacy becomes more pronounced when considering $K>2$ values, as this lower boundary of the set $\mathcal{S}$ is no longer a straight line but becomes piece-wise linear.

\subsection{Illustration of the set $\mathcal{S}$}

As equation \eqref{eqn:ABC} in \cref{theorem:S_K=2} suggests, there are two main regimes in characterizing the triangle $\mathcal{S}$: $t^* \geq \alpha^*$ and $t^* \leq \alpha^*$, where $t^*$ itself depends on $\eta$ and $\beta$. We next discuss how the set $\mathcal{S}$ evolves with varying $\beta$, for the two cases of $\alpha^* \geq \eta$ and $\alpha^* \leq \eta$.

Figure \ref{fig:S_Rmk1} depicts $\mathcal{S}$ for the case $\alpha^* \geq \eta$. In particular, Figure \ref{fig:S_Rmk1_4} illustrates the non-private case where the triangle $ABC$ coincides with the triangle $TQR$ that we introduced in the Introduction (see Figure \ref{fig:S_bergemann}). In the presence of privacy mechanism, triangle $TQR$ maps to triangle $\tilde{T}B\tilde{R}$ after going through the shifting and scaling operators, meaning that the difference between triangles $ABC$ and $\tilde{T}B\tilde{R}$ is purely attributed to the indirect effect of privacy.

Figure \ref{fig:S_Rmk1_1} illustrates the case $\eta \leq 1/2$ which implies $\alpha^* \geq \eta \geq t^*$. In this case,  as $\beta$ decreases from $2\eta$ to $0$, $t^*$ increases from $0$ to $\eta$ which implies $C$ moving down from $B$ to $\tilde{R}$.

Figures \ref{fig:S_Rmk1_2} and \ref{fig:S_Rmk1_3} correspond to the case $\alpha^* \geq \eta \geq 1/2$. Notice that $\eta \geq 1/2$ implies $t^* \geq \eta$, and hence in this case, we could have either $\alpha^* \geq t^*$ or $\alpha^* \leq t^*$, depending on the value of $\beta$. Let $\tilde{\beta}$ be the value of $\beta$ for which $t^* = \alpha^*$, i.e.,
\begin{equation}
\frac{\eta - \tilde{\beta}/2}{1-\tilde{\beta}} = \alpha^*.
\end{equation}
Figure \ref{fig:S_Rmk1_2} illustrates the case $\beta \geq \tilde{\beta}$ which implies $t^* \geq \alpha^* \geq \eta$. As $\beta$ decreases from $2(1-\eta)$ to $\tilde{\beta}$, $t^*$ decreases from $1$ to $\alpha^*$ which implies $A$ moving from $B$ down towards $\tilde{T}$.
Finally, Figure \ref{fig:S_Rmk1_4} illustrates the case $\beta \leq \tilde{\beta}$ which implies $\alpha^* \geq t^* \geq \eta$. In this case, as $\beta$ decreases from $\tilde{\beta}$ to $0$, $t^*$ decreases from $\alpha^*$ to $\eta$ which implies $C$ moving up towards $\tilde{R}$. 
\begin{figure}
\centering
\begin{subfigure}{.5\textwidth}
  \centering
  \begin{tikzpicture}
    \draw[->] (0,0) -- (5,0) node[above left] {Consumer utility}; 
    \draw[->] (0,0) -- (0,5) node[right] {Producer utility}; 
    
    \coordinate (B) at (0.2,3.4); 
    \coordinate (A) at (0.2,1.4);
    \coordinate (C) at (1.2,2.4); 
    \coordinate (D) at (2.2,1.4);

    \draw[fill=gray!30] (A) -- (B) -- (C) -- cycle; 

    \draw[dotted] (A) -- (D);
    \draw[dotted] (B) -- (D);

    \node[below right] at (0,1.4) {$A$ or $\tilde{T}$};
    \node[above] at (B) {$B$};
    \node[right] at (C) {$C$};
    \node[right] at (D) {$\tilde{R}$};
\end{tikzpicture}
  \caption{$\eta \leq 1/2$}
  \label{fig:S_Rmk1_1}
\end{subfigure}
\begin{subfigure}{.5\textwidth}
  \centering
  \begin{tikzpicture}
    \draw[->] (0,0) -- (5,0) node[above left] {Consumer utility}; 
    \draw[->] (0,0) -- (0,5) node[right] {Producer utility}; 
    
    \coordinate (B) at (0.4,3); 
    \coordinate (A) at (0.4,2.3);
    \coordinate (AA) at (0.4,1.4);
    \coordinate (C) at (2.7,0.7); 
    \coordinate (D) at (2,1.4);
    \coordinate (DD) at (0.4, 0.7);

    \draw[fill=gray!30] (A) -- (B) -- (C) -- cycle; 

    \draw[dotted] (A) -- (AA);
    \draw[dotted] (D) -- (AA);

    \node[right] at (A) {$A$};
    \node[left] at (AA) {$\tilde{T}$};
    \node[above] at (B) {$B$};
    \node[right] at (C) {$C$};
    \node[right] at (D) {$\tilde{R}$};
\end{tikzpicture}
  \caption{$\eta \geq 1/2$ and $\beta \geq \tilde{\beta}$}
  \label{fig:S_Rmk1_2}
\end{subfigure}
\newline \vspace{2mm}
\begin{subfigure}{.5\textwidth}
  \centering
  \begin{tikzpicture}
    \draw[->] (0,0) -- (5,0) node[above left] {Consumer utility}; 
    \draw[->] (0,0) -- (0,5) node[right] {Producer utility}; 
    
    \coordinate (B) at (0.2,3.4); 
    \coordinate (A) at (0.2,1.4);
    \coordinate (C) at (2.7,0.9); 
    \coordinate (D) at (2.2,1.4);

    \draw[fill=gray!30] (A) -- (B) -- (C) -- cycle; 

    \draw[dotted] (A) -- (D);
    \draw[dotted] (B) -- (D);

    \node[below right] at (0,1.3) {$A$ or $\tilde{T}$};
    \node[above] at (B) {$B$};
    \node[right] at (C) {$C$};
    \node[right] at (D) {$\tilde{R}$};
\end{tikzpicture}
  \caption{$\eta \geq 1/2$ and $\beta \leq \tilde{\beta}$}
  \label{fig:S_Rmk1_3}
\end{subfigure}
\begin{subfigure}{.5\textwidth}
  \centering
  \begin{tikzpicture}
    \draw[->] (0,0) -- (5,0) node[above left] {Consumer utility}; 
    \draw[->] (0,0) -- (0,5) node[right] {Producer utility}; 
    
    \coordinate (B) at (0,4.3); 
    \coordinate (A) at (0,2);
    \coordinate (D) at (3,2);

    \draw[fill=gray!30] (A) -- (B) -- (D) -- cycle; 

    \draw[dotted] (A) -- (D);
    \draw[dotted] (B) -- (D);

    \node[below right] at (A) {$A$ or $T$};
    \node[right] at (B) {$B$ or $Q$};
    \node[right] at (D) {$C$ or $R$};
\end{tikzpicture}
  \caption{Non-private case $\beta=0$}
  \label{fig:S_Rmk1_4}
\end{subfigure}
\caption{Illustration of $\mathcal{S}$ for the case $\alpha^* \geq \eta$.}
\label{fig:S_Rmk1}
\end{figure}
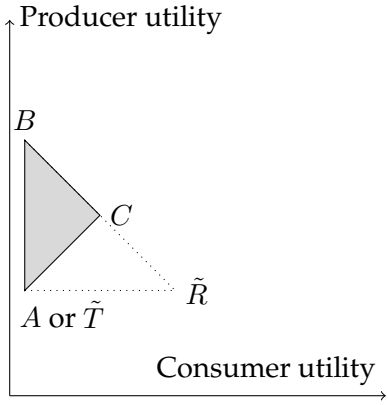
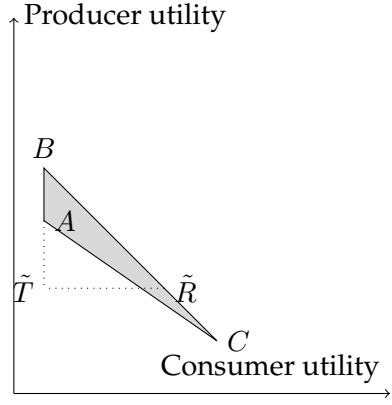
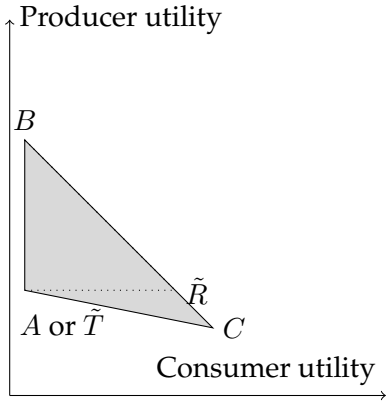
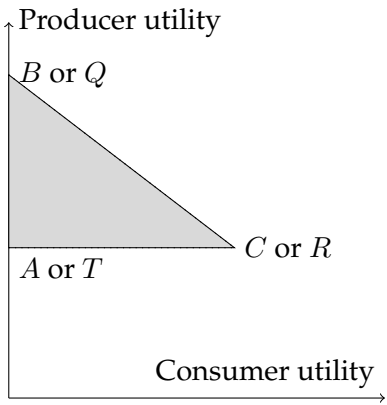

Similarly, Figure \ref{fig:S_Rmk2} illustrates the triangle $ABC$ for the case $\alpha^* \leq \eta$ and for different values of $\eta$ and $\beta$. The primary distinction in this scenario is that, in the non-private case, the optimal uniform price is $v_1$ instead of $v_2$. 

In Figure \ref{fig:S_Rmk2_1}, as $\beta$ decreases from $2\eta$ to $\tilde{\beta}$, $t^*$ increases from $0$ to $\alpha^*$ which implies $C$ moving down towards $\tilde{R}$. In Figure \ref{fig:S_Rmk2_2}, as $\beta$ decreases from $\tilde{\beta}$ to $0$, $t^*$ increases from $\alpha^*$ to $\eta$ which implies $A$ moving up towards $\tilde{T}$. In Figure \ref{fig:S_Rmk2_3}, as $\beta$ decreases from $2(1-\eta)$ to $0$, $t^*$ decreases from $1$ to $\eta$ and hence $A$ moves down from $B$ to $\tilde{T}$.
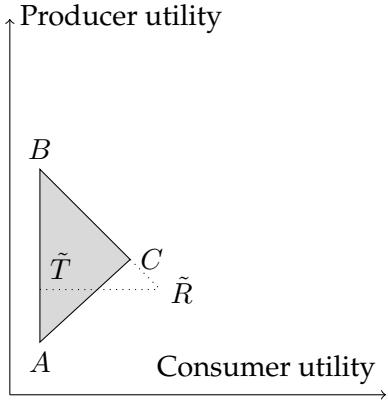
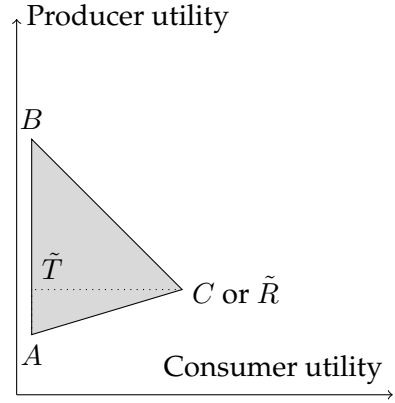
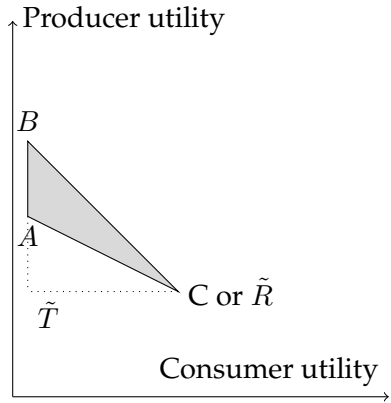
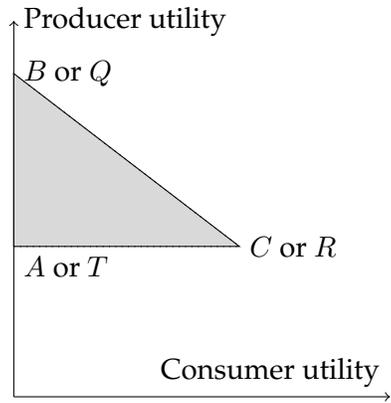
\begin{figure}
\centering
\begin{subfigure}{.5\textwidth}
  \centering
  \begin{tikzpicture}
    \draw[->] (0,0) -- (5,0) node[above left] {Consumer utility}; 
    \draw[->] (0,0) -- (0,5) node[right] {Producer utility}; 
    
    \coordinate (B) at (0.4,3); 
    \coordinate (A) at (0.4, 0.7);
    \coordinate (C) at (1.6,1.8); 
    \coordinate (D) at (2.7,0.7);
    \coordinate (DD) at (0.4, 1.4);
    \coordinate (CC) at (2, 1.4); 

    \draw[fill=gray!30] (A) -- (B) -- (C) -- cycle; 

    \draw[dotted] (C) -- (CC);
    \draw[dotted] (DD) -- (CC);

    \node[below] at (A) {$A$};
    \node[above] at (B) {$B$};
    \node[right] at (C) {$C$};
    \node[above right] at (DD) {$\tilde{T}$};
    \node[right] at (CC) {$\tilde{R}$};
\end{tikzpicture}
  \caption{$\eta \leq 1/2$ and $\beta \geq \tilde{\beta}$}
  \label{fig:S_Rmk2_1}
\end{subfigure}
\begin{subfigure}{.5\textwidth}
  \centering
  \begin{tikzpicture}
    \draw[->] (0,0) -- (5,0) node[above left] {Consumer utility}; 
    \draw[->] (0,0) -- (0,5) node[right] {Producer utility}; 
    
    \coordinate (B) at (0.2,3.4); 
    \coordinate (A) at (0.2,0.8);
    \coordinate (C) at (2.2,1.4); 
    \coordinate (D) at (0.2,1.4);

    \draw[fill=gray!30] (A) -- (B) -- (C) -- cycle; 

    \draw[dotted] (A) -- (D);
    \draw[dotted] (C) -- (D);

    \node[below] at (A) {$A$};
    \node[above] at (B) {$B$};
    \node[right] at (C) {$C$ or $\tilde{R}$};
    \node[above right] at (D) {$\tilde{T}$};
\end{tikzpicture}
  \caption{$\eta \leq 1/2$ and $\beta \leq \tilde{\beta}$}
  \label{fig:S_Rmk2_2}
\end{subfigure}
\newline \vspace{2mm}
\begin{subfigure}{.5\textwidth}
  \centering
  \begin{tikzpicture}
    \draw[->] (0,0) -- (5,0) node[above left] {Consumer utility}; 
    \draw[->] (0,0) -- (0,5) node[right] {Producer utility}; 
    
    \coordinate (B) at (0.2,3.4); 
    \coordinate (A) at (0.2,2.4);
    \coordinate (C) at (2.2,1.4); 
    \coordinate (D) at (0.2,1.4);

    \draw[fill=gray!30] (A) -- (B) -- (C) -- cycle; 

    \draw[dotted] (A) -- (D);
    \draw[dotted] (C) -- (D);

    \node[below] at (A) {$A$};
    \node[above] at (B) {$B$};
    \node[right] at (C) {C or $\tilde{R}$};
    \node[below right] at (D) {$\tilde{T}$};
\end{tikzpicture}
  \caption{$\eta \geq 1/2$}
  \label{fig:S_Rmk2_3}
\end{subfigure}
\begin{subfigure}{.5\textwidth}
  \centering
  \begin{tikzpicture}
    \draw[->] (0,0) -- (5,0) node[above left] {Consumer utility}; 
    \draw[->] (0,0) -- (0,5) node[right] {Producer utility}; 
    
    \coordinate (B) at (0,4.3); 
    \coordinate (A) at (0,2);
    \coordinate (D) at (3,2);

    \draw[fill=gray!30] (A) -- (B) -- (D) -- cycle; 

    \draw[dotted] (A) -- (D);
    \draw[dotted] (B) -- (D);

    \node[below right] at (A) {$A$ or $T$};
    \node[right] at (B) {$B$ or $Q$};
    \node[right] at (D) {$C$ or $R$};
\end{tikzpicture}
  \caption{Non-private case $\beta=0$}
  \label{fig:S_Rmk2_4}
\end{subfigure}
\caption{Illustration of $\mathcal{S}$ for the case $\alpha^* \leq \eta$.}
\label{fig:S_Rmk2}
\end{figure}

\subsection{Does imposing privacy hurt or help the producer and the consumers?} 
We next discuss that imposing privacy always hurts the producer but, interestingly, may hurt or help consumers. We do so by comparing the set $\mathcal{S}$ with the non-private case in Figures \ref{fig:S_Rmk1_4} and \ref{fig:S_Rmk2_4}, studied in \cite{bergemann2015limits}. 

\paragraph{Imposing privacy helps consumers to increase their minimum utility.} In the non-private case, the consumer's utility can be as low as zero. However, in the private case, a minimum utility of $\beta t^* \alpha^* (v_2 - v_1)$ is ensured for the consumer. 

\paragraph{Imposing privacy hurts the producer by decreasing its minimum utility.} In the non-private case, the line $TR$ represents the optimal uniform pricing, which constitutes the minimum utility for the producer. However, in the private case, the producer's utility falls below this value for two reasons. One reason is that the triangle $ABC$ always intersects with the line $\tilde{T}\tilde{R}$, and the producer's utility at line $\tilde{T}\tilde{R}$ is already (weakly) smaller than the utility corresponding to line $TR$. To see the latter, observe that the producer's utility corresponding to line $\tilde{T}\tilde{R}$ is given by
\begin{equation} \label{eqn:line_TR_shift}
\beta \Big(t^* v_1 + (1-t^*) \alpha^* v_2 \Big) + (1-\beta) \Big(\text{ producer utility at line } TR \Big).
\end{equation}
Notice that this term is a convex combination of two components: the first component is an average of $\alpha^* v_2$ and $v_1$, and the second component is the maximum of these two terms. Hence, the utility expressed in the equation above is less than the maximum of $\alpha^* v_2$ and $v_1$, i.e., the producer utility at line $TR$.

Moreover, as demonstrated by Figures \ref{fig:S_Rmk1} and \ref{fig:S_Rmk2}, when $\alpha^* \geq \eta \geq \frac{1}{2}$ or $\alpha^* \leq \eta \leq \frac{1}{2}$, certain segmentations may yield a producer's utility even lower than that at the line $\tilde{T}\tilde{R}$. In the next section, when we explore the case with $K$ values, we precisely characterize when this phenomenon occurs in the general case.

\paragraph{Imposing privacy hurts the producer by decreasing its maximum utility.} In the non-private case, the maximum utility corresponds to first-degree price discrimination (i.e., point $Q$) is given by $\alpha^* v_2 + (1-\alpha^*) v_1$. In the private case, the maximum producer utility aligns with point $B$ in all cases and is equal to
\begin{equation} \label{eqn:producer_max_utility_privacy}
\beta \Big(t^* v_1 + (1-t^*) \alpha^* v_2 \Big) + (1-\beta) \Big(\alpha^* v_2 + (1-\alpha^*) v_1 \Big),
\end{equation}
which is lower than in the non-private case.

\paragraph{Imposing privacy can hurt consumers by decreasing their maximum utility.} For consumers, the maximum utility could either increase or decrease. In the non-private case, the highest utility the consumer can achieve is at point $R$. This could approach zero as $\alpha^*$ nears one. However, as mentioned earlier, the consumers' utility in the private case has a minimum of $\beta t^* \alpha^* (v_2-v_1)$, which does not diminish to zero as $\alpha^*$ increases to one. Thus, the maximum possible utility for consumers could be higher in the private case. To illustrate a potential decrease, consider $\eta \leq 1/2 \leq \alpha^*$, where consumers' maximum utility corresponds to point $C$ in Figure \ref{fig:S_Rmk1_1}. As $\beta$ increases to $2\eta$, $C$ converges to $B$, and $B$ itself moves towards a point on the y-axis, leading to consumers' utility approaching zero.

\subsection{The indirect effect of privacy}
Next, we shift our focus to the shape of $\mathcal{S}$ compared to the triangle $\tilde{T} B \tilde{R}$. As we mentioned earlier, this comparison captures the indirect effect due to the change in the producer's pricing strategy and illustrates how various segmentations can lead to different utilities for the producer and consumer. For example, as discussed before Assumption \ref{assump:beta:intermediate}, when $\beta$ is sufficiently large, $\mathcal{S}$ reduces to a single point, indicating no advantage in price discrimination.

\paragraph{The case $\alpha^* \geq \eta$ :} First, when $\eta \leq 1/2$, indicating that the difference between the high and low values is significant, the set $\mathcal{S}$ loses the $AC\tilde{R}$ portion due to the indirect effect of privacy. In other words, the change in pricing strategy primarily eliminates segmentations beneficial to the consumer while retaining those favorable to the producer. To see why this occurs, notice that the consumer gains utility when their value for the product is high (i.e., $v_2$), but the market is priced low (i.e., $v_1$). However, with a small $\eta$, the threshold $t^*$ decreases, implying that the high value's gain is substantial enough for the producer to risk setting the market price at $v_2$. This shift adversely affects the consumers' utility.

Second, when $\eta \geq 1/2$, as discussed earlier, the threshold $t^*$ becomes larger than $\eta$. In other words, the difference between the high and low values is not significant, and hence, the producer is more inclined to use the lower price $v_1$. This effect intensifies when $\beta$ is very large, as $t^*$ then exceeds both $\alpha^*$ and $\eta$. It implies that even the aggregated market, priced at $v_2$ in the non-private case, is now priced at $v_1$ with a probability of $(1-\beta) + \beta t^*$ (illustrated by point $C$ in Figure \ref{fig:S_Rmk1_2}). Roughly speaking, in this scenario, a segment is more likely to be priced at $v_1$, unless it has a very high proportion of consumers with value $v_2$, which corresponds to point $B$ and nearby points where first-degree price discrimination occurs with two completely separate segments. This is why, in this case, $\mathcal{S}$ predominantly encompasses the area near the line $BC$ and loses out on the area around $A$, which corresponds to segments of a more mixed nature priced at $v_2$. 

Finally, when $\eta \geq 1/2$ but $\beta$ is small, the previously described situation mitigates, and hence the whole triangle $AB\tilde{R}$ is recovered. However, the producer still remains slightly inclined to price markets at $v_1$, which is why points like $C$ fall below the uniform pricing line $A\tilde{R}$. However, as $\beta$ decreases, $C$ moves towards $\tilde{R}$, indicating that Figure \ref{fig:S_Rmk1_3} gradually aligns with Figure \ref{fig:S_Rmk1_4}.

\paragraph{The case $\alpha^* \leq \eta$ :} With $\eta \geq 1/2$, as illustrated in Figure \ref{fig:S_Rmk2_3}, we observe that the triangle $AC\tilde{T}$ is eliminated due to the privacy mechanism. This leads to two intuitions: first, the privacy mechanism benefits the consumer by eliminating segmentations that yield low utility for them. Second, it does not significantly affect segmentations that distinctly separate high and low value consumers. The reason for these phenomena lies in the fact that $\eta \geq 1/2$ suggests a small difference between high and low values, prompting the producer to favor pricing the market at $v_1$ over $v_2$. This inclination minimally impacts segmentations that effectively distinguish between high and low value consumers (represented by point $B$ and its neighboring points). Additionally, incorrectly setting the price to $v_1$ increases the consumer's utility.

With $\eta \leq 1/2$, the producer is more inclined to set the market price at $v_2$. However, the assumption $\alpha^* \leq \eta$ indicates a scarcity of high-value customers, making $v_1$ the optimal uniform price. This is why in Figures \ref{fig:S_Rmk2_1} and \ref{fig:S_Rmk2_2}, we observe points below the uniform pricing line $\tilde{T}\tilde{R}$. This effect is particularly intensified when $\beta$ is high, i.e., Figure \ref{fig:S_Rmk2_1}, as the threshold for setting the price high (i.e., $t^*$) approaches zero. This adversely impacts consumers since consumers gain utility only when they have a high value and the marker is priced at the low value $v_1$.
\subsection{Varying the privacy parameter} \label{section:monotone_beta}
So far, we have shown that implementing the privacy mechanism reduces both the minimum and maximum possible utility of the producer, in contrast to the non-private case, while increasing the minimum utility for the consumer. A natural question arises: do these changes amplify as the privacy parameter increases? For instance, does an increase in the privacy parameter $\beta$ lead to a further decrease in the producer's maximum utility and an increase in the minimum consumer utility? Interestingly, this is not always the case! In other words, while the privacy constraint always reduces the maximum producer utility, more privacy does not necessarily equate to a more significant reduction in this utility. Similarly, although the privacy mechanism ensures a non-zero minimum utility for the consumer, this minimum does not always increase with a higher privacy parameter. We next formalize this observation. 
\paragraph{The impact of increasing the privacy parameter on the producer utility:} 
Let us start with the maximum possible producer utility. Notice that the maximum producer utility corresponds with the point $B$ in Figures \ref{fig:S_Rmk1} and \ref{fig:S_Rmk2}. 
\begin{lemma} \label{lemma:max_producer_nonmonotone}
The maximum producer utility across all segmentations is a decreasing function of $\beta$ over the interval $[0, \min\{2\eta, 2(1-\eta)\}]$ if and only if $\alpha^*, \eta \geq 1/2$ or $\alpha^*, \eta \leq 1/2$.   
\end{lemma}
We defer the proof to the appendix. This result suggests that when $\alpha^* \geq 1/2 \geq \eta$ or when $\alpha^* \leq 1/2 \leq \eta$, we could see an increase in the maximum producer utility by increasing the privacy parameter. 
Next, we provide an intuition for this result and explain why increasing the privacy parameter might actually lead to an increase in the producer's utility in these cases. 

The maximum utility for a producer aligns with first-degree price discrimination, where consumers with identical values are grouped into a single segment. With privacy implemented, each segment is subject to potential alteration with a certain probability, possibly resulting in suboptimal pricing by the producer. As $\beta$ increases, the likelihood of such occurrences also rises.

Consider the scenario where $\eta$ is close to one, i.e., the high and low values are relatively close. In this case, the cost of mistakenly pricing a segment of high-value customers at a lower price is not that severe. However, the suboptimal pricing of a segment of low-value consumers at a higher price can be considerably costly, as the revenue will be completely lost. This especially intensifies when the aggregated market is predominantly composed of low-value consumers, i.e., when $\alpha^*$ is small.

This scenario, where $\eta$ is large and $\alpha^*$ is small, is one of the two cases where the nonmonotonicity occurs. In this context, mistakingly offering a low price to high-value consumers is not detrimental. However, mistakingly offering a high price to low-value consumers can significantly reduce the producer's utility. Why might increasing privacy potentially enhance the producer's maximum utility in this scenario? Because higher $\beta$ induces a more conservative pricing strategy, making the producer more inclined to opt for the lower price. Specifically, for $\eta > 1/2$, the pricing threshold $t^*$ shifts from $\eta$ towards $1$ as $\beta$ increases, implying that producers are less likely to select the higher price as the privacy parameter grows. Consequently, the risk of mistakingly offering a low price to high-value consumers diminishes, potentially boosting the producer's maximum utility. 

The other instance of nonmonotonicity arises when $\eta$ is small and $\alpha^*$ is large. In this scenario, there is a substantial difference between high and low values. In such cases, mistakenly offering a low price to high-value consumers can be particularly detrimental, leading to a significant reduction in maximum producer utility as $\beta$ increases. Consequently, with the increase of $\beta$, and as the observed market becomes less informative, the producer becomes more inclined to take the risk of selecting the higher price. In such circumstances, and when the market is mainly composed of high-value customers, this riskier pricing strategy might lead to an increase in producer utility. 
\paragraph{The impact of increasing the privacy parameter on the consumer utility:}
\begin{lemma} \label{lemma:min_consumer_nonmonotone}
The minimum consumer utility across all segmentations is an increasing function of $\beta$ over the interval the interval $[0, \min\{2\eta, 2(1-\eta)\}]$ if and only if
$\eta \geq 1/2$.     
\end{lemma}
This result suggests that, for any $\eta < 1/2$, the minimum consumer utility exhibits nonmonotonic behavior as $\beta$ varies within the interval $[0, 2\eta]$. For producer utility, as we noted, an increase in $\beta$ when $\eta$ is small leads to a preference for the higher price. Now, it is important to note that the consumer utility is non-zero only when their value is higher than the chosen price, an occurrence that becomes less frequent if the producer is more willing to select the higher price. Consequently, an increase in the privacy parameter could lead to a decrease in the minimum consumer utility in such scenarios.
\section{Price discrimination limits under privacy: the general case}\label{Sec:GeneralK}
Here, we turn our focus to the general case with $K>2$ values. As we establish, once we consider the general case, the fundamental difference between our characterization and that of \cite{bergemann2015limits} without privacy becomes more apparent. In particular, the set of possible consumer and producer utilities is no longer a triangle and instead is a more nuanced polytope that we will characterize. However, we establish that many of the insights that we derived in the special case of $K=2$ continue to hold in the general setting as well.

We make use of the following notation that enables driving the analogue of \cref{lemma:pricing} in the general case. For any $i \in [K]$, let $\mathcal{X}_i^\beta \subseteq \mathcal{X}$ be the set of (observed) markets for which $v_i$ is an optimal price, i.e.,
\begin{equation} \label{eqn:X_i_beta}
\mathcal{X}_i^\beta = \Big \{ \hat{x} ~ \Big \vert ~ \mathcal{U}_p(v_i|\hat{x}) \geq \mathcal{U}_p(v_j|\hat{x}) \text{ for any } j \in [K] \Big \},    
\end{equation}
where $\mathcal{U}_p(\cdot|\hat{x})$ is defined in \eqref{eqn:expected_utility}. Going back to the case $K=2$, $\mathcal{X}_1^\beta$ and $\mathcal{X}_2^\beta$ corresponds to $\alpha \in [0, t^*]$ and $\alpha \in [t^*, 1]$, respectively. The following result establishes the necessary and sufficient condition for $\mathcal{X}_k^\beta$ to be non-empty.
\begin{lemma} \label{lemma:bar_beta}
For any $i \in [K]$, the set $\mathcal{X}_i^\beta$ is non-empty if and only if $\beta \leq \bar{\beta}_i$, where $\bar{\beta}_i$ is the unique solution of the following equation \footnote{$[z]_{+}$ denotes $\max\{z,0\}$.}:
\begin{equation} \label{eqn:bar_beta}
\frac{\bar{\beta}_i}{1-\bar{\beta}_i} :=
\max_{j} \frac{K \Big(v_i- \mathbbm{1}(j<i) v_j \Big)}{\Big[(K+1-j)v_j - (K+1-i)v_i \Big]_{+}}.
\end{equation}
\end{lemma}
Note that $\beta \geq \bar{\beta}_i$ implies that price $v_i$ is not optimal for any observed market, thereby precluding the producer from choosing $v_i$. To better understand this, recall that the expected utility of the producer, given an observed market $\hat{x}$ and upon choosing price $v_i$, denoted as $\mathcal{U}_p(v_i|\hat{x})$, is given by
\begin{equation}
\mathcal{U}_p(v_i|\hat{x}) = (1-\beta) ~ \mathcal{U}_p(v_i, \hat{x}) + \beta ~ \mathbb{E}_{x \sim \text{Unif}(\mathcal{X})} \left[\mathcal{U}_p(v_i, x) \right].
\end{equation}
Essentially, the expected utility of the producer can be viewed as a convex combination of their utility in the fully observed market scenario (analogous to the non-private case) and their utility under complete uncertainty, where they rely on a uniform prior over markets. As $\beta$ increases, indicating a noisier observed market, the weight the producer places on their utility from a uniform prior grows. Consequently, they may opt not to select certain prices that would result in lower utility when the market is drawn from a uniform distribution.

Specifically, for the case when $K=2$, one can verify $\bar{\beta}_1 = 2(1-\eta)$ and $\bar{\beta}_2 = 2\eta$. This is the rationale behind imposing Assumption \ref{assump:beta:intermediate} in the previous section, which guarantees that in the scenario where $K=2$, both prices are viable options. If this were not the case, only one pricing option would remain, causing the set $\mathcal{S}$ to reduce to a single point.

Our next result extends Proposition \ref{proposition:S_to_S'} to this general case. 
\begin{proposition} \label{proposition:S_to_S'_K}
The set $\mathcal{S}$ can be represented as
\begin{equation} \label{eqn:S_to_S'_K}
S = \beta \bm{c} + (1-\beta) \mathcal{S}',
\end{equation}
where $\bm{c} \in \mathbb{R}^2$ is a constant vector, given by
\begin{equation} \label{eqn:S_shift_K}
\bm{c} = \left [ \sum_{k=1}^K \mathbb{P}(\mathcal{X}_k^\beta) ~\mathcal{U}_c(v_k, x^*), \sum_{k=1}^K \mathbb{P}(\mathcal{X}_k^\beta) ~\mathcal{U}_p(v_k, x^*)\right]^\top,
\end{equation}
where $\mathbb{P}(\cdot)$ denotes the uniform distribution over $\mathcal{X}$, and $\mathcal{U}_c(\cdot,\cdot)$ and $\mathcal{U}_p(\cdot,\cdot)$ are as defined in \eqref{eqn:utilities}. Moreover, the set $\mathcal{S}'$ is given by
\begin{align} \label{eqn:S'_K}
\mathcal{S}' = \left \{ \Big[ \sum_{k=1}^K \gamma_k ~\mathcal{U}_c(v_k, x_k), \sum_{k=1}^K \gamma_k ~\mathcal{U}_p(v_k, x_k)  \Big]^\top
~ \Big \vert ~ 
x_k \in \mathcal{X}_k^\beta, \gamma_k \in [0,1],  \sum_{k=1}^K \gamma_k x_k = x^*, \sum_{k=1}^K \gamma_k = 1
\right \}.
\end{align}
\end{proposition}
This result shows that, in the general case, and similar to the case when $K=2$, the set $\mathcal{S}$, which depicts the limits of price discrimination under privacy mechanisms, is influenced by three factors: a constant shift $\beta \bm{c}$, a scaling factor $1-\beta$, and the set $\mathcal{S}'$ that determines its shape. Consequently, the insights derived for $K=2$ based on this characterization extend to the general case. In particular, the shift $\beta \bm{c}$ indicates that, in contrast to the non-private case, the consumers' minimum utility is not zero. Moreover, the scaling factor $1-\beta$ suggests that market segmentation generally becomes less impactful as the privacy factor $\beta$ increases.

Our primary goal now is to specify the set $\mathcal{S}'$. As we discussed in the case of $K=2$, the set $\mathcal{S}'$ captures the indirect impact of privacy, as it illustrates the effects of the change in the producer's pricing strategy, given that they know the privacy mechanism is applied. More specifically, $\mathcal{S}'$ is the set of potential pairs of utilities for the consumer and producer when the producer accurately perceives the segmented markets (as in the non-private case) but adopts the pricing rule from the private case (associated with sets $\{\mathcal{X}_k^\beta\}_k$) rather than the optimal pricing rule for the non-private scenario. In fact, in the non-private case with $\beta=0$, this set aligns with $\mathcal{S}$. 

Proposition \ref{proposition:S_to_S'_K} also suggests that our analysis can be limited to segmentations where, for any value $v_k$, there is at most one segment whose optimal corresponding price is $v_k$. This effectively broadens the insights of Proposition \ref{proposition:S_to_S'} from the case of $K=2$ and simplifies the characterization of $\mathcal{S}'$.
As we stated in the previous section, the set $\mathcal{S}'$ is a triangle for $K=2$. In general, we next establish that it is a convex polygon.
\begin{theorem} \label{theorem:S'_polytope}
Let $\mathcal{P}$ denote a convex polytope in $\mathbb{R}^{K^2}$, represented by the variables $\{z(i,j)\}_{i,j=1}^K$ and the following equations:
\begin{subequations} \label{eqn:polytope}
\begin{align}
& z(i,1) \geq \cdots \geq z(i,K) \geq 0 \text{ for all } i \in [K], \label{eqn:polytope_a} \\
& z(i,i) v_i - z(i,j) v_j \geq \frac{\beta}{K(1-\beta)} z(i,1)\Big ( (K+1-j) v_j - (K+1-i) v_i \Big ) \text{ for all } i,j \in [K], \label{eqn:polytope_b} \\
& \sum_{i=1}^K z(i,j) = \sum_{k=j}^K x^*(v_k) \text{ for all } j \in [K]. \label{eqn:polytope_c}
\end{align}
\end{subequations}
Then, the set $\mathcal{S}'$, given in \eqref{eqn:S'_K}, can be cast as a linear transformation of $\mathcal{P}$ from $\mathbb{R}^{K^2}$ to $\mathbb{R}^2$, given by
\begin{equation} \label{eqn:utilities_linear_mapping}
\mathcal{S}'= \left \{ 
\sum_{i=1}^{K-1} \sum_{j=i+1}^{K} (v_{j}-v_{j-1})z(i,j), ~ \sum_{i=1}^K v_i z(i,i)
\right \}.
\end{equation}
\end{theorem}
\noindent \textbf{\textit{Proof sketch:}} To understand how this result is established, recall the definition of $\mathcal{S}'$, which involves $K$ markets $x_1, \cdots, x_K$, with each market $x_i \in \mathcal{X}_i^\beta$. This implies that $v_i$ is an optimal price for $x_i$. Now, let $y(i,\cdot)$ represent the complementary cumulative distribution function of market $x_i$, defined as:
\begin{equation*}
y(i,j) := \sum_{k=j}^K x_i(v_k).
\end{equation*}
It is evident that:
\begin{equation*}
1= y(i,1) \geq \cdots \geq y(i,K) \geq 0 \text{ for all } i \in [K].
\end{equation*}
Furthermore, we verify that $v_i$ being the optimal price for $x_i$ necessitates:
\begin{equation} \label{eqn:optimal_pricing_x_i}
y(i,i) v_i - y(i,j) v_j \geq \frac{\beta}{K(1-\beta)} \Big ( (K+1-j) v_j - (K+1-i) v_i \Big ) \text{ for all } i,j \in [K].
\end{equation}
Therefore, conditions \eqref{eqn:polytope_a} and \eqref{eqn:polytope_b} are satisfied if we substitute $z(\cdot, \cdot)$ with $y(\cdot, \cdot)$. Additionally, the aggregated market imposes the condition $\sum_{k=1}^K \gamma_k x_k = x^*$. Using $y(\cdot, \cdot)$, this condition can be represented as:
\begin{equation}
\sum_{i=1}^K \gamma_i y(i,j) = \sum_{k=j}^K x^*(v_k) \text{ for all } j \in [K].
\end{equation}
By defining $z(i,j) = \gamma_i y(i,j)$, we immediately see that \eqref{eqn:polytope_c} is valid. The other two conditions, \eqref{eqn:polytope_a} and \eqref{eqn:polytope_b}, are also satisfied since all $y(i,\cdot)$ terms are multiplied by the same variable $\gamma_i$. Therefore \eqref{eqn:polytope} provide a representation for the space of segmentations in the definition $\mathcal{S}'$, i.e., $(\gamma_k)_{k=1}^K$ and $(x_k)_{k=1}^K$, in form of a convex polytope. It remains to show that the linear mapping \eqref{eqn:utilities_linear_mapping} is equal to the consumer and producer utilities. We defer this part to the appendix. $\square$

We should highlight that the proof outlined above provides a construction for a segmentation to achieve any feasible point for the consumer and producer utilities.

Before proceeding, let us highlight that in the special of $\beta=0$, the polytope characterized in \cref{theorem:S'_polytope} simplifies to a triangle as characterized in \cite{bergemann2015limits}. However, in our setting, with privacy, it takes a more nuanced form. Let us present an example illustrating some potential shapes of $\mathcal{S}'$. This example will serve as a basis to introduce and motivate several results concerning $\mathcal{S}'$ and, consequently, the set $\mathcal{S}$.
\begin{example} \label{example:S'_K_5}
\textup{
Consider the case $K=5$ with values $[v_k]_{k=1}^5 = [0.8, 2, 3, 4.2, 5]$ and $\beta = 0.3$. Figure \ref{fig:S'_example} illustrates $\mathcal{S'}$ for three distinct aggregated market values $x^*$. Specifically, the aggregated markets associated with Figures \ref{fig:S'_example_1}, \ref{fig:S'_example_2}, and \ref{fig:S'_example_3} are $(0.2,0.1,0.4,0.2,0.1)$, $(0.2,0.3,0.2,0.2,0.1)$, and $(0.2,0.1,0.1,0.05,0.55)$, respectively. The dashed triangle $TQR$ in the figures represents the non-private case, in which, as previously stated, the set $\mathcal{S}'$ coincides with the set $\mathcal{S}$.
}
\begin{figure}
\centering
\begin{subfigure}{.33\textwidth}
  \centering
  \includegraphics[width=1.1\linewidth]{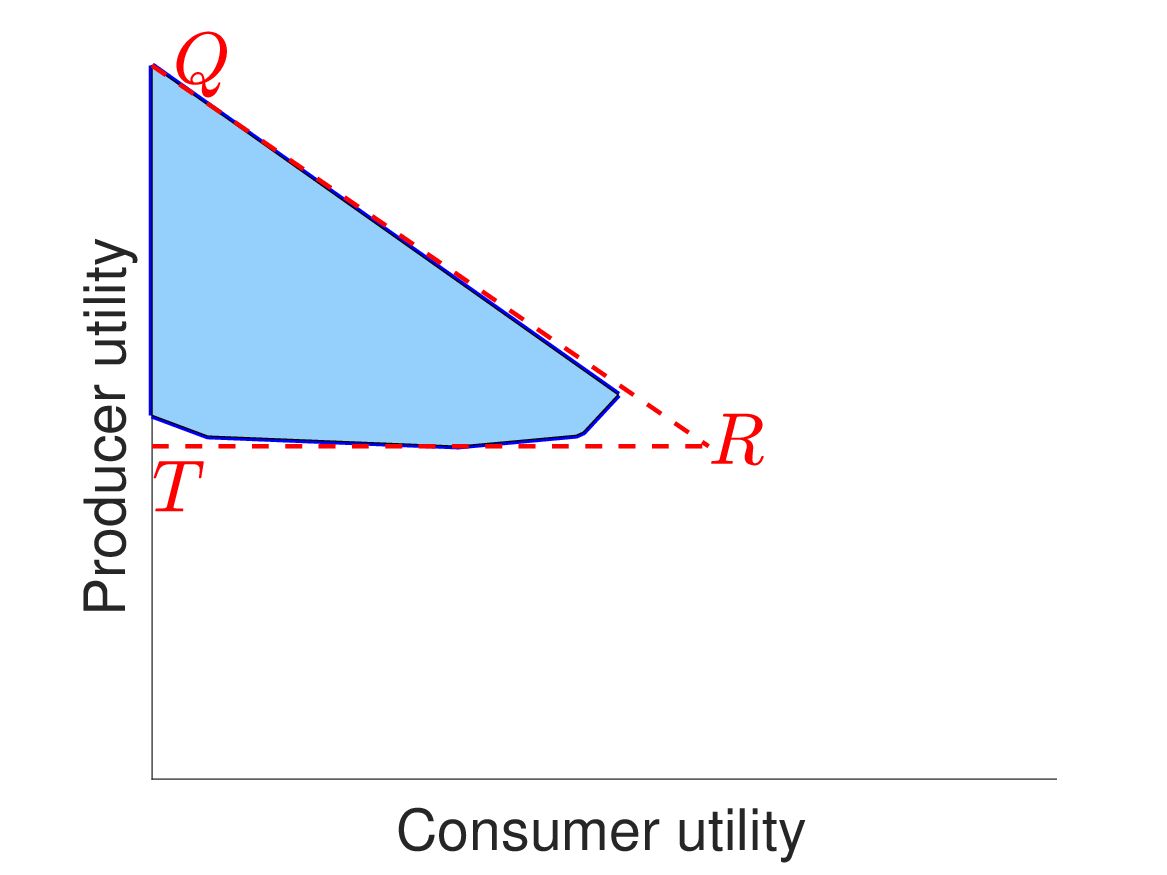}
  \caption{$x^*=(0.2,0.1,0.4,0.2,0.1)$}
  \label{fig:S'_example_1} 
\end{subfigure}
\begin{subfigure}{.33\textwidth}
  \centering
  \includegraphics[width=1.1\linewidth]{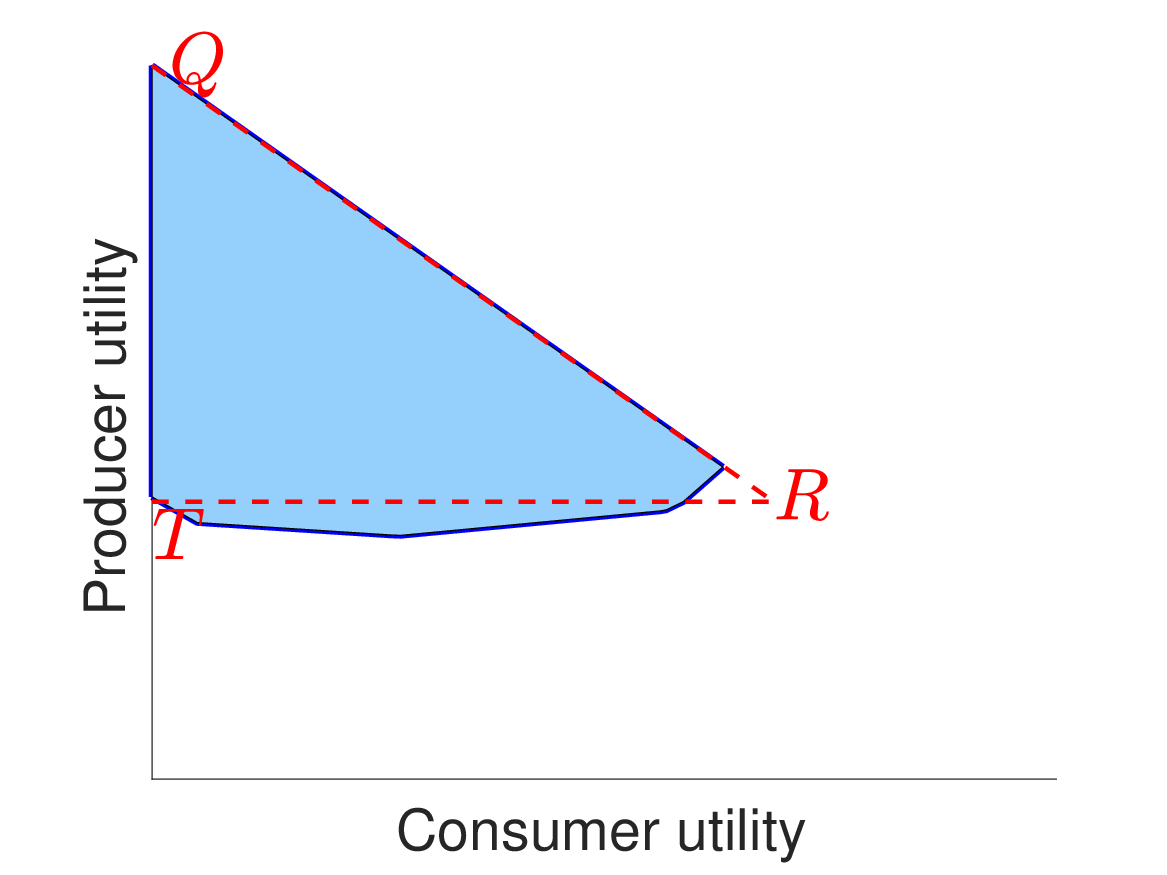}
  \caption{$x^* = (0.2,0.3,0.2,0.2,0.1)$}
  \label{fig:S'_example_2}
\end{subfigure}
\begin{subfigure}{.33\textwidth}
  \centering
  \includegraphics[width=1.1\linewidth]{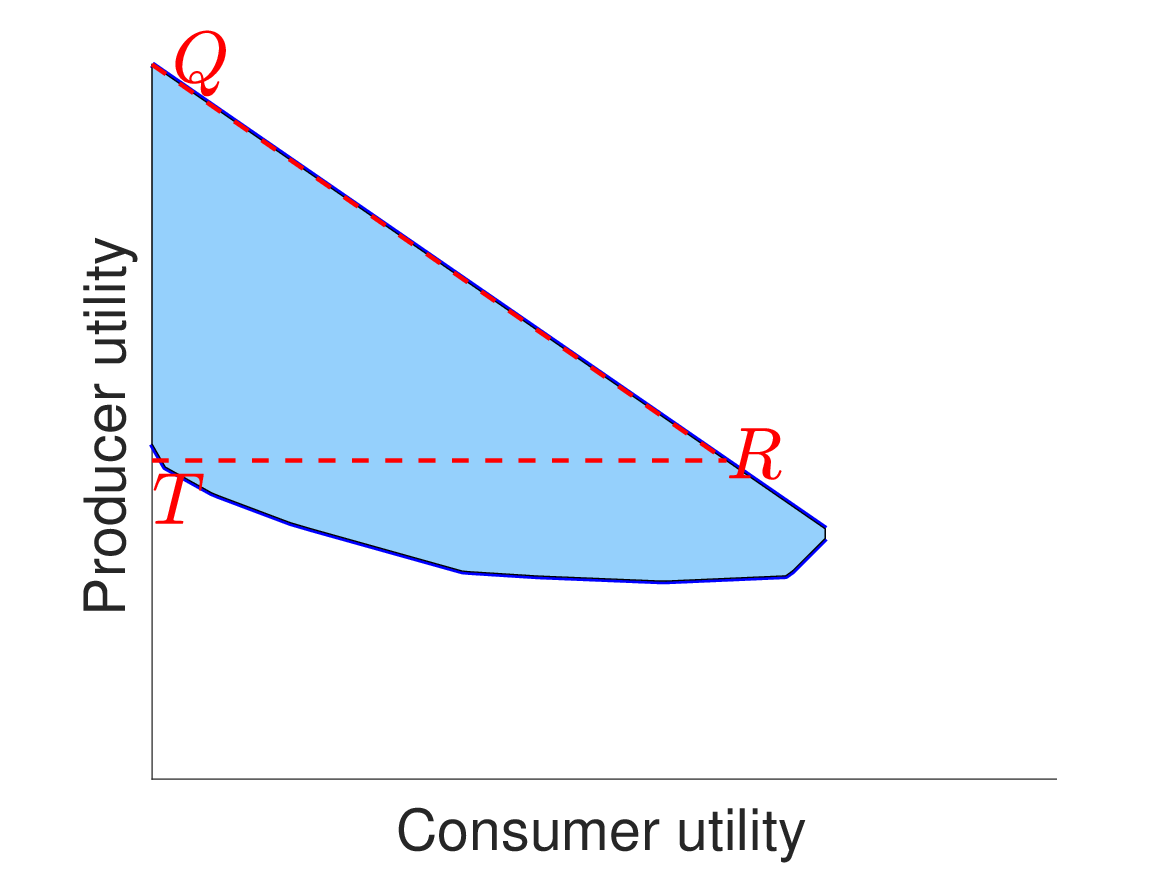}
  \caption{$x^*=(0.2,0.1,0.1,0.05,0.55)$}
  \label{fig:S'_example_3}
\end{subfigure}
\caption{Illustration of $\mathcal{S'}$ with values $[v_k]_{k=1}^5 = [0.8, 2, 3, 4.2, 5]$, $\beta=0.3$, and three different aggregated markets.}
\label{fig:S'_example}
\end{figure}
\end{example}
Notice that the point $Q$ in Figure \ref{fig:S'_example} corresponds to the first-degree price discrimination and is given by
\begin{equation*}
\left(0, \sum_{k=1}^K x^*(v_k) v_k\right).    
\end{equation*}
We next list a number of results regarding $\mathcal{S}$ and $\mathcal{S}'$. The first one is a straightforward consequence of \cref{theorem:S'_polytope}:
\begin{corollary}
The set $\mathcal{S}'$ forms a polygon and is situated to the right of the y-axis and to the left of the line passing through points $Q$ and $R$, which represents the total surplus equal to the maximum possible surplus, i.e., $\sum_{k=1}^K x^*(v_k) v_k$.
\end{corollary}
The figures in Figure~\ref{fig:S'_example} suggest that the set $\mathcal{S}'$ includes the point $Q$. This observation holds true when $\beta$ is sufficiently small, ensuring that all values within the support of $x^*$ remain viable as optimal prices. In such scenarios, a segmentation that groups all consumers with the same value into the same segment can attain point $Q$. However, if $\beta$ is so large that, for some $k$, the value $v_k$ is never selected by the producer while $x^*(v_k) > 0$, the producer cannot achieve the maximum utility as the consumers with value $v_k$ will never be priced at $v_k$. We formalize this observation in the following corollary:
\begin{corollary}
The set $\mathcal{S}'$ includes the point $Q$ if and only if
\begin{equation*}
\beta \leq \min_{k:x^*(v_k) > 0} \bar{\beta}_k.     
\end{equation*}
\end{corollary}
Returning to Example \ref{example:S'_K_5}, we calculate that $(\bar{\beta}_k)_{k=1}^5$ equals $(0.44, 0.91, 1, 0.91, 0.54)$. Therefore, with $\beta=0.5$, for instance, one of the values ($v_1$) becomes infeasible as optimal market prices. Figure \ref{fig:S'_example2} displays $\mathcal{S}'$ using the parameters from Example \ref{example:S'_K_5}, but with $\beta=0.5$ instead of $\beta=0.3$. As anticipated, the point $Q$ is no longer included in the set $\mathcal{S}'$.
\begin{figure}
\centering
\begin{subfigure}{.33\textwidth}
  \centering
  \includegraphics[width=1.1\linewidth]{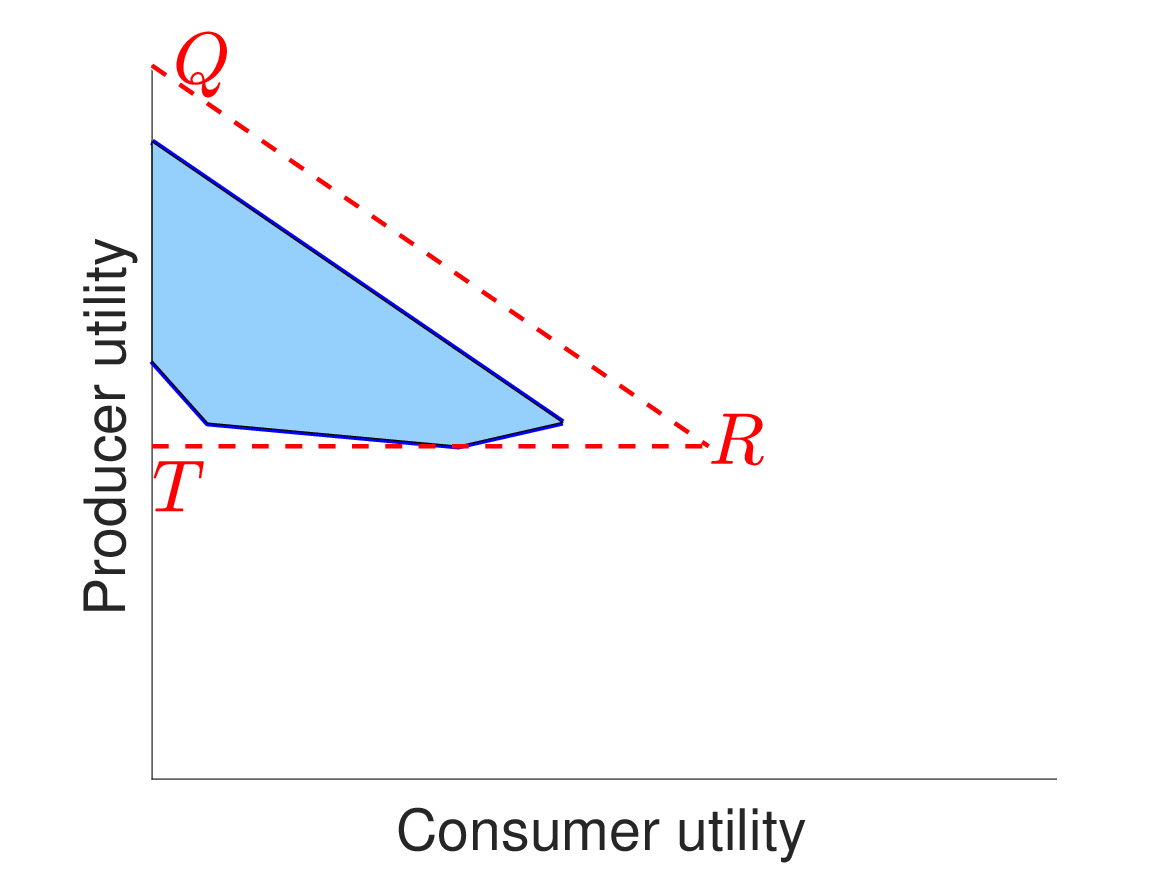}
  \caption{$x^*=(0.2,0.1,0.4,0.2,0.1)$}
  \label{fig:S'_example2_1} 
\end{subfigure}
\begin{subfigure}{.33\textwidth}
  \centering
  \includegraphics[width=1.1\linewidth]{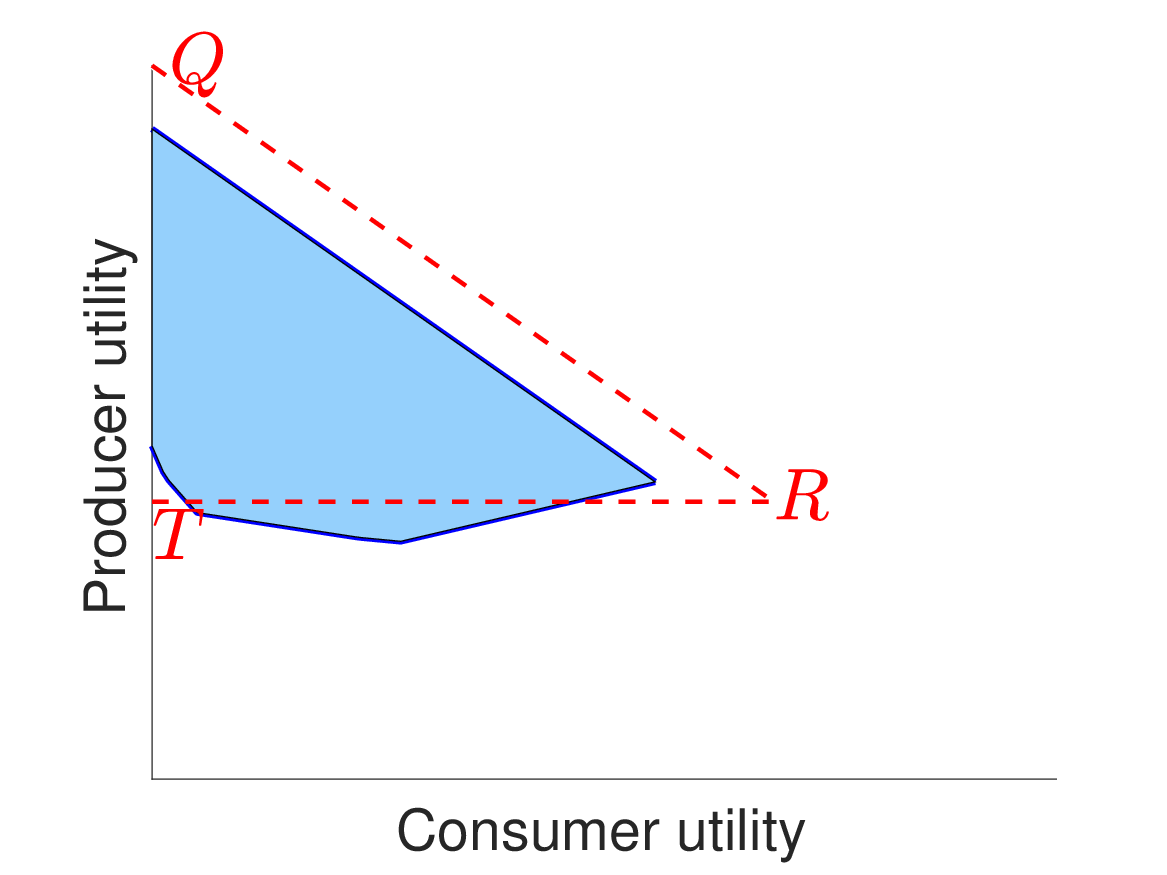}
  \caption{$x^* = (0.2,0.3,0.2,0.2,0.1)$}
  \label{fig:S'_example2_2}
\end{subfigure}
\begin{subfigure}{.33\textwidth}
  \centering
  \includegraphics[width=1.1\linewidth]{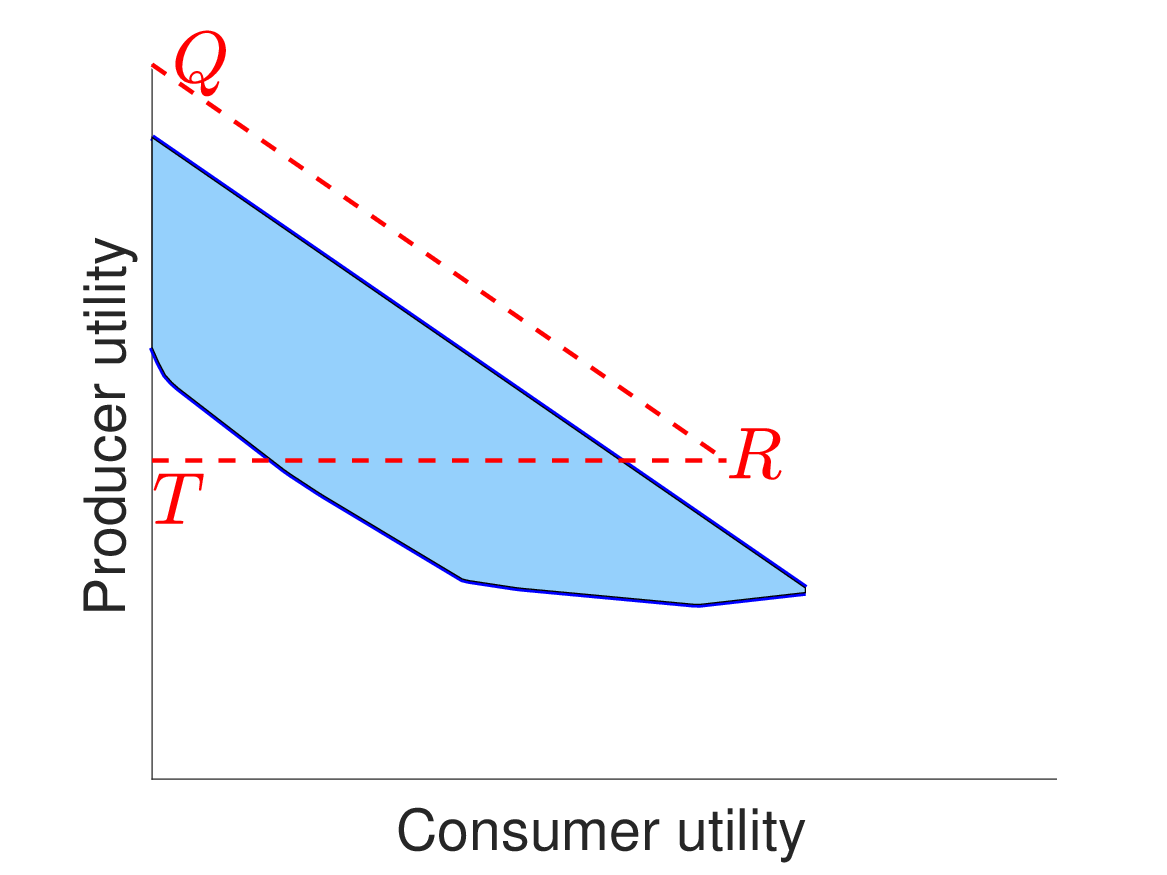}
  \caption{$x^*=(0.2,0.1,0.1,0.05,0.55)$}
  \label{fig:S'_example2_3}
\end{subfigure}
\caption{Illustration of $\mathcal{S'}$ with values $[v_k]_{k=1}^5 = [0.8, 2, 3, 4.2, 5]$, $\beta=0.5$, and three different aggregated markets.}
\label{fig:S'_example2}
\end{figure}

Now, let us consider what these two facts imply for the set $\mathcal{S}$, i.e., the constraints on price discrimination under privacy considerations.
\begin{fact}
The maximum possible utility of the producer from market segmentation is always weakly smaller when a privacy mechanism is applied compared to the non-private case.
\end{fact}
To understand this, note that the maximum utility of the producer in the non-private case corresponds to point $Q$, which is expressed by
\begin{equation} \label{eqn:max_producer_non_private}
\sum_{k=1}^K x^*(v_k) v_k.
\end{equation}
However, in the private case, as Proposition \ref{proposition:S_to_S'_K} demonstrates, the utility is bounded by
\begin{equation} \label{eqn:max_producer_private}
\beta \left ( \max_{v_k}~ \mathcal{U}_p(v_k, x^*) \right) + (1-\beta) \left ( \text{maximum producer utility in } \mathcal{S}' \right ).
\end{equation}
This term is a convex combination of two components, where the second component is at most equal to \eqref{eqn:max_producer_non_private} (when $Q$ is included in $\mathcal{S}'$), and the first component corresponds to the optimal uniform pricing (line $TR$) and is weakly smaller than \eqref{eqn:max_producer_non_private}. This confirms the reduction in the maximum producer utility under privacy. Furthermore, when $\beta$ is sufficiently large such that $Q$ is not included in $\mathcal{S}'$, the second component in \eqref{eqn:max_producer_private} also becomes strictly smaller than \eqref{eqn:max_producer_non_private}.

Now, let us focus on the impact of the privacy mechanism on the minimum producer utility. In the non-private case, the minimum producer's utility is represented by the line $TR$. To compare this with the private case, we need to examine the lower boundary of $\mathcal{S}'$. Observing Figures \ref{fig:S'_example} and \ref{fig:S'_example2}, it appears that $\mathcal{S}'$ consistently intersects with line $TR$ and sometimes even crosses this line. We formalize this observation in the following result:
\begin{proposition} \label{proposition:minimum_S'}
Suppose $\beta>0$ and $x^*(v_k) > 0$ for all $k\in[K]$. The lowest point of the set $\mathcal{S}'$ lies on or beneath the line that represents optimal uniform pricing in the non-private case, i.e., line $TR$. Moreover, it is below this line if and only if the optimal pricing for the non-private case and the fully private case differ, i.e., there exists $i$ such that
\begin{equation} \label{eqn:crossing_condition}
\frac{(K+1-i)}{K} v_i > \frac{(K+1-j)}{K} v_j 
\quad \text{ for all } \,
j \in \argmax_{k \in [K]}  \mathcal{U}_p(v_k, x^*).
\end{equation}
\end{proposition}
We defer the proof of this proposition to the appendix. Revisiting Example \ref{example:S'_K_5}, we find that while the condition \eqref{eqn:crossing_condition} is not met for $x^*$ corresponding to Figure \ref{fig:S'_example_1}, it is fulfilled for the aggregated markets associated with Figures \ref{fig:S'_example_2} and \ref{fig:S'_example_3}. Consequently, the set $\mathcal{S}'$ intersects with line $TR$ in the former case without crossing it, whereas in the latter two cases, it extends below this line.

For the case $K=2$, the condition \eqref{eqn:crossing_condition} translates into either $\alpha^* \geq \eta \geq 1/2$ or $\alpha^* \leq \eta \leq 1/2$. The former condition, $\alpha^* \geq \eta \geq 1/2$, corresponds to Figures \ref{fig:S_Rmk1_2} and \ref{fig:S_Rmk1_3}. It is important to note that these figures depict the set $\mathcal{S}$, which is derived from the set $\mathcal{S}'$ after applying the scaling and shift described in Proposition \ref{proposition:S_to_S'_K}. Specifically, line $TR$ in $\mathcal{S}'$ corresponds to line $\tilde{T}\tilde{R}$ in these figures. Consequently, the lowest point of $\mathcal{S}'$ maps to point $C$ that is located below the line $\tilde{T}\tilde{R}$ in both cases. In contrast, for Figure \ref{fig:S_Rmk1_1}, where the condition \eqref{eqn:crossing_condition} is not met, the lowest point of $\mathcal{S}'$ maps to point $A$, which lies precisely on the line $\tilde{T}\tilde{R}$. Similarly, for the case $\alpha^* \leq \eta \leq 1/2$, as illustrated in Figures \ref{fig:S_Rmk2_1} and \ref{fig:S_Rmk2_2}, the lowest point of $\mathcal{S}'$ maps to point $A$  which is situated below the line $\tilde{T}\tilde{R}$.

It is worth emphasizing that condition \eqref{eqn:crossing_condition} can indicate the ambiguity that a privacy mechanism introduces into the producer's pricing strategy. Note that when this condition is not met, it implies that the optimal pricing strategy remains constant, regardless of the privacy level. This makes the privacy mechanism less detrimental to the producer's utility. Conversely, when this condition is satisfied, it indicates that the optimal strategy varies with the level of privacy, thereby diminishing the producer's utility due to suboptimal pricing decisions.

\begin{corollary} \label{corollary:minimum_producer}
The minimum utility of the producer across all segmentations is always weakly smaller when a privacy mechanism is applied compared to the non-private case.
\end{corollary}
To see why this result holds, notice that, according to Proposition \ref{proposition:S_to_S'_K}, the minimum utility of the producer is given by
\begin{equation} \label{eqn:min_producer_private}
\beta \left ( \sum_{k=1}^K \mathbb{P}(\mathcal{X}_k^\beta) ~\mathcal{U}_p(v_k, x^*) \right) + (1-\beta) \left ( \text{minimum producer utility in } \mathcal{S}' \right ),
\end{equation}
which is a convex combination of two terms. The first one is upper bounded by $ \max_{v_k}~ \mathcal{U}_p(v_k, x^*)$, i.e., the minimum utility of the producer in the non-private case (corresponding to line $TR$).  \cref{proposition:minimum_S'} implies that the minimum producer utility in $\mathcal{S}'$ is (weakly) smaller than this term which establishes \cref{corollary:minimum_producer}.

We now shift our attention to the consumer's utility. As previously discussed following Proposition \ref{proposition:S_to_S'_K}, the privacy mechanism guarantees a non-zero minimum utility for the consumer, a stark contrast to the non-private case where the consumer's minimum utility can be zero. In the following result, we explore another phenomenon that contributes to consumer's minimum utility.
\begin{fact} \label{fact:distance_y_axis}
The set $\mathcal{S}'$ is distanced from the y-axis by a minimum of $x^*(v_K) (v_K - v_{K-1})$ when $\beta > \bar{\beta}_K$.
\end{fact}
The main point to note is that for sufficiently large values of $\beta$, the set $\mathcal{S}'$ has a positive distance from the y-axis. This results in a guaranteed additional positive minimum utility term for the consumer within the set $\mathcal{S}$. An illustration of this phenomenon is provided in Figure \ref{fig:S'_example3}, which displays the set $\mathcal{S}'$ for the values specified in Example \ref{example:S'_K_5}, but with $\beta=0.6$. As stated earlier, in Example \ref{example:S'_K_5}, $\bar{\beta}_5$ equals $0.54$.
\begin{figure}
\centering
\begin{subfigure}{.33\textwidth}
  \centering
  \includegraphics[width=1.1\linewidth]{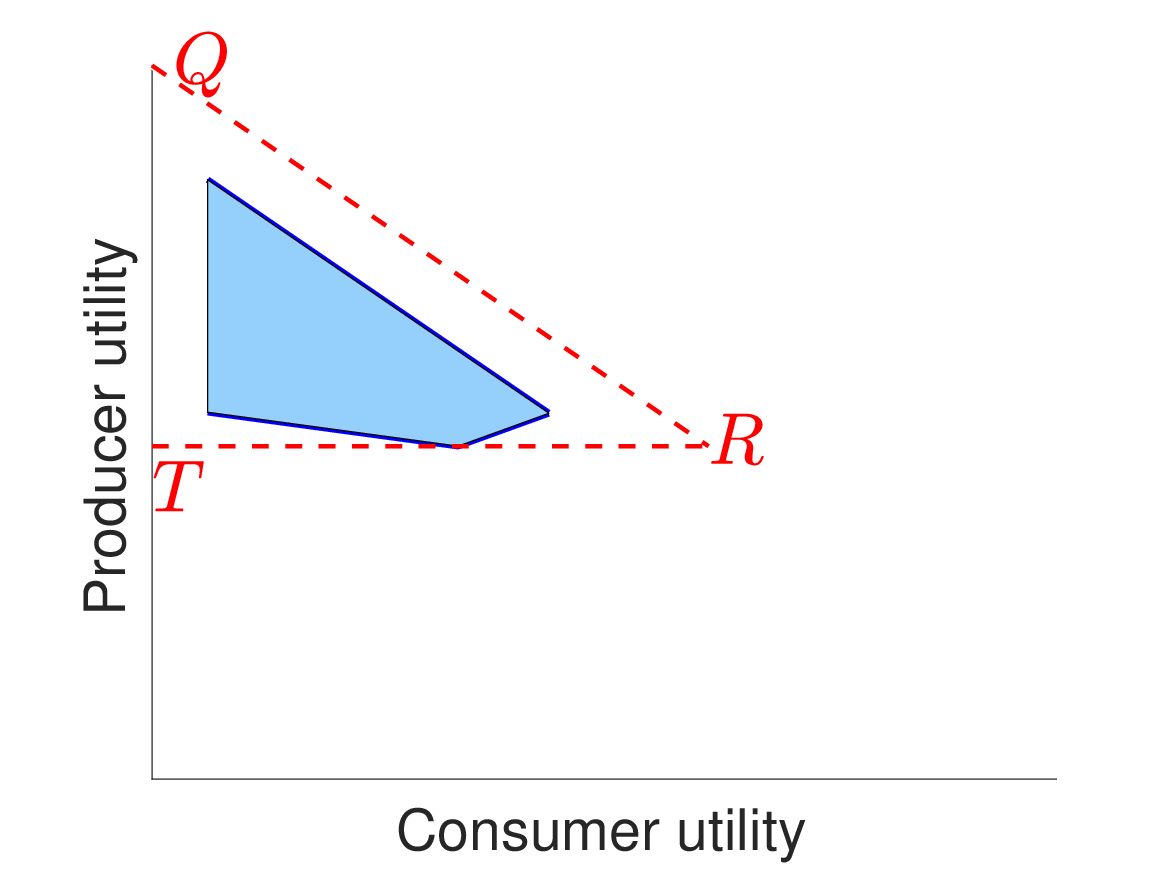}
  \caption{$x^*=(0.2,0.1,0.4,0.2,0.1)$}
  \label{fig:S'_example3_1} 
\end{subfigure}
\begin{subfigure}{.33\textwidth}
  \centering
  \includegraphics[width=1.1\linewidth]{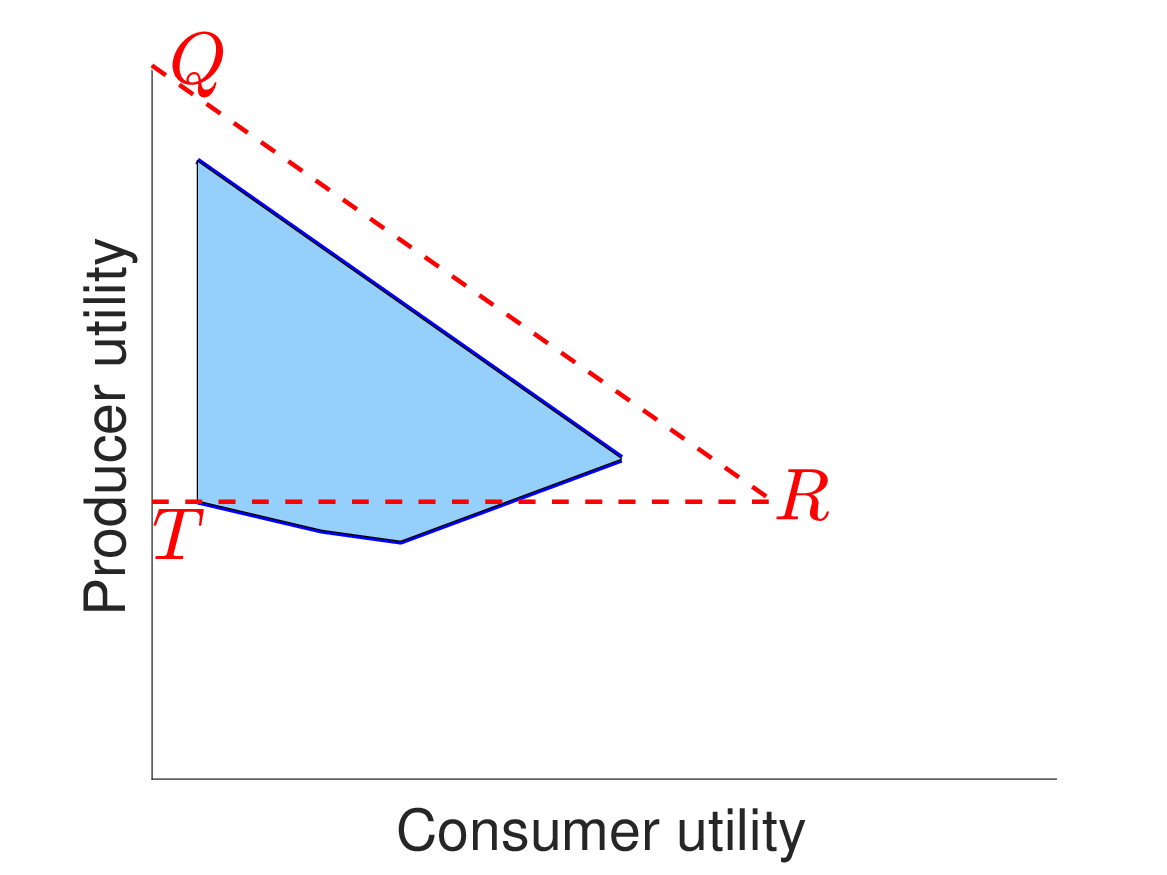}
  \caption{$x^* = (0.2,0.3,0.2,0.2,0.1)$}
  \label{fig:S'_example3_2}
\end{subfigure}
\begin{subfigure}{.33\textwidth}
  \centering
  \includegraphics[width=1.1\linewidth]{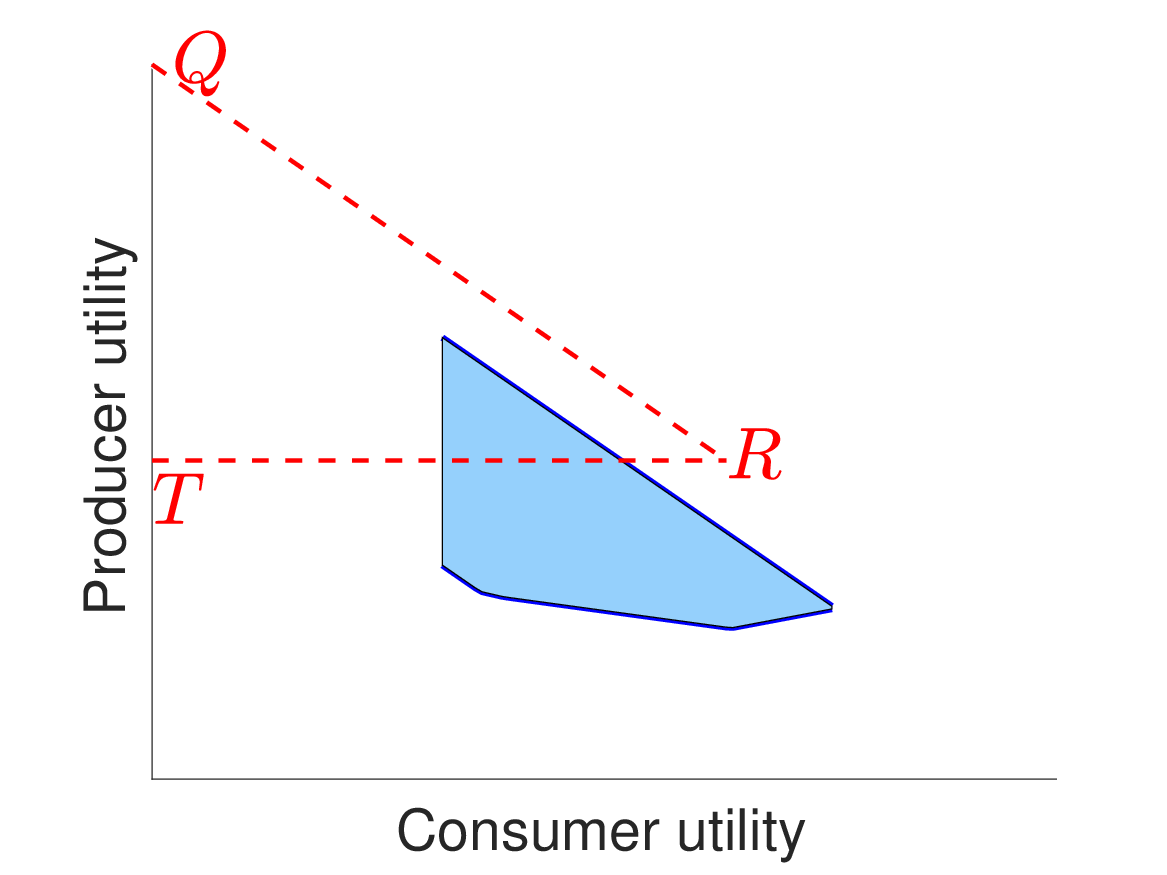}
  \caption{$x^*=(0.2,0.1,0.1,0.05,0.55)$}
  \label{fig:S'_example3_3}
\end{subfigure}
\caption{Illustration of $\mathcal{S'}$ with values $[v_k]_{k=1}^5 = [0.8, 2, 3, 4.2, 5]$, $\beta=0.6$, and three different aggregated markets.}
\label{fig:S'_example3}
\end{figure}

Finally, as we discussed following Theorem \ref{theorem:S_K=2} through examples with $K=2$, the maximum utility for consumers, in general, could either increase or decrease under the privacy mechanism. 
The following result better highlights why both cases are possible.
\begin{proposition} \label{proposition:consumer_max}
Suppose $\beta < \min_{k} \bar{\beta}_k$. 
\begin{enumerate}[(i)]
\item Assume $v_1$ is the optimal uniform price in the non-private case, and the values satisfy the condition:
\begin{equation} \label{eqn:increasing_uniform}
v_{i+1} \geq \frac{K-i}{K+1-i}v_i \text{ for all } i\in[K-1].
\end{equation}
Then, the maximum utility for the consumer is a decreasing function of $\beta$.
\item Assume $v_K$ is the optimal uniform price in the non-private case, and the values satisfy the condition:
\begin{equation} \label{eqn:decreasing_uniform}
v_{i+1} \leq \frac{K-i}{K+1-i}v_i \text{ for all } i\in[K-1].
\end{equation}
Then, there exists $\tilde{\alpha}$ and $\tilde{\beta} > 0$ such that the maximum utility for the consumer is an increasing function over the interval $[0, \tilde{\beta}]$ when $x^*(v_K) \geq \tilde{\alpha}$.
\end{enumerate}    
\end{proposition}
The proof is provided in the appendix. In particular, for the case when $K=2$, the conditions outlined in part (i), which lead to a decrease in maximum consumer utility under privacy, translate to $\alpha^* \leq \eta \leq 1/2$. On the other hand, the conditions of part (ii), which implies an increase in maximum consumer utility in the presence of the privacy mechanism, correspond to $\alpha^* \geq \eta \geq 1/2$ (also, as established in the proof, we can set $\tilde{\alpha} = \frac{1}{2-\eta}$ in this case).

But what differentiates these two cases? Why does privacy help consumer utility in one scenario while hurting it in another? The rationale here is similar to the explanations provided in Section \ref{section:monotone_beta}. The condition \eqref{eqn:increasing_uniform} indicates that under a uniform market prior, the highest price $v_K$ leads to the maximum expected utility for the producer, while the lowest price $v_1$ yields the minimum. Consequently, when privacy introduces ambiguity, the producer has more incentive to prefer higher prices. However, the assumption that $v_1$ is the optimal uniform price suggests that the market is mainly composed of low-value customers. Therefore, a tendency to set higher prices due to privacy considerations adversely affects consumer utility. In other words, privacy prompts the producer to risk setting higher prices, resulting in diminished consumer utility.

Conversely, the condition \eqref{eqn:decreasing_uniform} means that under a uniform market prior, the lowest price $v_1$ leads to the highest expected utility for the producer, while the highest price $v_K$ offers the least. Thus, with privacy in play, the producer is likely to adopt a more conservative stance and favor lower prices. That said, given that $v_K$ is the optimal uniform price and $x^*(v_K) \geq \tilde{\alpha}$, we have a market with a majority of high-value customers who benefit more when the market price is set lower. In summary, applying the privacy mechanism makes the producer to take a more conservative pricing approach, resulting in more consumers purchasing the product at a lower price than their value.
\section{Conclusion}\label{Sec:Conclusion}
This paper explores the intersection of price discrimination and consumer privacy. Building upon the framework of \cite{bergemann2015limits}, we first introduce a privacy mechanism that adds a layer of uncertainty regarding consumer values, thereby directly impacting both consumer and producer utilities in a market setting. We use a novel analysis based on a \emph{merging technique} to characterize the set of all pairs of consumer and producer utilities and establish that, unlike the existing results, it is not necessarily a triangle. Instead, it can be viewed as the linear mapping of a high-dimensional polytope into $\mathbb{R}^2$. We used our characterization to derive several insights. In particular, we prove that imposing a privacy constraint does not necessarily help the consumers and, in fact, can reduce their utility. 

We view our paper as a first attempt to understand how user privacy impacts the limits of price discrimination. There are several exciting directions to explore in this vein. For instance, an interesting future direction is to extend our results to a setting with a multiproduct seller (similar to that of \cite{haghpanah2022limits}). 
\section{Acknowledgments}
Alireza Fallah acknowledges support from the European Research Council Synergy Program, the National Science Foundation under grant number DMS-1928930, and the Alfred P. Sloan Foundation under grant G-2021-16778. The latter two grants correspond to his residency at the Simons Laufer Mathematical Sciences Institute (formerly known as MSRI) in Berkeley, California, during the Fall 2023 semester. Michael Jordan acknowledges support from the Mathematical Data Science program of the Office of Naval Research under grant number N00014-21-1-2840 and the European Research Council Synergy Program.
\bibliography{References}
\newpage
\appendix
\section{Deferred Proofs}
\subsection{Proof of \cref{Pro:privacy:leakage}}
In our case, the prior distribution is uniform, and the posterior distribution after observing the output of the privacy mechanism is uniform only with probability $\beta$. Therefore, for $Y=\mathcal{M}_{\beta}(X)$, we have 
\begin{align*}
    \mathrm{TV}(p_{X|Y}(\cdot|Y),p_X(\cdot))= 1- \beta
\end{align*}
and therefore the privacy-leakage defined in Definition \ref{def:privacy:leakage} becomes $1- \beta$. 
{
\subsection{Connections with differential privacy} \label{sec:dp_histogram}

In the context of pricing based on histograms or distribution of values, the dataset is the set of consumers' values, and the algorithm's output is the selected market price. Thus, the goal of this privacy framework is to ensure that one cannot infer too much about the consumers' values by observing the optimal price. Note that this addresses a different aspect of privacy compared to that which we discussed previously. In our setting, the main concern is not the producer itself but rather what an \textit{adversary} can infer about the consumers' values by observing the price set for them. For instance, if one observes that the insurance premium is set high for a zip code, one may infer that residents of that zip code often require medical assistance. The application of differential privacy to pricing has also been studied in the context of revenue management \citep[see, e.g.,][]{lei2020privacy}. A caveat is that differential privacy is typically applicable to settings with a finite number of users, as opposed to our setting with a continuum of consumers. We  discuss this matter in \cref{sec:dp_histogram},  provide an adjusted definition of differential privacy, and demonstrate how our privacy mechanism \eqref{eqn:privacy_mechanism} is, in fact, differentially private.

Let us first recall the definition of differential privacy for an algorithm $\mathcal{A}$ over a set of $n$ users' data sets, denoted by $\bm{z} = \{z_1, \cdots, z_n\}$~\cite[see][for a detailed discussion of the definition and its applications]{dwork2014algorithmic}.
\begin{definition}
A randomized algorithm $\mathcal{A}:\mathcal{Z}^n \to \mathcal{Y}$ is called $\varepsilon$-differentially private if, for any measurable set $\mathcal{O}$, we have
\begin{equation}
{\mathbb{P}(\mathcal{A}(\bm{z}) \in \mathcal{{O}})} \leq e^\varepsilon ~{\mathbb{P}(\mathcal{A}(\bm{z'}) \in \mathcal{{O}})},
\end{equation}
for any two neighboring datasets $\bm{z}$ and $\bm{z'}$ that only differ in one coordinate. 
\end{definition}
Here, $\varepsilon$ quantifies the $\textit{privacy loss}$. The higher $\varepsilon$ is, it means the output is more sensitive to the user's data, meaning that it provides a weaker privacy guarantee. 

One challenge here is that, in our case, we work with a continuum of consumers, and hence, changing one user’s data is not meaningful. A natural proposal is to consider two markets that are \textit{close enough} as a surrogate for neighboring distributions. In particular, we can use a distribution distance function $d:\mathcal{X} \times \mathcal{X} \to \mathbb{R}^{\geq0}$ as a measure of the closeness of markets and say two markets $x$ and $x'$ are neighboring markets if $d(x,x') \leq \Delta$ for some choice of $\Delta >0$. We, therefore, use the following adjusted definition. Recall that $\phi(\cdot)$ denotes the optimal pricing rule and $\mathcal{X}$ represents the set of markets. 
\begin{definition}
A randomized mechanism $\mathcal{M}:\mathcal{X} \to \mathcal{X}$ is called $\varepsilon$-differentially private if, for any value $v_i \in \mathcal{V}$, we have
\begin{equation}
\mathbb{P}(\phi(\mathcal{M}(x)) =v_i)  \leq e^\varepsilon ~\mathbb{P}(\phi(\mathcal{M}(x')) =v_i),
\end{equation}
for two markets $d(x, x') \leq \Delta$.
\end{definition}
Now, let us see how our mechanism fits into this definition. 
\begin{proposition}
The mechanism $\mathcal{M}_\beta(\cdot)$ is $\varepsilon$-differentially private with 
\begin{equation}
\varepsilon = \max_{i:\text{Vol}(X_i^\beta) \neq 0} \frac{1-\beta + \beta \text{Vol}(X_i^\beta)}{\beta \text{Vol}(X_i^\beta)},    
\end{equation}
where $X_i^\beta$ is defined in \eqref{eqn:X_i_beta}.
\end{proposition}
\begin{proof}
To simplify the notation, let us denote $\mathcal{M}_\beta(x)$ by $\hat{x}$. Then, we have
\begin{equation*}
\mathbb{P}(\phi(\hat{x}) =v_i) = \mathbb{P}(\hat{x} \in X_i^\beta) = (1-\beta) \mathbbm{1}(x \in X_i^\beta) + \beta \text{Vol}(X_i^\beta).  
\end{equation*}
By changing $x$, the term $\mathbbm{1}(x \in X_i^\beta)$ could change from zero to one, but the second term $\beta \text{Vol}(X_i^\beta)$ remains unchanged. This, consequently, implies the result. 
\end{proof}
}
{
\subsection{Proof of \cref{propsotion:reference_prior}}
Let $\nu_\beta(x,\hat{x})$ denote the joint distribution of $x$ and $\hat{x}$. With a slight abuse of notation, we denote the distribution of $\hat{x}$ condition on $x$ by $\nu_\beta(\hat{x}|x)$. Note that the term in \eqref{eqn:argmax_reference_prior} that we wish to maximize is equal to
\begin{align}
\int_{\mathcal{X}} \mu_\beta(\hat{x}) & \left( \int_{\mathcal{X}} \pi_\beta(x|\hat{x}) \log(\frac{\pi_\beta(x|\hat{x})}{\pi(x)}) ~dx \right) d\hat{x} =  \int_{\mathcal{X}} \int_{\mathcal{X}} \nu_\beta(x,\hat{x}) \log(\frac{\nu_\beta(x,\hat{x})}{\pi(x)\mu_\beta(\hat{x})})  ~dx d\hat{x} \\
& = \int_{\mathcal{X}} \pi(x) \int_{\mathcal{X}} \nu_\beta(\hat{x}|x) \log(\frac{\nu_\beta(\hat{x}|x)}{\mu_\beta(\hat{x})}) ~d\hat{x} dx \\
& = \int_{\mathcal{X}} \pi(x) \left ( \int_{\mathcal{X}} \nu_\beta(\hat{x}|x) \log(\nu_\beta(\hat{x}|x)) \right) dx
- \int_{\mathcal{X}} \int_{\mathcal{X}} \pi(x) \nu_\beta(\hat{x}|x) \log(\mu_\beta(\hat{x})) \\
& = \int_{\mathcal{X}} \pi(x) \left ( \int_{\mathcal{X}} \nu_\beta(\hat{x}|x) \log(\nu_\beta(\hat{x}|x)) \right) dx
- \int_{\mathcal{X}} \mu_\beta(\hat{x}) \log(\mu_\beta(\hat{x})).
\label{eqn:proof_reference_prior_1}
\end{align}
Note that, given the privacy mechanism $\mathcal{M}_\beta(\cdot)$, $\int_{\mathcal{X}} \nu_\beta(\hat{x}|x) \log(\nu_\beta(\hat{x}|x))$ is independent of $x$, and therefore, the first term in \eqref{eqn:proof_reference_prior_1} does not depend on the choice of $\pi(\cdot)$. The second term is the entropy of $\mu_\beta(\hat{x})$ and we know that it is maximized when $\mu_\beta(\cdot)$ is the uniform distribution, which is the case when $\pi(\cdot)$ is the uniform distribution. This completes the proof. 
}
\subsection{Proof of Lemma \ref{lemma:pricing}}
Notice that the producer's utility by setting the price equal to $v_1$ is given by $\mathcal{U}_p (v_1 | \hat{x}) = v_1$ as all the consumers will buy the product at this price. For price $v_2$, the expected utility is given by
\begin{align}
\mathcal{U}_p (v_2 | \hat{x}) 
&= v_2 \left ( (1-\beta) \hat{\alpha} + \beta \int_{0}^1 \alpha ~d \alpha \right ) 
= v_2 \left ( (1-\beta) \hat{\alpha} + \frac{\beta}{2} \right ).
\end{align}
Comparing this with $\mathcal{U}_p (v_1 | \hat{x}) = v_1$ completes the proof. 
\subsection{Proof of Proposition  \ref{proposition:S_to_S'}}
First, notice that different optimal pricing rules only differ in how they price a market upon observing $\hat{\alpha} = t^*$, since in this situation both prices $v_2$ and $v_1$ are optimal. Therefore, we can parameterize the optimal pricing rules by the probability of setting the price to $v_2$ given $\hat{\alpha} = t^*$. This probability is denoted by $\delta \in [0,1]$ and its corresponding optimal pricing rule is represented by $\phi_\delta(\cdot)$.

Next, we characterize the expected utilities of consumer and the producer from a market $\alpha$. Given the above discussion, and with a slight abuse of notation, we denote these expected utilities by $U_p^\delta(\cdot)$ and $U_c^\delta(\cdot)$, respectively.
\begin{lemma} \label{lemma:expected_utility}
Consider a market $x = (1-\alpha, \alpha)$ and a pricing rule $\phi(\cdot)$ that sets the price of a market $v_2$ with probability $\delta$ in the case of observing $\hat{\alpha} = t^*$. 
\begin{enumerate}[(i)]
\item If $\alpha < t^*$, then $U_p^\delta(x) = f_1(\alpha)$ and $U_c^\delta(x) = g_1(\alpha)$, where $f_1(\cdot)$ and $g_1(\cdot)$ are given by
\begin{align*}
f_1(\alpha) = (1-\beta + \beta t^*) v_1 + \beta(1-t^*) \alpha v_2,
\quad 
g_1(\alpha) = (1-\beta + \beta t^*) \alpha (v_2-v_1).
\end{align*}
\item If $\alpha > t^*$, then $U_p^\delta(x) = f_2(\alpha)$ and $U_c^\delta(x) = g_2(\alpha)$, where $f_2(\cdot)$ and $g_2(\cdot)$ are given by
\begin{align*}
f_2(\alpha) = \beta t^* v_1 + (1-\beta + \beta (1-t^*)) \alpha v_2, \quad
g_2(\alpha) = \beta t^* \alpha(v_2-v_1).
\end{align*}
\item If $\alpha = t^*$, then $U_p^\delta(x)$ and $U_c^\delta(x)$ are given by
\begin{align*}
U_p^\delta(x) = (1-\delta) f_1(t^*) + \delta f_2(t^*), \quad 
U_c^\delta(x) = (1-\delta) g_1(t^*) + \delta g_2(t^*).
\end{align*}
\end{enumerate}
\end{lemma}
\begin{proof}
Suppose $\alpha < t^*$. Then, with probability $1-\beta + \beta t^*$, the realized $\hat{\alpha}$ falls within the interval $[0,t^*)$, and hence the optimal price is $v_1$ as per Lemma \ref{lemma:pricing}. In this scenario, the producer's utility would be $v_1$ and the consumer' utility would be $\alpha (v_2-v_1)$. Also, with probability $\beta(1-t^*)$, $\hat{\alpha}$ falls into the interval $(t^*,1]$, and the market would be priced at $v_2$. Here, the producer's utility is $\alpha v_2$ and the consumer' utility is zero. Summing these scenarios gives the expected utilities. It is also worth noting that the probability of $\hat{\alpha} = t^*$ is zero in this case.  

The case $\alpha > t^*$ can be argued similarly. Finally, for the case $\alpha=t^*$, with probability $1-\beta$ we have $\hat{\alpha} = t^*$ which leads to choosing price $v_2$ with probability $\delta$ and price $v_1$ with probability $1-\delta$. This observation establishes the final result.    
\end{proof}
Now, we turn back to our goal of characterizing the set of utility pairs  \eqref{eqn:all_pairs}. In general, searching over the entire set of segmentations and pricing rules can be challenging. However, the following lemma enables us to narrow our search without loss of generality. 
\begin{lemma} \label{lemma:merge}
Suppose Assumption \ref{assump:beta:intermediate} holds. Then, 
\begin{align}
\mathcal{S} & = 
\biggl\{ \Big[ \gamma_1 g_1(\alpha_1) + \gamma_2 g_2(\alpha_2), \gamma_1 f_1(\alpha_1) + \gamma_2 f_2(\alpha_2) \Big]^\top ~ \Big \vert \label{eqn:prop_merge} \\
&
\qquad \alpha_1 \in [0,\min\{t^*, \alpha^*\}], \alpha_2 \in [\max\{t^*, \alpha^*\},1], \gamma_1 \text{ and } \gamma_2 \in [0,1], \gamma_1 + \gamma_2 = 1, \gamma_1 \alpha_1 + \gamma_2 \alpha_2 = \alpha^*
\biggr\} \nonumber
\end{align}
\end{lemma}
\begin{proof}
We first prove that for any segmentation $\sigma \in \Sigma$ and optimal pricing rule $\phi_\delta(\cdot)$, there is a representation in the form of right hand side of \eqref{eqn:prop_merge} such that
\begin{equation} \label{eqn:equivalence}
\sum_{x \in \text{supp}(\sigma)} \sigma(x) U_p^\delta(x) = \gamma_1 f_1(\alpha_1) + \gamma_2 f_2(\alpha_2)
\quad \text{ and }
\sum_{x \in \text{supp}(\sigma)} \sigma(x) U_c^\delta(x) =  \gamma_1 g_1(\alpha_1) + \gamma_2 g_2(\alpha_2).  
\end{equation}
To simplify the notation, we use $\sigma(\alpha)$ and $U_p^\delta(\alpha)$ to denote $\sigma(x)$ and $U_p^\delta(x)$ when $x= (1-\alpha, \alpha)$. Notice that 
\begin{align}
\sum_{\alpha: \sigma(\alpha) >0} & \sigma(\alpha) U_p^\delta(\alpha) =
\sum_{\alpha < t^* \& \sigma(\alpha) >0} \sigma(\alpha) U_p^\delta(\alpha)
+
\sum_{\alpha >t^* \& \sigma(\alpha) >0} \sigma(\alpha) U_p^\delta(\alpha)
+
\sigma(t^*) U_p^\delta(t^*) \nonumber \\
&= 
\sum_{\alpha < t^* \& \sigma(\alpha) >0} \sigma(\alpha) f_1(\alpha)
+
\sum_{\alpha >t^* \& \sigma(\alpha) >0} \sigma(\alpha) f_2(\alpha)
+
\sigma(t^*) \left ( (1-\delta) f_1(t^*) + \delta f_2(t^*) \right ), \label{eqn:prop_proof_1}
\end{align}
where the last equation uses Lemma \ref{lemma:expected_utility}. Now, notice that $f_1(\cdot)$ and $f_2(\cdot)$ are linear functions. Therefore, we can cast \eqref{eqn:prop_proof_1} as
\begin{equation*}
\gamma_1 f_1(\alpha_1) + \gamma_2 f_2(\alpha_2) 
\end{equation*}
with
\begin{align*}
\gamma_1 &= \sum_{\alpha < t^* \& \sigma(\alpha) >0} \sigma(\alpha) + (1-\delta) \sigma(t^*), \\
\alpha_1 & = \frac{1}{\gamma_1} \left ( \sum_{\alpha < t^* \& \sigma(\alpha) >0} \sigma(\alpha) \alpha + (1-\delta) \sigma(t^*) t^* \right ), \\
\gamma_2 &= \sum_{\alpha > t^* \& \sigma(\alpha) >0} \sigma(\alpha) + \delta \sigma(t^*), \\
\alpha_2 & = \frac{1}{\gamma_2} \left ( \sum_{\alpha > t^* \& \sigma(\alpha) >0} \sigma(\alpha) \alpha + \delta \sigma(t^*) t^* \right ).
\end{align*}
Notice that a similar representation can also be shown for the consumer' utility. We next verify that the conditions in \eqref{eqn:prop_merge} are met. Firstly, as $\sigma(\cdot)$ is a segmentation within $\Sigma$, it can be verified that $\gamma_1$ and $\gamma_2$ are nonnegative and together sum up to one. Furthermore, the condition $\sum_{\alpha: \sigma(\alpha) >0} \sigma(\alpha) \alpha = \alpha^*$ ensures that $\gamma_1 \alpha_1 + \gamma_2 \alpha_2 = \alpha^*$.

Given that $\alpha_1$ is effectively a weighted average of $t^*$ and several values of $\alpha < t^*$, it follows that $\alpha_1 \leq t^*$. Similarly, it can be established that $\alpha_2 \geq t^*$. Moreover, a weighted average of $\alpha_1$ and $\alpha_2$ equates to $\alpha^*$. Consequently, we should have $\alpha_1 \leq \alpha^* \leq \alpha_2$. All these considerations imply that $\alpha_1 \in [0,\min\{t^*, \alpha^*\}]$ and $\alpha_2 \in [\max\{t^*, \alpha^*\},1]$.

Now we establish the converse. Suppose $\alpha_1, \alpha_2, \gamma_1$, and $\gamma_2$ are given, satisfying the conditions on the right-hand side of \eqref{eqn:prop_merge}. We will show that there exists a segmentation $\sigma \in \Sigma$ and an optimal pricing rule $\phi_\delta(\cdot)$ such that \eqref{eqn:equivalence} is satisfied. If $\alpha_1 < t^*$ and $\alpha_2 > t^*$, we can consider a segmentation with two markets $(1-\alpha_1, \alpha_1)$ and $(1-\alpha_2, \alpha_2)$, assigned probabilities $\gamma_1$ and $\gamma_2$ respectively, and apply any optimal pricing rule to satisfy \eqref{eqn:equivalence}. If $\alpha_1 = t^*$ and $\alpha_2 > t^*$, then $\phi_0(\cdot)$ should be taken as the optimal pricing rule. Conversely, if $\alpha_1 < t^*$ and $\alpha_2 = t^*$, $\phi_1(\cdot)$ becomes the optimal pricing rule. The only remaining scenario is $\alpha_1 = \alpha_2 = t^**$, where we take $\phi_\delta(\cdot)$ with $\delta = \gamma_2$ as the optimal pricing rule. This completes the proof.    
\end{proof}
Next, note that we can cast $\gamma_1 f_1(\alpha_1) + \gamma_2 f_2(\alpha_2)$ in the right hand side of \eqref{eqn:prop_merge} in the following way:
\begin{align} \label{eqn:U_p_1}
\gamma_1 f_1(\alpha_1) + \gamma_2 f_2(\alpha_2) &= \beta \left ( t^* v_1 + (1-t^*) \alpha^* v_2 \right ) 
+ (1-\beta) \left ( \gamma_1 v_1 + \gamma_2 \alpha_2 v_2 \right), 
\end{align}
where we used the fact that $\gamma_1 \alpha_1 + \gamma_2 \alpha_2 = \alpha^*$ and $\gamma_1 + \gamma_2 = 1$. Similarly, we have
\begin{align} \label{eqn:U_c_1}
\gamma_1 g_1(\alpha_1) + \gamma_2 g_2(\alpha_2) 
= \beta t^* \alpha^* (v_2-v_1)
+ (1-\beta) \gamma_1 \alpha_1 (v_2-v_1).
\end{align}
Notice that the terms $\beta \left ( t^* v_1 + (1-t^*) \alpha^* v_2 \right )$ and $\beta t^* \alpha^* (v_2-v_1)$ in \eqref{eqn:U_p_1} and \eqref{eqn:U_c_1}, respectively, are constants. Thus, it suffices to characterize the set $\mathcal{S}'$.
\subsection{Proof of Theorem \ref{theorem:S_K=2}}
We first identify the set $\mathcal{S}'$, doing so by separately considering the cases $\alpha^* \geq t^*$ and $\alpha^* \leq t^*$. 
\paragraph{Case (I): $\alpha^* \geq t^*$}
\begin{lemma} \label{lemma:S'_alpha_geq_t}
Suppose $\beta \leq \min\{2\eta, 2(1-\eta)\}$ and $\alpha^* \geq t^*$. Then $\mathcal{S}'$ is in the form of a triangle with three vertices $E=[0, \alpha^* v_2 + (1-\alpha^*)v_1]^\top$, $F=[0, \alpha^* v_2]^\top$, and $G= E + [\xi, -\xi]^\top$ with
\begin{equation*}
\xi := 
(1-\alpha^*) \frac{t^*}{1-t^*} (v_2-v_1).
\end{equation*}
\end{lemma}
\begin{proof}
First, note that, the two conditions $\gamma_1 + \gamma_2 =1$ and $\gamma_1 \alpha_1 + \gamma_2 \alpha_2 = \alpha^*$ together imply that
\begin{equation} \label{eqn:lemma_4_2}
\gamma_1 = \frac{\alpha_2 - \alpha^*}{\alpha_2 - \alpha_1}, \quad	
\gamma_2 = \frac{\alpha^*- \alpha_1}{\alpha_2 - \alpha_1}.
\end{equation}
Next, let us define $z_1 := \gamma_1(1-\alpha_1)$. Using $\alpha_2 \in [\alpha^*, 1]$ along with \eqref{eqn:lemma_4_2} implies that, for any $\alpha_1$, $z_1$ spans the interval $[0, 1-\alpha^*]$ as $\alpha_2$ varies over $[\alpha^*, 1]$.
Additionally, let us define $z_2 := \frac{\alpha_1}{1-\alpha_1}$. Note that as $\alpha_1$ spans the interval $[0, t^*]$, $z_2$ spans the interval $[0, \frac{t^*}{1-t^*}]$.

Now, notice that we can rewrite $\gamma_1 v_1 + \gamma_2 \alpha_2 v_2$ as 
\begin{align}
\gamma_1 v_1 + \gamma_2 \alpha_2 v_2 &= 
\gamma_1 \alpha_1 v_2 - \gamma_1 \alpha_1 (v_2 - v_1) + \gamma_1 (1-\alpha_1) v_1 + \gamma_2 \alpha_2 v_2 \nonumber \\
&= \alpha^* v_2 + \gamma_1 (1-\alpha_1) v_1 - \gamma_1 \alpha_1 (v_2 - v_1) \nonumber \\
&= \alpha^* v_2 + z_1 v_1 - z_1 z_2 (v_2-v_1), \label{eqn:lemma4_1}
\end{align}   
where the second equality used the fact that $\gamma_1 \alpha_1 + \gamma_2 \alpha_2 = \alpha^*$. Therefore, we can cast $\mathcal{S}'$ as
\begin{align}
\mathcal{S'} = 
\biggl\{ \Big[ z_1 z_2 (v_2-v_1), \alpha^* v_2 + z_1 v_1 - z_1 z_2 (v_2-v_1) \Big]^\top ~ \Big \vert  z_1 \in [0,1-\alpha^*], z_2 \in [0, \frac{t^*}{1-t^*}] \biggr\}.
\end{align}
It is straightforward to verify that this set is indeed the triangle $EFG$.   
\end{proof}
Figure \ref{fig:S'_alpha_geq_t} illustrates this triangle for the two cases of $t^* \leq \eta$ and $t^* \geq \eta$. Notice that the triangle formed by points $A$, $B$, and $D$ corresponds with the non-private case where $\beta = 0$, implying $t^* = \eta$ and $\mathcal{S}' = \mathcal{S}$. 
\begin{figure}
\centering
\begin{subfigure}{.5\textwidth}
  \centering
  \begin{tikzpicture}
    \draw[->] (0,0) -- (4,0) node[right] {}; 
    \draw[->] (0,0) -- (0,4) node[above] {}; 
    
    \coordinate (A) at (0,3); 
    \coordinate (B) at (0,1); 
    \coordinate (C) at (1,2); 
    \coordinate (D) at (2,1);
    \coordinate (E') at (2,0);

    \draw[fill=gray!30] (A) -- (B) -- (C) -- cycle; 

    \draw[dotted] (A) -- ++(2,-2);
    \draw[dotted] (B) -- ++(2,0);
    \draw[dotted] (D) -- ++(0,-1);

    \node[right] at (A) {$\alpha^* v_2 + (1-\alpha^*) v_1$};
    \node[left] at (A) {E};
    \node[right] at (B) {$\alpha^* v_2$};
    \node[left] at (B) {F};
    \node[right] at (C) {G};
    \node[right] at (D) {H};
    \node[below] at (E') {$(1-\alpha^*) v_1$};
\end{tikzpicture}
  \caption{$t^* \leq \eta$}
  \label{fig:S'_alpha_eta_geq_t}
\end{subfigure}%
\begin{subfigure}{.5\textwidth}
  \centering
  \begin{tikzpicture}
    \draw[->] (0,0) -- (4,0) node[right] {}; 
    \draw[->] (0,0) -- (0,4) node[above] {}; 
    
    \coordinate (A) at (0,3); 
    \coordinate (B) at (0,1); 
    \coordinate (C) at (2.5,0.5); 
    \coordinate (D) at (2,1);
    \coordinate (E') at (2,0);

    \draw[fill=gray!30] (A) -- (B) -- (C) -- cycle; 

    \draw[dotted] (A) -- ++(2,-2);
    \draw[dotted] (B) -- ++(2,0);
    \draw[dotted] (D) -- ++(0,-1);

    \node[right] at (A) {$\alpha^* v_2 + (1-\alpha^*) v_1$};
    \node[left] at (A) {E};
    \node[right] at (B) {$\alpha^* v_2$};
    \node[left] at (B) {F};
    \node[right] at (C) {G};
    \node[right] at (D) {H};
    \node[below] at (E') {$(1-\alpha^*) v_1$};
\end{tikzpicture}
  \caption{$t^* \geq \eta$}
  \label{fig:S'_alpha_geq_t_geq_eta}
\end{subfigure}
\caption{Illustration of $\mathcal{S}'$ for the case $\alpha^* \geq t^*$}
\label{fig:S'_alpha_geq_t}
\end{figure}
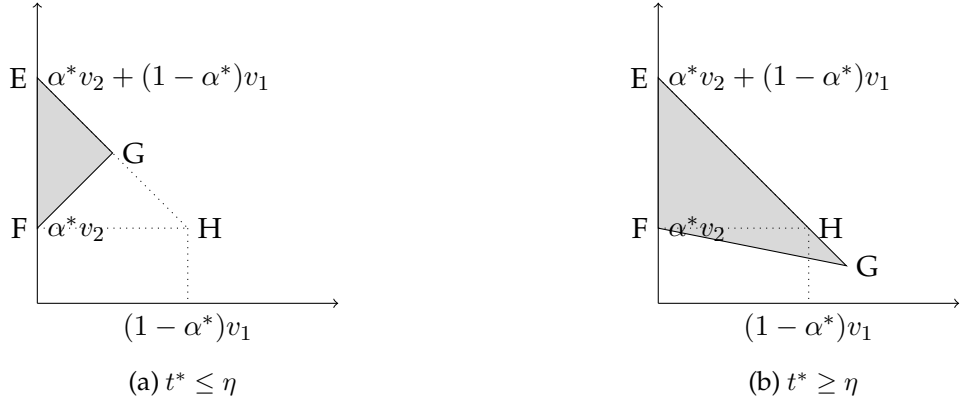
\paragraph{Case (II): $\alpha^* \leq t^*$}
\begin{lemma} \label{lemma:S'_alpha_leq_t}
Suppose $\beta \leq \min\{2\eta, 2(1-\eta)\}$ and $\alpha^* \leq t^*$. Then $\mathcal{S}'$ is in the form of a triangle with three vertices $E=[0, \alpha^* v_2 + (1-\alpha^*)v_1]^\top$, $F=[\alpha^*(v_2-v_1), v_1]^\top$, and $G=[0,\kappa]^\top$ with
\begin{equation*}
\kappa := v_1 + \alpha^* v_2 (1-\frac{\eta}{t^*}).     
\end{equation*}
\end{lemma}
\begin{proof}
The proof technique is similar to the ones used in the proof of Lemma \ref{lemma:S'_alpha_geq_t}. Here, we define $w_1 := \gamma_2 \alpha_2$ and $w_2 := \frac{1}{\alpha_2}$. Using \eqref{eqn:lemma_4_2}, we can see that $\alpha_1$ and $\alpha_2$ sweeping over $[0,\alpha^*]$ and $[t^*,1]$, respectively, imply $w_1$ and $w_2$ spanning over $[0,\alpha^*]$ and $[1,\frac{1}{t^*}]$, respectively.  

We next represent $\mathcal{S}'$ using $w_1$ and $w_2$. Particularly, we can rewrite $\mathcal{S}'$ as
\begin{align}
\mathcal{S'} = 
\biggl\{ \Big[ (\alpha^*-w_1) (v_2-v_1), v_1 + w_1 v_2 - w_1 w_2 v_1 \Big]^\top ~ \Big \vert  w_1 \in [0,\alpha^*], w_2 \in [1, \frac{1}{t^*}] \biggr\}.
\end{align}
We can see that this is the triangle given in the lemma's statement. 
\end{proof}
Figure \ref{fig:S'_alpha_leq_t} illustrates this triangle for the two cases of $t^* \leq \eta$ and $t^* \geq \eta$. Here, the triangle formed by points $A$, $E$, and $G$ corresponds with the non-private case where $\beta = 0$, implying $t^* = \eta$ and $\mathcal{S}' = \mathcal{S}$. 
\begin{figure}
\centering
\begin{subfigure}{.5\textwidth}
  \centering
  \begin{tikzpicture}
    \draw[->] (0,0) -- (4,0) node[right] {}; 
    \draw[->] (0,0) -- (0,4) node[above] {}; 
    
    \coordinate (E) at (0,3); 
    \coordinate (H) at (0,1); 
    \coordinate (G) at (0,0.5); 
    \coordinate (F) at (2,1);
    \coordinate (X) at (2,0);

    \draw[fill=gray!30] (E) -- (F) -- (G) -- cycle; 

    \draw[dotted] (E) -- ++(2,-2);
    \draw[dotted] (H) -- ++(2,0);
    \draw[dotted] (F) -- ++(0,-1);

    \node[right] at (E) {$\alpha^* v_2 + (1-\alpha^*) v_1$};
    \node[left] at (E) {E};
    \node[right] at (H) {$v_1$};
    \node[left] at (H) {H};
    \node[left] at (G) {G};
    \node[right] at (F) {F};
    \node[below] at (X) {$\alpha^*(v_2-v_1)$};
\end{tikzpicture}
  \caption{$t^* \leq \eta$}
  \label{fig:S'_alpha_leq_t_leq_eta}
\end{subfigure}%
\begin{subfigure}{.5\textwidth}
  \centering
  \begin{tikzpicture}
    \draw[->] (0,0) -- (4,0) node[right] {}; 
    \draw[->] (0,0) -- (0,4) node[above] {}; 
    
    \coordinate (E) at (0,3); 
    \coordinate (H) at (0,1); 
    \coordinate (G) at (0,1.5); 
    \coordinate (F) at (2,1);
    \coordinate (X) at (2,0);

    \draw[fill=gray!30] (E) -- (F) -- (G) -- cycle; 

    \draw[dotted] (E) -- ++(2,-2);
    \draw[dotted] (H) -- ++(2,0);
    \draw[dotted] (F) -- ++(0,-1);

    \node[right] at (E) {$\alpha^* v_2 + (1-\alpha^*) v_1$};
    \node[left] at (E) {E};
    \node[right] at (H) {$v_1$};
    \node[left] at (H) {H};
    \node[left] at (G) {G};
    \node[right] at (F) {F};
    \node[below] at (X) {$\alpha^*(v_2-v_1)$};
\end{tikzpicture}
  \caption{$t^* \geq \eta$}
  \label{fig:S'_alpha_eta_leq_t}
\end{subfigure}
\caption{Illustration of $\mathcal{S}'$ for the case $\alpha^* \leq t^*$}
\label{fig:S'_alpha_leq_t}
\end{figure}
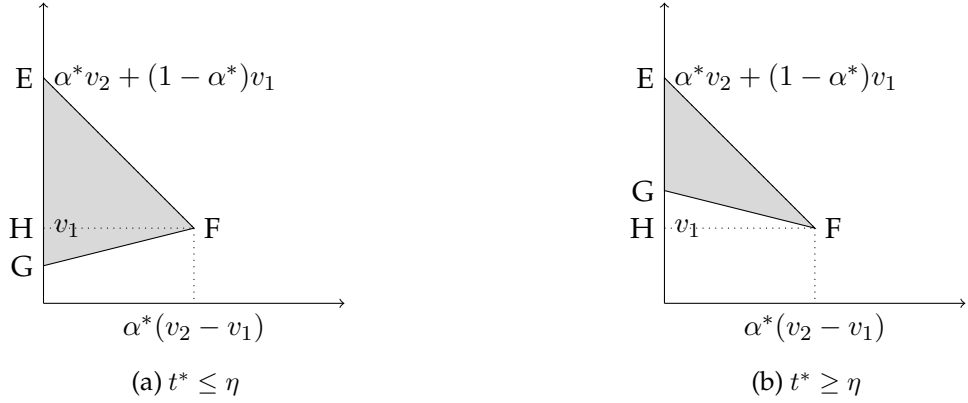
Now, using Proposition \ref{proposition:S_to_S'} along with Lemmas \ref{lemma:S'_alpha_geq_t} and \ref{lemma:S'_alpha_leq_t} completes the proof.
\subsection{Proof of Lemma \ref{lemma:max_producer_nonmonotone}}
Recall by \eqref{eqn:producer_max_utility_privacy} that the maximum producer utility is given by
\begin{equation} \label{eqn:producer_max_utility_1}
u_{\text{max}}(\beta):= \alpha^* v_2 + (1-\alpha^*) v_1 - \beta  \left ( t^* (\alpha^* v_2 - v_1) + (1-\alpha^*) v_1 \right ),
\end{equation}
where $t^*$ is given by Lemma \ref{lemma:pricing}. Taking the derivative with respect to $\beta$, we have
\begin{align}
\frac{d}{d \beta} u_{\text{max}}(\beta) &= 
-\frac{(\eta+\alpha^*-2\alpha^* \eta)(\beta^2 - 2 \beta) + 2\eta - 2 \eta^2}{2\eta(1-\beta)^2} \nonumber \\
&= -\frac{\eta+\alpha^*-2\alpha^* \eta}{2\eta(1-\beta)^2} \left ( (1-\beta)^2 + \frac{(\alpha^*-\eta)(2\eta-1)}{\eta + \alpha^* - 2\alpha^* \eta} \right ). \label{eqn:producer_max_utility_2}
\end{align}
Notice that $\alpha^* + \eta - 2\alpha^*\eta \geq 0$ as $\alpha^*, \eta \leq 1$. Now, to see whether the maximum producer utility is monotone or not, we need to check whether the sign of 
\begin{equation}\label{eqn:producer_max_utility_3}
h(\beta):= (1-\beta)^2 + \frac{(\alpha^*-\eta)(2\eta-1)}{\eta + \alpha^* - 2\alpha^* \eta}    
\end{equation}
changes as $\beta$ varies over $[0, \min\{2\eta, 2(1-\eta)\}]$. Note that $h(\beta)$ is decreasing over this interval. We also know that the maximum utility under privacy is (weakly) lower than the non-private case, meaning that \eqref{eqn:producer_max_utility_2} is nonpositive at $\beta=0$ which implies $h(0) \geq 0$ (we could also simplify verify this by setting $\beta=0$ at \eqref{eqn:producer_max_utility_3}).

Consequently, the maximum producer utility is monotone in $\beta$ if and only if
\begin{equation*}
h\left (\min\{2\eta, 2(1-\eta)\} \right ) \geq 0.    
\end{equation*}
Notice that
\begin{align}
h\left (\min\{2\eta, 2(1-\eta)\} \right ) &=
(1-2\eta)^2 + \frac{(\alpha^*-\eta)(2\eta-1)}{\eta + \alpha^* - 2\alpha^* \eta} = 
\frac{(1-2\eta)2\eta(1-\eta)(1-2\alpha^*)}{\eta + \alpha^* - 2\alpha^* \eta},
\end{align}
which is nonnegative if and only if $(1-2\eta)(1-2\alpha^*) \geq 0$. This completes the proof. 
\subsection{Proof of Lemma \ref{lemma:min_consumer_nonmonotone}}
Recall that, by Theorem \ref{theorem:S_K=2}, the minimum consumer utility is given by
\begin{equation}
\beta t^* \alpha^* (v_2-v_1).    
\end{equation}
To understand when this is a monotone function of $\beta$, we need to look into the derivative of $\beta t^*$ as a function of $\beta$. This derivative is given by
\begin{equation} \label{eqn:consumer_min_utility_1}
\frac{2\eta - \beta(2-\beta)}{2(1-\beta)^2} = \frac{(\beta-1)^2 -(1-2\eta)}{2(1-\beta)^2}.   
\end{equation}
We know that this derivative is nonnegative at $\beta=0$. Hence, the minimum consumer utility is monotone if and only if the above derivative is nonnegative at $\beta = \min\{2\eta, 2(1-\eta)\}$. Plugging this value for $\beta$ into \eqref{eqn:consumer_min_utility_1}, we have
\begin{equation*}
\frac{(1-2\eta)^2 -(1-2\eta)}{2(1-\beta)^2} = 
\frac{-2\eta(1-2\eta)}{2(1-\beta)^2},
\end{equation*}
which is nonnegative if and only if $\eta \geq 1/2$. This completes the proof. 
\subsection{Proof of Lemma \ref{lemma:bar_beta}}
Recall that $\mathcal{U}_p(v_i|\hat{x})$, is given by
\begin{equation} \label{eqn:lemma_X_i_1}
\mathcal{U}_p(v_i|\hat{x}) = (1-\beta) ~ \mathcal{U}_p(v_i, \hat{x}) + \beta ~ \mathbb{E}_{x \sim \text{Unif}(\mathcal{X})} \left[\mathcal{U}_p(v_i, x) \right].
\end{equation}
Also, recall that 
\begin{equation*}
\mathcal{U}_p(v_i, x) = v_i \sum_{k=i}^K x(v_k),    
\end{equation*}
and therefore, we have
\begin{equation}
\mathbb{E}_{x \sim \text{Unif}(\mathcal{X})} \left[\mathcal{U}_p(v_i, x) \right]
= v_i \sum_{k=i}^K \mathbb{E}_{x \sim \text{Unif}(\mathcal{X})} \left[x(v_k) \right]
= v_i (K+1-i) \frac{1}{K},
\end{equation}
where the last argument follows from the fact that the expectation of $x(v_1), \cdots, x(v_K)$ are all equal and equal to $1/K$ when $x$ is drawn uniformly at random. Thus, we can rewrite  $\mathcal{U}_p(v_i|\hat{x})$ as
\begin{equation} \label{eqn:lemma_X_i_2}
\mathcal{U}_p(v_i|\hat{x}) = 
v_i \left ( (1-\beta) ~ \sum_{k=i}^K \hat{x}(v_k) + \beta ~ \frac{K+1-i}{K} \right ).
\end{equation}
Now, if $\hat{x} \in \mathcal{X}_i^\beta$, then we have 
\begin{equation} \label{eqn:lemma_X_i_3}
\mathcal{U}_p(v_i|\hat{x}) \geq \mathcal{U}_p(v_j|\hat{x}) \text{ for all } j \in [K]. 
\end{equation}
Plugging \eqref{eqn:lemma_X_i_2} into this equation implies
\begin{equation} \label{eqn:lemma_X_i_4}
v_i \sum_{k=i}^K \hat{x}(v_k) + \frac{\beta}{1-\beta} \cdot \frac{v_i(K+1-i)}{K} 
\geq 
v_j \sum_{k=j}^K \hat{x}(v_k) + \frac{\beta}{1-\beta} \cdot \frac{v_j(K+1-j)}{K}. 
\end{equation}
First, suppose $j < i$. In this case, $\sum_{k=j}^K \hat{x}(v_k) \geq \sum_{k=i}^K \hat{x}(v_k)$, and hence, we have
\begin{equation} \label{eqn:lemma_X_i_5}
(v_i - v_j) \sum_{k=i}^K \hat{x}(v_k) + \frac{\beta}{1-\beta} \cdot \frac{v_i(K+1-i)}{K} 
\geq 
\frac{\beta}{1-\beta} \cdot \frac{v_j(K+1-j)}{K}.
\end{equation}
Furthermore, using $1 \geq \sum_{k=i}^K \hat{x}(v_k)$, we obtain 
\begin{equation} \label{eqn:lemma_X_i_6}
\frac{\beta}{1-\beta} \geq 
\frac{K (v_i- v_j)}{\Big[(K+1-j)v_j - (K+1-i)v_i \Big]_{+}}.
\end{equation}
Next, consider the case $j > i$. In this case, we use the bounds  $1 \geq \sum_{k=i}^K \hat{x}(v_k)$ and $\sum_{k=j}^K \hat{x}(v_k) \geq 0$ along with \eqref{eqn:lemma_X_i_4} to obtain
\begin{equation} \label{eqn:lemma_X_i_7}
\frac{\beta}{1-\beta} \geq 
\frac{K v_i}{\Big[(K+1-j)v_j - (K+1-i)v_i \Big]_{+}}.
\end{equation}
Notice that \eqref{eqn:lemma_X_i_6} and \eqref{eqn:lemma_X_i_7} together completes the proof of the fact that if $\mathcal{X}_i^\beta$ is non-empty, then $\beta \geq \bar{\beta}_i$. To see why the reverse is also true, notice the all the inequalities that we used would change to equality for the case that the market only consists of consumers with value $v_i$, i.e., $\hat{x}(v_i) = 1$. So this market is indeed in $\mathcal{X}_i^\beta$ as long as $\beta \geq \bar{\beta}_i$.
\subsection{Proof of Proposition \ref{proposition:S_to_S'_K}}
First, note that $\mathcal{M}_\beta(x)$ is equal to $x$ itself with probability $1-\beta$ and is equal to a uniformly random market with probability $\beta$. In the latter case, the market falls into $\mathcal{X}_k^\beta$ with probability $\mathbb{P}(\mathcal{X}_k^\beta)$ and hence is priced at $v_k$ (the boundary of $\mathcal{X}_k^\beta$'s is measure-zero). Therefore, we have
\begin{align*}
U_p^\phi(x) &= \beta~\sum_{k=1}^K \mathbb{P}(\mathcal{X}_k^\beta) ~\mathcal{U}_p(v_k, x) + (1-\beta)~\mathbb{E}_{p \sim \phi(x)} [~\mathcal{U}_p(p,x)] \\ 
U_c^\phi(x) &= \beta~\sum_{k=1}^K \mathbb{P}(\mathcal{X}_k^\beta) ~\mathcal{U}_c(v_k, x) + (1-\beta)~\mathbb{E}_{p \sim \phi(x)} [~\mathcal{U}_c(p,x)].
\end{align*}
Hence, $\mathcal{S}$ in \eqref{eqn:all_pairs} can be written as
\begin{align} \label{eqn:proof_merge_K_1}
& \beta \left \{ 
\Big[ \sum_{x \in \text{supp}(\sigma)} \sigma(x) \sum_{k=1}^K \mathbb{P}(\mathcal{X}_k^\beta) ~\mathcal{U}_c(v_k, x), \sum_{x \in \text{supp}(\sigma)} \sigma(x) \sum_{k=1}^K \mathbb{P}(\mathcal{X}_k^\beta) ~\mathcal{U}_p(v_k, x)
\Big]^\top ~ \Big \vert ~ \sigma \in \Sigma
\right \} + (1-\beta) \mathcal{S}_1,    
\end{align}
where $\mathcal{S}_1$ is given by
\begin{equation} \label{eqn:proof_merge_K_2}
\left \{ 
\Big[ \sum_{x \in \text{supp}(\sigma)} \sigma(x) \mathbb{E}_{p \sim \phi(x)} [~\mathcal{U}_c(p,x)], \sum_{x \in \text{supp}(\sigma)} \sigma(x) \mathbb{E}_{p \sim \phi(x)} [~\mathcal{U}_c(p,x)]
\Big]^\top ~ \Big \vert ~ \sigma \in \Sigma \text{ and } \phi \text{ optimal pricing}
\right \}    
\end{equation}
Now, note that $\mathcal{U}_c(v_k, x)$ is linear in $x$, and thus, we have
\begin{equation*}
\sum_{x \in \text{supp}(\sigma)} \sigma(x) \sum_{k=1}^K \mathbb{P}(\mathcal{X}_k^\beta) ~\mathcal{U}_c(v_k, x)
=  \sum_{k=1}^K \mathbb{P}(\mathcal{X}_k^\beta) \sum_{x \in \text{supp}(\sigma)} ~\sigma(x) \mathcal{U}_c(v_k, x)
=  \sum_{k=1}^K \mathbb{P}(\mathcal{X}_k^\beta) ~\mathcal{U}_c(v_k, x^*).
\end{equation*}
We can similarly simplify the other term corresponding the produce utility and hence the first term in \eqref{eqn:proof_merge_K_1} is $\beta \mathbbm{c}$. Now, it remains to show that $\mathcal{S}_1$ is same as the set $\mathcal{S}'$, defined in the statement of the proposition.

Next, we argue that each market segment in $\mathcal{S}_1$ can be subdivided into smaller segments, each priced at a single value, by allowing multiple versions of the same market priced at different, yet optimal, values. Formally, considering a $\sigma \in \Sigma$ and an optimal pricing rule $\phi(\cdot)$, for any $x \in \text{supp}(\sigma)$, let $\mathcal{V}(x) \subseteq \mathcal{V}$ be the support of $\phi(x)$. All values in $\mathcal{V}(x)$ are optimal prices for $x$.

For each $x \in \text{supp}(\sigma)$, we create $\vert \mathcal{V}(x) \vert$ replicas of $x$. Specifically, for each $v \in \mathcal{V}(x)$, we replicate market $x$, corresponding to a segment of the consumer continuum $\sigma(x) \mathbb{P}(\phi(x) = v)$, and price it at $v$. Notice that this segmentation breakdown maintains identical consumer and producer utilities. We can formally express this as:
\begin{align*}
& \mathcal{S}_1 = \left \{
\Big[ \sum_{k=1}^K \sum_{i=1}^{\ell_k} \gamma_k^i ~\mathcal{U}_c(v_k,x_k^i),
\sum_{k=1}^K \sum_{i=1}^{\ell_k} \gamma_k^i ~\mathcal{U}_p(v_k,x_k^i) 
\Big]^\top \Big \vert
\gamma_k^i \geq 0, \sum_{i,k} \gamma_k^i = 1, x_k^i \in \mathcal{X}_k^\beta, 
\sum_{i,k} \gamma_k^i x_k^i=x^*
\right \}.
\end{align*}
Finally, once again using the fact that $\mathcal{U}_c(v_k, x)$ and $\mathcal{U}_p(v_k, x)$ are linear in $x$, we can simplify the above term to
\begin{align} \label{eqn:proof_merge_K_3}
\left \{
\Big[ \sum_{k=1}^K \gamma_{k} ~\mathcal{U}_c(v_k,x_{k}),
\sum_{k=1}^K \gamma_{k} ~\mathcal{U}_p(v_k,x_{k}) 
\Big]^\top ~ \Big \vert ~
\gamma_{k} \geq 0, \sum_{k} \gamma_{k} = 1, x_{k} \in \mathcal{X}_k^\beta, 
\sum_{k} \gamma_{k} x_{k}=x^*
\right \},
\end{align}
with 
\begin{equation}
\gamma_k = \sum_{i=1}^{\ell_k} \gamma_k^i, \text{ and }
x_k = \sum_{i=1}^{\ell_k} \gamma_k^i x_k^i.
\end{equation}
Notice that \eqref{eqn:proof_merge_K_3} is the set $\mathcal{S}'$ and hence the proof is complete. 
\subsection{Proof of Theorem \ref{theorem:S'_polytope}}
Here, we provide the details of the proof sketch provided in the main text. Before computing the utilities, it is worth noting that the proof of inequality \eqref{eqn:optimal_pricing_x_i} is similar to the derivation of \eqref{eqn:lemma_X_i_4} in the proof of Lemma \ref{lemma:bar_beta}. 

We start with the producer utility corresponding to the set $\mathcal{S}'$. We have
\begin{align}
\sum_{i=1}^K \gamma_i ~ \mathcal{U}_p(v_i, x_i) 
= \sum_{i=1}^K \gamma_i v_i \sum_{j=i}^K x_i(v_j) 
= \sum_{i=1}^K \gamma_i v_i y(i,i) = \sum_{i=1}^K v_i z(i,i).
\end{align}
Next, regarding the consumer's side, notice that we have
\begin{align}
\sum_{i=1}^K \gamma_i ~ \mathcal{U}_c(v_i, x_i) &=
\sum_{i=1}^K \gamma_i ~ \sum_{j=i}^K (v_j - v_i) (y(i,j) - y(i,j+1)),  
\end{align}
with the convenience that $y(i,K+1) := 0$. Notice that we could start the index $j$ from $i+1$ and adjust the range of index $i$ to end at $K-1$. By making these changes, we obtain 
\begin{equation}
\sum_{i=1}^{K-1} \gamma_i ~ \sum_{j=i+1}^K (v_j - v_i) (y(i,j) - y(i,j+1)).    
\end{equation}
Notice that we can rewrite this as 
\begin{equation}
\sum_{i=1}^{K-1} \sum_{j=i+1}^K ~\gamma_i ~y(i,j) (v_j - v_i - (v_{j-1} - v_i))
\end{equation}
which is equal to 
\begin{equation}
\sum_{i=1}^{K-1} \sum_{j=i+1}^K ~z(i,j) (v_j - v_{j-1}).
\end{equation}
This completes the proof.
\subsection{Proof of \cref{proposition:minimum_S'}}
Let $y^*$ denote the complementary cumulative distribution function of market $x^*$, i.e.,
\begin{equation}
y^*(i) := \sum_{k=i}^K x^*(v_k).
\end{equation}
First, we establish that the lowest point of the set $\mathcal{S}'$ always lies on or beneath the line $TR$. To see this, note that there exists an $\ell^*$ for which $x^* \in \mathcal{X}_{\ell^*}^\beta$. Consequently, there exists one point in $\mathcal{S}'$ corresponding to uniformly pricing the market at $v_{\ell^*}$. Since the line $TR$ represents the optimal uniform pricing, this point either resides on $TR$ (if $\ell^*$ is the optimal uniform price for the non-private case) or falls beneath it in other cases. This substantiates the first claim.

Next, we show that if the condition \eqref{eqn:crossing_condition} does not hold, the set $\mathcal{S}'$ will not cross the line $TR$. To prove this, we use the polytope representation in Theorem \ref{theorem:S'_polytope}. Summing up both sides of \eqref{eqn:polytope_b} over $i$ (for a fixed $j$), we obtain
\begin{equation} \label{eqn:crossing_proof1}
\sum_{i=1}^K z(i,i) v_i \geq 
\sum_{i=1}^K  z(i,j) v_j + 
\frac{\beta}{K(1-\beta)} \sum_{i=1}^K z(i,1)\Big ( (K+1-j) v_j - (K+1-i) v_i \Big ).   
\end{equation}
Notice that, by \eqref{eqn:polytope_c}, $\sum_{i=1}^K z(i,j) = y^*(j)$. In particular, $\sum_{i=1}^K z(i,1) = y^*(1) = 1$. Hence, we can rewrite \eqref{eqn:crossing_proof1} as
\begin{equation} \label{eqn:crossing_proof2}
\sum_{i=1}^K z(i,i) v_i \geq 
y^*(j) v_j + 
\frac{\beta}{K(1-\beta)} \Big ( (K+1-j) v_j - \sum_{i=1}^K z(i,1) (K+1-i) v_i \Big ).   
\end{equation}
Now, notice that the left hand side is the producer utility in $\mathcal{S'}$. Choose $j$ such that $v_j$ be the optimal uniform pricing (if there are multiple of them, choose the one with the highest $(K+1-j) v_j$ value). Hence, $y^*(j) v_j$ represents the producer utility corresponding to the line $TR$. Now, recall that $\sum_{i=1}^K z(i,1) = y^*(1) = 1$, and therefore, the term 
\begin{equation} \label{eqn:crossing_proof3}
\sum_{i=1}^K z(i,1) (K+1-i) v_i    
\end{equation}
is a weighted average of $((K+1-i) v_i)_{i=1}^K$. Since the condition \eqref{eqn:crossing_condition} does not hold, each term $(K+1-i) v_i$, and hence the weighted average \eqref{eqn:crossing_proof3}, are weakly smaller than $(K+1-j) v_j$. Plugging this into \eqref{eqn:crossing_proof2} shows that when the condition \eqref{eqn:crossing_condition} does not hold, the term $\sum_{i=1}^K z(i,i) v_i$ is always weakly larger than the producer utility at the optimal uniform price, and hence the set $\mathcal{S}'$ does not go below the line $TR$.

Now, suppose the condition \eqref{eqn:crossing_condition} holds. Let $\mathcal{I}_\beta$ be the set of optimal uniform prices when we observe the market $x^*$ but use the pricing strategy corresponding to the privacy parameter $\beta$, i.e., 
\begin{equation*}
\mathcal{I}_\beta := \argmax_{k} \Big[ y^*(k)v_k + \frac{\beta}{K(1-\beta)} (K+1-k)k \Big].  
\end{equation*}
In particular, $\mathcal{I}_0$ represents the set of optimal uniform prices in the non-private case and the set $\mathcal{I}_0$ includes those $i$'s that maximize $(K+1-i)i$. 

Now, we only need to consider the case $\mathcal{I}_\beta \subseteq \mathcal{I}_0$, as otherwise the segmentation corresponding to that value that is in $\mathcal{I}_\beta$ but not in $\mathcal{I}_0$ would be below the line $TR$. Also, the condition \eqref{eqn:crossing_condition} implies that the sets $\mathcal{I}_0$ and $\mathcal{I}_1$ are disjoint, and hence, the sets $\mathcal{I}_\beta$ and $\mathcal{I}_1$ are also disjoint.

We define the market $\tilde{x}$ in the following way:
\begin{equation}
\sum_{k=i}^K \tilde{x}(v_i) = 
\begin{cases}
\frac{Z}{v_i} & \text{ if } i \notin \mathcal{I}_1, \\   
\frac{Z}{v_i} - \epsilon & \text{ if } i \in \mathcal{I}_1,
\end{cases}
\end{equation}
where $Z$ is chosen such that $\tilde{x}$ becomes a distribution. It is clear that such $\tilde{x}$ exists for sufficiently small $\epsilon$ given that $(v_i)_{i=1}^K$ is a strictly ascending sequence. 

Notice that, given the way that $\tilde{x}$ is defined, $\mathcal{U}_p(v_i, \tilde{x})$ is equal to its maximum if and only if $i \notin \mathcal{I}_1$. In other words, a price $v_i$ maximizes the producer utility for this market if and only if $i \notin \mathcal{I}_1$. 
However, in the definition of $\mathcal{S}'$, we use the pricing strategy corresponding to the private case, i.e., we look at the sets $\mathcal{X}_k^\beta$'s. Now, we claim that, for sufficiently small $\epsilon$, $\tilde{x} \in \mathcal{X}_i^\beta$ for some $i \in \mathcal{I}_1$. In other words, for sufficiently small $\epsilon$, we end up choosing a price $v_i$ with $i \in \mathcal{I}_1$ for the market $\tilde{x}$, although it would lead to a lower utility compared to any price $v_j$ with $j \notin \mathcal{I}_1$. 

To show this claim holds, it suffices to establish the following claim and then use the continuity of utilities in $\epsilon$. 
\begin{claim}
For $\epsilon = 0$, any price $v_i$ with  $i \in \mathcal{I}_1$ is strictly preferred over any price $v_j$ with $j \notin \mathcal{I}_1$ if we use the pricing strategy corresponding to parameter $\beta > 0$.  
\end{claim}
\begin{proof}
To prove this, we need to show that
\begin{equation*}
(1-\beta) \sum_{k=i}^K \tilde{x}(v_k) v_i + \beta \frac{K+1-i}{K} v_i 
> 
(1-\beta) \sum_{k=j}^K \tilde{x}(v_k) v_j + \beta \frac{K+1-j}{K} v_j 
\end{equation*}
Note that the first term on the left and right hand side are equal when $\epsilon=0$. Now the second term on the left hand side is bigger than the second term on the right hand side given that $i \in \mathcal{I}_1$ and $j \notin \mathcal{I}_1$.
\end{proof}
 
Now, consider a segmentation with two segments: the first segment's market is $\tilde{x}$ and it consists $\delta$ fraction of consumers and the second segment's market is $\frac{x^* - \delta \tilde{x}}{1-\delta}$ which consists the remaining $1-\delta$ fraction of consumers. For sufficiently small $\delta$, the optimal price for the second segment would be sill in $\mathcal{I}_\beta$ given that the market is perturbed minimally. Let us denote that optimal price as $v_{j^*}$ with $j^* \in \mathcal{I}_\beta$. 

This implies $j^* \notin \mathcal{I}_1$. Therefore, as we elaborated above, this would mean that the price picked for the first segment would lead to a lower utility compared to the case that we would have picked $v_{j^*}$. Given that we pick $v_{j^*}$ for the second segment, the overall producer utility corresponding to this segmentation would be strictly lower than $\mathcal{U}_p(v_{j^*}, x^*)$.

Recall that, as we discussed above, we can assume $\mathcal{I}_\beta \subseteq \mathcal{I}_0$ and hence $j^* \in \mathcal{I}_0$ as well. Thus, $\mathcal{U}_p(v_{j^*}, x^*)$ is in fact the utility corresponding to the line $TR$. Hence, the segmentation above corresponds to a point in the set $\mathcal{S}'$ which falls below the line $TR$. This completes the proof. 
\subsection{Proof of Fact \ref{fact:distance_y_axis}}
Notice that $\beta > \bar{\beta}_K$ implies that no segment's price is going to be set equal to $v_K$. Hence, the consumers with value $v_k$ will face a a price upper bounded by $v_{K-1}$. As a result, the consumer utility is lower bounded by 
\begin{equation*}
x^*(v_K) (v_K - v_{K-1}).    
\end{equation*}
\subsection{Proof of \cref{proposition:consumer_max}}
For any $\beta$, let $\rho(\beta)$ denote the maximum consumer utility in the set $\mathcal{S}'$. Using Theorem \ref{theorem:S'_polytope}, we have
\begin{subequations} \label{eqn:max_consumer_LP}
\begin{align} 
\rho(\beta) &= \max_{\bm{z}} \sum_{i=1}^{K-1} \sum_{j=i+1}^{K} (v_{j}-v_{j-1})z(i,j) \label{eqn:max_consumer_LP_a} \\   
\text{s.t.} \quad  & z(i,1) \geq \cdots \geq z(i,K) \geq 0 \text{ for all } i \in [K], \label{eqn:max_consumer_LP_b} \\
& z(i,i) v_i - z(i,j) v_j \geq \frac{\beta}{K(1-\beta)} z(i,1)\Big ( (K+1-j) v_j - (K+1-i) v_i \Big ) \text{ for all } i,j \in [K], \label{eqn:max_consumer_LP_c} \\
& \sum_{i=1}^K z(i,j) = \sum_{k=j}^K x^*(v_k) \text{ for all } j \in [K]. \label{eqn:max_consumer_LP_d}
\end{align}
\end{subequations} 
Notice that, as long as $\beta < \min_{k} \bar{\beta}_k$, all the constraints \eqref{eqn:max_consumer_LP_c} are strict when the vector $\bm{z}$ is set equal to the segmentation in which all consumers with the same value go into one segment. Therefore, we can use Corollary 5 from \cite{milgrom2002envelope} to deduce that  the function $\rho(\cdot)$ is absolutely continuous (and hence differentiable almost everywhere). Therefore, we have
\begin{equation} \label{eqn:integral_rho}
\rho(\beta) = \rho(0) + \int_{0}^\beta \rho'(\zeta) d \zeta,    
\end{equation}
where $\rho'(\cdot)$ denotes the derivative of $\rho(\cdot)$ (when it exists). 
Next, we make the following claim:
\begin{claim} \label{claim:derivative_rho}
Suppose $\rho(\cdot)$ is differentiable at some $\beta$. Then, condition \eqref{eqn:increasing_uniform} implies $\rho'(\beta) \leq 0$ and condition \eqref{eqn:decreasing_uniform} implies $\rho'(\beta) \geq 0$.     
\end{claim}
\begin{proof}
We use the Envelope theorem to characterize the derivative of $\rho(\cdot)$. Let $\lambda_{i,j}$ denote the Lagrangian multiplier corresponding to the condition \eqref{eqn:max_consumer_LP_c}. Then, we have
\begin{equation} \label{eqn:envelope_1}
\rho'(\beta) = -\frac{d}{d \beta} \frac{\beta}{1-\beta}
\Big[ \sum_{i,j} \lambda_{i,j} \Big( (K+1-j)v_j - (K+1-i)v_i \Big) z^*_\beta(i,1) 
\Big],
\end{equation}
where $\bm{z}^*_\beta$ is the solution of \eqref{eqn:max_consumer_LP}. 
Next, using the Karush–Kuhn–Tucker conditions \cite{bertsekas1997nonlinear}, we know that (i) $\lambda_{i,j} \geq 0$, and (ii) $\lambda_{i,j} >0$ if and only if the condition \eqref{eqn:max_consumer_LP_c} is active, i.e.,
\begin{equation} \label{eqn:active_constraint}
z^*_\beta(i,i) v_i - z^*_\beta(i,j) v_j = \frac{\beta}{K(1-\beta)} z^*_\beta(i,1)\Big ( (K+1-j) v_j - (K+1-i) v_i \Big ).   
\end{equation}
Next, recall from the proof of Theorem \ref{theorem:S'_polytope} that $z^*_\beta(i, \cdot)$ corresponds to a market that is priced $v_i$. Now, if \eqref{eqn:active_constraint} holds for some $i$ and $j$, it means that producer has been indifferent between $v_i$ or $v_j$ as the market price, but $v_i$ has been chosen. Since $\bm{z}^*_\beta$ is the solution to the consumer maximization problem, this would mean that $v_i < v_j$, or equivalently, $i <j$. Therefore, $\lambda_{i,j}$ is either positive or zero, and we have $\lambda_{i,j} >0$ only if $j>i$.

Notice that condition \eqref{eqn:increasing_uniform} implies $(K+1-j)v_j \geq (K+1-i)v_i$ when $j \geq i$. Using this result, along with \eqref{eqn:envelope_1} and the fact that $\beta/(1-\beta)$ is an increasing function of $\beta$,  we obtain $\rho'(\beta) \leq 0$. Similarly, we can see that, under condition \eqref{eqn:decreasing_uniform}, we have $\rho'(\beta) \geq 0$. This completes the proof of the claim.
\end{proof}

Let us turn our focus to the maximum consumer utility. As Proposition \ref{proposition:S_to_S'_K} establishes, for a given $\beta$, the maximum consumer utility is given by
\begin{equation}
\beta \Big( \sum_{k=1}^K \mathbb{P}(\mathcal{X}_k^\beta) ~\mathcal{U}_c(v_k, x^*) \Big) + (1-\beta) \rho(\beta). 
\end{equation}
In particular, for the non-private case, this is equal to $\rho(0)$. Claim \ref{claim:derivative_rho} along with \eqref{eqn:integral_rho} implies that, under condition \eqref{eqn:increasing_uniform} we have $\rho(\beta) \leq \rho(0)$ and under condition \eqref{eqn:decreasing_uniform} we have $\rho(\beta) \geq \rho(0)$. Therefore, to establish \cref{proposition:consumer_max}, it suffices to compare $\sum_{k=1}^K \mathbb{P}(\mathcal{X}_k^\beta) ~\mathcal{U}_c(v_k, x^*)$ with $\rho(0)$. The following claim does this, and therefore, completes the proof. 
\begin{claim} \label{claim:comparison}
\begin{enumerate}[(i)]
\item Assume $v_1$ is the optimal uniform price in the non-private case. Then, we have
\begin{equation*}
\rho(0) \geq \sum_{k=1}^K \mathbb{P}(\mathcal{X}_k^\beta) ~\mathcal{U}_c(v_k, x^*).
\end{equation*}
\item Assume $v_K$ is the optimal uniform price in the non-private case. Then, there exists $\tilde{\alpha}$ and $\tilde{\beta} > 0$ such that, if $x^*(v_K) \geq \tilde{\alpha}$ and $\beta \leq \tilde{\beta}$, then we have 
\begin{equation*}
\rho(0) \leq \sum_{k=1}^K \mathbb{P}(\mathcal{X}_k^\beta) ~\mathcal{U}_c(v_k, x^*).
\end{equation*}
\end{enumerate} 
\end{claim}
\begin{proof}
Let us denote the maximum producer utility (corresponding to the first-degree price discrimination) by $M$. We know that the maximum consumer utility in the non-private case lies on the line that sets the sum of consumer and producer utilities equal to $M$. Therefore, when $v_1$ is the optimal uniform price, we have $\rho(0) = M-v_1$. Also, notice that for any value $v_k$, we have
\begin{equation}
\mathcal{U}_c(v_k, x^*) = \sum_{i=k}^K x^*(v_i) (v_i - v_k) 
\leq \sum_{i=k}^K x^*(v_i) (v_i - v_1)
\leq \sum_{i=1}^K x^*(v_i) (v_i - v_1) = M - v_1.
\end{equation}
Thus, and given that $\sum_{k=1}^K \mathbb{P}(\mathcal{X}_k^\beta) = 1$, we deduce that
\begin{equation*}
\sum_{k=1}^K \mathbb{P}(\mathcal{X}_k^\beta) ~\mathcal{U}_c(v_k, x^*) \leq M - v_1 
\end{equation*}
which completes the proof of part (i). Next we prove part (ii). Given that, in this part, $v_K$ is the optimal uniform price in the non-private case, we have
\begin{equation} \label{eqn:consumer_max_ii_1}
\rho(0) = M - x^*(v_K) v_K \leq (1-x^*(v_K)) v_{K-1}.    
\end{equation}
Also, note that, for any $k \leq K-1$ we have
\begin{equation}
\mathcal{U}_c(v_k, x^*) \geq x^*(v_K) (v_K - v_k),    
\end{equation}
which implies
\begin{align} \label{eqn:consumer_max_ii_2}
\sum_{k=1}^K \mathbb{P}(\mathcal{X}_k^\beta) ~\mathcal{U}_c(v_k, x^*) 
\geq x^*(v_K) \Big[ \sum_{k=1}^{K-1} \mathbb{P}(\mathcal{X}_k^\beta)  (v_K - v_k) \Big].
\end{align}
The term $ \sum_{k=1}^{K-1} \mathbb{P}(\mathcal{X}_k^\beta)  (v_K - v_k)$ is positive for $\beta=0$, and since it is continuous in $\beta$, we can choose $\tilde{\beta} > 0$ such that 
\begin{equation}
\inf_{\beta \leq \tilde{\beta}} \Big[ \sum_{k=1}^{K-1} \mathbb{P}(\mathcal{X}_k^\beta)  (v_K - v_k) \Big]  \geq \epsilon,     
\end{equation}
for some $\epsilon > 0$. Therefore, for any $\beta \leq \tilde{\beta}$, we have
\begin{align} \label{eqn:consumer_max_ii_3}
\sum_{k=1}^K \mathbb{P}(\mathcal{X}_k^\beta) ~\mathcal{U}_c(v_k, x^*) 
\geq \epsilon~ x^*(v_K).
\end{align}

Now, notice that the right-hand side of \eqref{eqn:consumer_max_ii_1} goes to zero as $x^*(v_K)$ increases to one. On the other hand, the right-hand side of \eqref{eqn:consumer_max_ii_3} is an increasing and positive function of $x^*(v_K)$. Therefore, there exists $\tilde{\alpha}$ such that, if $x^*(v_K) \geq \tilde{\alpha}$, we have
\begin{equation}
\sum_{k=1}^K \mathbb{P}(\mathcal{X}_k^\beta) ~\mathcal{U}_c(v_k, x^*) \geq \rho(0),    
\end{equation}
for any $\beta \leq \tilde{\beta}$. This completes the proof. 
\end{proof}
\end{document}